\documentclass[10pt,journal,onecolumn,final]{IEEEtran}

\usepackage[T1]{fontenc}

\usepackage{booktabs} 

\usepackage[a4paper]{geometry}
\usepackage{blkarray}
\usepackage{multirow}

\usepackage{cite}
\usepackage{amsmath,amssymb,amsthm,amsfonts}
\usepackage{mathtools}
\usepackage{caption,subcaption}
\usepackage{graphicx}
\graphicspath{{./}}
\usepackage{textcomp}
\usepackage{xcolor}
\usepackage{bm}

\usepackage[nameinlink]{cleveref}
\def\BibTeX{{\rm B\kern-.05em{\sc i\kern-.025em b}\kern-.08em
		T\kern-.1667em\lower.7ex\hbox{E}\kern-.125emX}}
\newtheorem{thm}{Theorem}
\newtheorem{lem}[thm]{Lemma}
\newtheorem{prop}[thm]{Proposition}

\DeclarePairedDelimiterX{\norm}[1]{\lVert}{\rVert}{#1}
\newtheorem{defn}{Definition}

\newtheorem{rem}{Remark}
\newtheorem{ass}{Assumption}
\usepackage{algorithm}
\usepackage{algpseudocode}

 \newtheoremstyle{problemstyle}
   {\topsep} 
   {\topsep} 
   {\itshape} 
   {}         
   {\bfseries} 
   {.}        
   {.5em}     
   {}         
   
\theoremstyle{problemstyle} 
\newtheorem{problem}{Problem}

\usepackage[cmintegrals]{newtxmath}
	\usepackage[normalem]{ulem}
	\usepackage{placeins}
	\usepackage[displaymath, mathlines]{lineno}
	
\begin{document}

		\title{Local Information-Theoretic Security via Euclidean Geometry\thanks{This work was presented in part at the 32nd European Signal Processing Conference (EUSIPCO), Lyon, France, 2024 \cite{AthanasakosEUSIPCO24}.}}
		%
		%
		\author{\IEEEauthorblockN{
				Emmanouil M.~Athanasakos, Nicholas~Kalouptsidis, and Hariprasad Manjunath}
			
			\thanks{E.M. Athanasakos and N.Kalouptsidis is with the Department of Informatics and Telecommunications, National and Kapodistrian University of Athens, Greece, e-mail: emathan@di.uoa.gr,kalou@di.uoa.gr.}%
			\thanks{H. Manjunath is with the Chanakya University, Bengaluru, India, e-mail: mhariprasadkansur@gmail.com}
		}
		%
		%

	\markboth{PREPRINT}
	{PREPRINT}
	%

	
	\maketitle
\begin{abstract}
	This paper introduces a methodology based on Euclidean information theory to investigate local properties of secure communication over discrete memoryless wiretap channels. We formulate a constrained optimization problem that maximizes a legitimate user's information rate while imposing explicit upper bounds on both the information leakage to an eavesdropper and the informational cost of encoding the secret message. By leveraging local geometric approximations, this inherently non-convex problem is transformed into a tractable quadratic programming structure. It is demonstrated that the optimal Lagrange multipliers governing this approximated problem can be found by solving a linear program. The constraints of this linear program are derived from Karush-Kuhn-Tucker conditions and are expressed in terms of the generalized eigenvalues of channel-derived matrices. This framework facilitates the derivation of an analytical formula for an approximate local secrecy capacity. Furthermore, we define and analyze a new class of secret local contraction coefficients. These coefficients, characterized as the largest generalized eigenvalues of a matrix pencil, quantify the maximum achievable ratio of approximate utility to approximate leakage, thus measuring the intrinsic local leakage efficiency of the channel. We establish bounds connecting these local coefficients to their global counterparts defined over true mutual information measures. The efficacy of the proposed framework is demonstrated through detailed analysis and numerical illustrations for both general multi-mode channels and the canonical binary symmetric wiretap channel.
\end{abstract}
	
	

\section{Introduction}\label{sec_I}
	\IEEEPARstart{T}{he} fundamental challenge of ensuring confidential data transmission in the presence of adversarial eavesdroppers has been a cornerstone of information theory since Shannon's foundational work on perfect secrecy \cite{Shannon} and Wyner's subsequent characterization of the wiretap channel \cite{Wyner75}. These seminal contributions by Wyner, and Csiszár and Körner \cite{Csiz_Kor}, established the theoretical limits of secure communication. A significant body of research has since focused on designing practical coding schemes to approach these limits, with notable progress in leveraging capacity-approaching codes such as LDPC \cite{ldpc_1}, polar codes \cite{ath_polar}, and turbo codes \cite{kolok_1}, as well as specialized constructions for specific channel models like the Gaussian wiretap channel \cite{ath_Sparcs,athanasakos2025sparse}. While this line of work is crucial for implementation, classical information-theoretic approaches often focus on asymptotic performance or involve complex optimizations that may not fully capture the operational nuances of contemporary systems, particularly those dealing with finite resources, constraints on encoding complexity, or the need for efficient and secure transmission of relatively small data units. In this work, we propose and develop a framework for what we term \textit{Local Information-Theoretic Security}. In contrast to classical approaches that characterize asymptotic performance, our focus is on the local behavior of secrecy metrics within a neighborhood of a fixed operating point. Specifically, we study the performance trade-offs under small perturbations of a reference input distribution. This local analysis allows us to move beyond numerical black-box solutions and instead derive analytical expressions for performance sensitivities, revealing the underlying structure of secrecy in a regime relevant to finite-resource communication systems.
		
	This paper confronts this challenge by systematically employing Euclidean Information Theory (EIT) \cite{Borade_1, Zheng_1, Zheng_2} to investigate local phenomena in secure communications over wiretap channels. EIT, rooted in the differential geometry of statistical manifolds, provides a rigorous methodology for developing local approximations of information-theoretic measures. By approximating quantities such as Kullback-Leibler (KL) divergence and mutual information with quadratic forms of perturbation vectors defined around a reference operating point, EIT facilitates the transformation of intractable non-linear problems into more structured algebraic domains, often yielding problems amenable to linear algebra or quadratic programming techniques and thereby enabling new analytical perspectives. Furthermore, this approach not only yields computable approximations but, more importantly, reveals the fundamental connection between local secrecy performance and the spectral theory of the underlying channel matrices.
		
	The impetus for exploring such local analytical tools is also driven by the characteristics of emerging communication paradigms. Scenarios involving the transmission of limited information, such as control signals in cyber-physical systems, state updates in IoT networks, critical alert messages, or elements of cryptographic protocols, demand not only security but also high efficiency in terms of encoding resources and latency~\cite{Durisi_20, FengPoor_2022}. For these systems, understanding the achievable secure performance under specific, possibly restrictive, operational constraints is crucial. Furthermore, there is a theoretical need for metrics beyond asymptotic capacity that can quantify the intrinsic efficiency of the secrecy mechanisms themselves, for instance, by relating the utility gained by the legitimate recipient to the information leaked or the resources used in the encoding process.
		
	Our research addresses these needs by developing a comprehensive EIT-based framework centered on a novel constrained optimization problem, termed as the Secret Information Coupling (SIC) problem. The SIC problem is formulated to maximize the legitimate user's information rate while imposing explicit upper bounds on both the information leakage to an eavesdropper and on the encoding rate. This latter constraint is particularly significant as it models the efficient use of channel resources for transmitting limited information and aligns naturally with the local perturbation philosophy inherent in EIT. A preliminary version of this work, which introduced the local approximation of secrecy capacity using some of these concepts, was presented in \cite{AthanasakosEUSIPCO24}. The SIC problem, with its interplay of three distinct mutual information terms under various structural and operational constraints, presents a rich information-theoretic optimization challenge.	

	The SIC problem formulation shares conceptual parallels with the Information Bottleneck (IB) principle \cite{TishbyPereiraBialek00} and the broader class of Privacy-Utility Tradeoff (PUT) problems \cite{MakhdoumiSalamatianFawazMedard14}. These frameworks also contend with balancing information utility against some form of information restriction. For instance, the authors in \cite{SreekumarGunduz19} analyze an optimal PUT problem under an $I(X;U) \le R$ rate constraint using exact mutual information measures, establishing conditions for perfect privacy and characterizing the utility-leakage slope for infinitesimal leakage.  In a related vein, Zamani et al. \cite{Zamani_eit} explore data disclosure mechanisms with non-zero leakage guarantees under a strong $\chi^2$-privacy criterion, which is particularly relevant as the $\chi^2$-divergence is closely connected to the quadratic forms used in EIT. Other works have explored PUT in various practical contexts, such as linear regression under noise \cite{showkatbakhsh2019privacy} or for long-term conversation analysis \cite{pohlhausen2024long}, and have sought tight bounds in differential privacy settings \cite{geng2020tight,kopf11}.  While EIT has been applied to various other privacy contexts \cite{Tan_1, Tan_2, Zhou}, our approach is distinguished by its systematic application of local geometric approximations, which allows us to derive new analytical expressions for local capacity and to introduce novel structural coefficients.
		
	The concept of contraction coefficients, which describe how information is preserved or dissipated through transformations, is central to this work and is well-studied in information theory \cite{Poliyanski,Polyanskiy_notes,Polyanskiy_16, Anantharam_13} and within EIT for standard channels \cite{Zheng_2, Makur_1, Makur_20}. Recent work has also established linear bounds for these coefficients for general $f$-divergences \cite{makur2018linear}. Our work extends this line of inquiry by identifying a relationship between EIT-approximated secrecy and notions akin to hypercontractivity. This direction is related to efforts that use stronger, alternative measures of statistical dependence for privacy and secrecy. For example, Asoodeh et al. employ the Hirschfeld-Gebelein-R\'{e}nyi maximal correlation \cite{Renyi, Anantharam_12, Anathram_13_2} to analyze information extraction under privacy constraints \cite{asoodeh2016information, Asoodeh}. Building on this, Li and El Gamal define and analyze maximal correlation secrecy for discrete memoryless channels \cite{Li_Gamal, li2017maximal}. Further concepts such as maximal leakage \cite{Khisti19} have been introduced as robust privacy metrics, and the asymptotics of covert communication have also been analyzed using maximal correlation \cite{tahmasbi2018information}. A core element of our contribution is the development of a specialized secret local contraction coefficient specifically tailored to and derived from the EIT secrecy framework, providing a new tool in this landscape of advanced secrecy metrics. The utility of this framework is demonstrated on the binary symmetric wiretap channel (BSWC).
		
	The primary theoretical contributions of this paper are therefore as follows. We first establish the EIT-approximated SIC problem, transforming the original non-convex formulation into a quadratic programming structure. A key result of this work is the demonstration that the optimal Lagrange multipliers for this approximated problem, which in turn determine the capacity, can be found by solving a general Linear Program (LP). The constraints of this LP are rigorously derived from the Karush-Kuhn-Tucker (KKT) conditions and are expressed in terms of the generalized eigenvalues of the channel matrices. This LP framework facilitates the derivation of an explicit analytical formula for an \textit{approximate local secrecy capacity}. The LP framework allows us to fully characterize the behavior of the local secrecy capacity across distinct operational regimes. This provides clear design principles by showing how the optimal signaling strategy shifts depending on whether the system is constrained by encoding resources or by the secrecy requirement. We also demonstrate that SIC is structurally equivalent to the IB with side information problem, also called PUT. Consequently the approximations and methods developed for SIC directly apply to PUT and IB.
		
	Secondly, we introduce and analyze a new information-theoretic quantity termed as the \textit{secret local contraction coefficient}. We characterize it as the largest generalized eigenvalue of a channel-derived matrix pencil and show that it quantifies the maximum achievable local efficiency, defined as the ratio of the approximated utility to leakage. Furthermore, we establish upper and lower bounds connecting this EIT-derived local coefficient to its global counterpart defined over exact mutual information measures.
	The behavior of the local secrecy capacity and these coefficients is thoroughly investigated. This includes an analysis of various special conditions. To validate and illustrate our framework, we provide comprehensive numerical results. The BSWC is employed as a canonical example to compare the local secrecy capacity with the true one and demonstrate the different operational regimes. Furthermore, we validate the general LP formulation for a multi-mode channel with a larger alphabet size, confirming the solution structure by comparing the output of a standard LP solver with an exhaustive search of the feasible vertices.
		
	The remainder of this paper is organized as follows. Section \ref{sec_2} formally introduces the problem and its information-theoretic context. Section \ref{sec:eit_preliminaries} reviews the essential EIT concepts and approximations that underpin our analysis. In Section \ref{sec:local_approx_slic}, we apply EIT to transform the SIC problem into its analytically tractable quadratic form. Section \ref{sec:approx_capacity_solution_structure} presents the core theoretical results, where we analyze this approximated problem to derive optimality conditions and the approximate local secrecy capacity. Section \ref{sec:numerical_illustations} provides comprehensive numerical illustrations, including a validation of the EIT framework and a detailed case study on the BSWC. The secret local contraction coefficients are formally defined and analyzed in Section \ref{sec:secret_contraction_coeffs}. Finally, Section \ref{sec:conclusion} concludes the paper and discusses future research. Detailed mathematical proofs are provided in the appendices.

      
        \section{ THE SECURE INFORMATION COUPLING PROBLEM}\label{sec_2}
        This section lays the information-theoretic foundation for our analysis. We first briefly revisit the classical discrete memoryless wiretap channel (DM-WTC) and the definition of secrecy capacity. Subsequently, we introduce the specific constrained optimization problem, which forms the core of our investigation. 
        \subsection{The Discrete Memoryless Wiretap Channel and Secrecy Capacity}
The transmitter (Alice) wishes to send reliably a confidential message to a legitimate receiver (Bob) through a noisy DMC while avoiding information leakage to an eavesdropper (Eve). It was shown in the seminal works of \cite{Csiz_Kor, Wyner75} that rates below the secrecy capacity of the channel are achievable; that is, there exists a sufficiently long code (encoder/decoder pair) that can keep the error probability at the receiver as well as the information leakage at the eavesdropper below any prespecified positive level. Moreover, for a general DM-WTC, the secrecy capacity is given by the single-letter expression \cite{Blo_Bar}:
\begin{equation} \label{eq:cs_general}
	C_{\mathrm{S}} = \max_{\{P_U(u)\}, \{P_{X|U}(x|u)\}} [I(U;Y) - I(U;Z)]
\end{equation}
where the maximization is over all PMFs $\{P_U(u)\}$ and $\{P_{X|U}(x|u)\}$ such that the Markov chain $U \to X \to (Y,Z)$ holds. The auxiliary variable $ U $ allows for sophisticated stochastic encoding strategies essential for achieving capacity in general wiretap channels. A common special case is the physically degraded wiretap channel, where $ X \to Y \to Z $ forms a Markov chain. In this scenario, the auxiliary variable $ U $ can be identified with $ X $, and the secrecy capacity simplifies to $C_{\mathrm{S}} =  \max_{\{P_X(x)\}} [I(X;Y) - I(X;Z)] $.

\subsection{The Secret Information Coupling Problem} \label{subsec:slic_problem}
The above secrecy capacity result is asymptotic with respect to the code blocklength. In this work, we investigate secure communications from a perspective that stresses the efficient transmission of limited information under explicit operational constraints. Assuming an optimal deterministic decoding rule, such as maximum a posteriori decoding, our emphasis focuses on the design of stochastic encoders that achieve reliable transmission subject to rate and leakage constraints. This leads to the formulation of the SIC problem, which is particularly suited for analysis using the local approximation techniques of the EIT framework. A central result of our analysis is that the solution structure of the approximated problem is governed by an LP, revealing the linear nature of the local trade-offs. 
Let $ U $ be a discrete random variable representing the confidential message Alice intends to transmit, with PMF $ P_U(u) $. Alice encodes $ U $ into a channel input $ X \in \mathcal{X} $ according to a conditional PMF $ P_{X|U}(x|u) $, resulting in the Markov chain $ U \to X \to (Y,Z) $. We consider the following optimization problem:

\begin{problem} \label{prob:slic_original}
	Given a reference input distribution $ P_X(x) \in \mathrm{relint}(\mathcal{P(X)}) $ and the channel $P_{Y,Z|X}$, the objective is to find PMFs $ P_U(u) $ and $ P_{X|U}(x|u) $ that:
	\begin{align}
		\underset{ \{P_U(u)\}, \{P_{X|U}(x|u)\} }{\text{maximize}} \quad & I(U;Y) \label{eq:slic_prob_obj} \\
		\text{subject to} \quad & I(U;X) \le R \label{eq:slic_prob_rate} \\
		& I(U;Z) \le \Theta \label{eq:slic_prob_leakage} \\
		& \sum_{u \in \mathcal{U}} P_U(u) P_{X|U}(x|u) = P_X(x), \quad \forall x \in \mathcal{X} \label{eq:slic_prob_consistency} \\
		& \sum_{x \in \mathcal{X}} P_{X|U}(x|u) = 1, \quad \forall u \in \mathcal{U} \label{eq:slic_prob_pxu_norm} \\
		& P_{X|U}(x|u) \ge 0, \quad \forall x \in \mathcal{X}, u \in \mathcal{U} \label{eq:slic_prob_pxu_nonneg} \\
		& \sum_{u \in \mathcal{U}} P_U(u) = 1 \label{eq:slic_prob_pu_norm} \\
		& P_U(u) \ge 0, \quad \forall u \in \mathcal{U} \label{eq:slic_prob_pu_nonneg}
	\end{align}
	where $ R > 0 $ and $ \Theta > 0 $ are predefined thresholds. The leakage constraint \eqref{eq:slic_prob_leakage} on $ I(U;Z) $ should be interpreted in the spirit of a non-asymptotic total leakage bound for the small amount of information represented by $ U $.
\end{problem}
The SIC problem, as formulated above, seeks to optimize the communication strategy defined by the PMFs $ P_U(u) $ and $ P_{X|U}(x|u) $. The primary objective, given by \eqref{eq:slic_prob_obj}, is to maximize the utility $ I(U;Y) $, which quantifies the rate of reliable information successfully conveyed from the secret message $U$ to the legitimate recipient. This maximization is performed under a set of carefully chosen operational constraints. The encoding rate constraint \eqref{eq:slic_prob_rate} limits the mutual information between the message and the channel input, which can be interpreted as a budget on the complexity or bandwidth used to modulate $U$ onto the channel input $X$. In the context of EIT, a small $R$ implies that the conditional distributions $P_{X|U=u}(x)$ remain in a local neighborhood of the reference distribution $P_X(x)$.
Concurrently, the leakage limit constraint \eqref{eq:slic_prob_leakage} directly bounds the total information about the message that the eavesdropper can obtain from her observation $Z$. A smaller threshold $ \Theta $ signifies a more stringent secrecy requirement, approaching the ideal of perfect secrecy and analogous to the total leakage parameter of the strong secrecy \cite{Maurer}. The framework is grounded by the input distribution consistency constraint \eqref{eq:slic_prob_consistency}, which ensures that the input distribution, when averaged over all messages, conforms to a predefined reference $ P_X(x) $. This reference serves as the operating point around which local approximations will be developed. Finally, constraints \eqref{eq:slic_prob_pxu_norm} through \eqref{eq:slic_prob_pu_nonneg} are the standard probability properties, ensuring that $ P_{X|U}(x|u) $ and $ P_U(u) $ are valid conditional and marginal PMFs, respectively. While the general problem of maximizing $ I(U;Y) $ subject to information-theoretic constraints is often non-convex and computationally challenging, the EIT framework allows for the approximation of these mutual information terms by quadratic forms, transforming Problem \ref{prob:slic_original} into a more tractable linear algebraic structure, as will be detailed in Section IV.

\begin{rem} \label{rem:dpi_interplay}
The Markov chain $ U \to X \to Z $ implies the data processing inequality $ I(U;Z) \le I(U;X) $. Consequently, if $ R \le \Theta $, then constraint \eqref{eq:slic_prob_rate} implies $ I(U;Z) \le I(U;X) \le R \le \Theta $, making the explicit leakage constraint \eqref{eq:slic_prob_leakage} redundant. However, in practical scenarios where a non-trivial amount of information is encoded and strong secrecy is desired, it is common to encounter $ \Theta < R $. In such cases, both constraints \eqref{eq:slic_prob_rate} and \eqref{eq:slic_prob_leakage} shape the optimization landscape \cite{AthanasakosEUSIPCO24}.
\end{rem}
\subsection{Fundamental Properties and Bounds for the SIC Problem} \label{subsec:slic_properties}
Before applying the EIT approximation techniques, we establish some fundamental properties and straightforward bounds related to the original Problem \ref{prob:slic_original}. These help in understanding the inherent limitations and its behavior. Let $ I^*(R, \Theta) $ denote the optimal value of the objective function in Problem \ref{prob:slic_original} for given parameters $ R $ and $ \Theta $.
\begin{lem} \label{lem:slic_fundamental_props}
Any feasible solution $ (P_U, P_{X|U}) $ to the Problem \ref{prob:slic_original} satisfy the following properties:
\begin{enumerate}
    \item The achievable utility $ I(U;Y) $ is upper bounded by the encoding rate constraint $ R $:
    \begin{equation} \label{eq:utility_cap_by_R}
    I(U;Y) \le R.
    \end{equation}
    \item The information leakage $ I(U;Z) $ is also upper bounded by the encoding rate constraint $ R $:
    \begin{equation} \label{eq:leakage_cap_by_R}
    I(U;Z) \le R.
    \end{equation}
    \item The optimal utility $ I^*(R, \Theta) $ is:
    \begin{enumerate}
        \item A non-decreasing function of the encoding rate $ R $ (for a fixed $ \Theta $).
        \item A non-decreasing function of the leakage threshold $ \Theta $ (for a fixed $ R $).
    \end{enumerate}
\end{enumerate}
\end{lem}
\begin{proof}
The upper bounds in (1) and (2) follow directly from the Data Processing Inequality (DPI) and the problem constraints.
(3) To establish the existence and monotonicity of the optimal utility $I^*(R, \Theta)$, we first prove that the maximum is guaranteed to exist. The mutual information objective function, $I(U;Y)$, is a continuous function of its arguments, the PMFs $P_U(u)$ and $P_{X|U}(x|u)$. The feasible set of these PMFs, defined by the problem constraints, is a closed and bounded subset of a finite-dimensional Euclidean space, and is therefore a compact set. By the Weierstrass Extreme Value Theorem, a continuous function on a compact set always attains its maximum. Thus, the optimal value $I^*(R, \Theta)$ is well-defined and always exists, provided the feasible set is non-empty.

\quad (a) Consider two encoding rates $R_1 \le R_2$. Any pair of distributions $(P_U, P_{X|U})$ feasible for $R_1$ is also feasible for $R_2$, since $I(U;X) \le R_1 \le R_2$. Thus, the set of feasible solutions for $(R_2, \Theta)$ is a superset of the feasible solutions for $(R_1, \Theta)$. The maximization of the continuous function $I(U;Y)$ over a larger compact set cannot result in a smaller optimal value, hence $I^*(R_1, \Theta) \le I^*(R_2, \Theta)$.

\quad (b) A similar argument applies for $\Theta$. If $\Theta_1 \le \Theta_2$, any solution feasible under the constraint $I(U;Z) \le \Theta_1$ is also feasible under $I(U;Z) \le \Theta_2$. Thus, the feasible set for $(R, \Theta_2)$ contains the set for $(R, \Theta_1)$, leading to $I^*(R, \Theta_1) \le I^*(R, \Theta_2)$.
\end{proof}

This lemma highlights the fact that the utility is fundamentally limited by how much information about $ U $ is encoded into $ X $, which also indirectly bounds the leakage. Furthermore, allowing a higher encoding rate or more leakage cannot decrease the maximum achievable utility. These basic relationships define the operational boundaries of the SIC problem before any approximation techniques are employed.

\subsection{Analogy to the information bottleneck with side information and the privacy-utility trade-off} \label{subsec:put_analogy}

The SIC Problem is closely related to the PUT framework~\cite{MakhdoumiSalamatianFawazMedard14} and the IB with side information \cite{ChechikTishby_NIPS2002}. This connection provides valuable context and highlights inherent algorithmic challenges.
\begin{defn} \label{def:put_general}
In a canonical PUT problem, an agent wishes to release a representation $ U_{\mathrm{rep}} $ of an observed variable $ X_{\mathrm{obs}} $. $ X_{\mathrm{obs}} $ is correlated with a utility-relevant variable $ Y_{\mathrm{util}} $ and a private (sensitive) variable $ S_{\mathrm{priv}} $. The goal is to design the mapping $ P_{U_{\mathrm{rep}}|X_{\mathrm{obs}}}(u_{\mathrm{rep}}|x_{\mathrm{obs}}) $ (and possibly $ P_{U_{\mathrm{rep}}}(u_{\mathrm{rep}}) $) to maximize a utility metric, typically $ I(U_{\mathrm{rep}}; Y_{\mathrm{util}}) $, subject to a privacy constraint, $ I(U_{\mathrm{rep}}; S_{\mathrm{priv}}) \le \Theta_S $, and often a compression constraint on the representation, $ I(U_{\mathrm{rep}}; X_{\mathrm{obs}}) \le R_X $. The underlying Markov structure is $ (Y_{\mathrm{util}}, S_{\mathrm{priv}}) \to X_{\mathrm{obs}} \to U_{\mathrm{rep}} $.
\end{defn}

The SIC Problem exhibits a strong structural analogy to the general PUT problem outlined in Definition \ref{def:put_general}. This correspondence can be seen by making the following identifications: (1) The message $ U $ in the SIC problem plays a role analogous to the representation $ U_{\mathrm{rep}} $ in the PUT framework. (2) The channel input $ X $ in SIC, generated based on message $U$, corresponds to the \text{observed variable $ X_{\mathrm{obs}} $} in PUT, from which $U_{\mathrm{rep}}$ is derived. The structural similarity is due to the property that if $X \to Y \to Z$ is Markov, $Z \to Y \to X$ is also Markov. the mutual information constraint $I(U;X) \le R$, or $I(U_{\mathrm{rep}}; X_{\mathrm{obs}}) \le R_X$, serves a similar purpose of limiting the informational content linking these two variables. (3) Bob's observation $ Y $ in SIC is the utility-relevant variable $ Y_{\mathrm{util}} $ in PUT. (4) Eve's observation $ Z $ in SIC acts as the private variable $ S_{\mathrm{priv}} $ in PUT.

Furthermore, the consistency constraint \eqref{eq:slic_prob_consistency} in the SIC problem, aligns with common settings in IB/PUT problems where the distribution of the observed variable is assumed to be given or fixed. This analogy connects the SIC problem to a broader class of information-theoretic optimization problems concerned with balancing utility, privacy/secrecy, and complexity/rate. 
The structural similarity to PUT problems implies that SIC problem generally faces similar algorithmic complexities. For a fixed $ P_U(u) $, the objective $ I(U;Y) $ is a convex function of $ P_{X|U}(x|u) $ or $ P_{UX}(u,x) $, while the constraint functions $ I(U;X) $ and $ I(U;Z) $ are also convex in $ P_{X|U}(x|u) $. Maximizing a convex function or minimizing a concave function, over a convex set, defined by convex inequality constraints, is a non-convex optimization problem and generally NP-hard \cite{Boyd}.

The inherent computational complexity of globally solving such multi-term information-theoretic optimization problems motivates the search for alternative analytical approaches and tractable approximations. While iterative techniques analogous to the Blahut-Arimoto algorithm \cite{Blahut72, Arimoto72} can find local optima for certain classes of these problems, such as the IB \cite{TishbyPereiraBialek00}, they typically do not yield closed-form insights into capacity-like terms or local system behavior, and their application to multi-constraint problems involving three or more mutual information terms can be intricate. Specific analyses for perfect privacy, such as the work in \cite{SreekumarGunduz19}, identifies conditions for positive utility under zero leakage, address important special cases but often rely on the specific structure of exact zero leakage.

Our work takes a different path by employing EIT. As will be detailed, EIT provides a framework for developing local approximations of the mutual information terms in SIC problem. This transforms the original problem into a domain characterized by quadratic forms and linear algebra, which is more amenable to analytical investigation and can yield explicit approximate expressions for performance metrics like local secrecy capacity and gain structural insights. The EIT preliminaries essential for this transformation are presented in the following section.

\section{Euclidean Information Theory Preliminaries} \label{sec:eit_preliminaries}

The SIC problem, as formulated in Problem \ref{prob:slic_original}, involves the optimization of mutual information terms which are generally non-linear and can lead to intractable problems. To address this, we employ a framework that provides local approximations for information-theoretic quantities by leveraging the geometric properties of the space of probability distributions. This section reviews the key EIT concepts and results that will be utilized in subsequent sections to analyze the SIC problem. The core idea of EIT is that for probability distributions that are close to a reference distribution, the KL divergence, and consequently mutual information, can be approximated by weighted squared Euclidean norms of perturbation vectors.

\subsection{Local Approximations of Probability Distributions and KL Divergence} \label{subsec:eit_local_approx_dist}

Let $ P_X(x) $ be a reference PMF over a finite alphabet $ \mathcal{X} $, such that $ P_X(x) > 0 $ for all $ x \in \mathcal{X} $ , i.e., $ P_X \in \mathrm{relint}(\mathcal{P(X)}) $. Consider another PMF $ Q_X(x) $ over the same alphabet that is a small perturbation of $ P_X(x) $. We can write $ Q_X(x) $ as:
\begin{equation} 
	\label{eq:eit_perturbation}
    Q_X(x) = P_X(x) + \epsilon J_X(x)
\end{equation}
where $ \epsilon \in (0,1) $ is a small positive scalar controlling the magnitude of the perturbation, and $ J_X(x) $ is a perturbation vector. For $ Q_X(x) $ to be a valid PMF, it must satisfy $ \sum_{x \in \mathcal{X}} Q_X(x) = 1 $ and $ Q_X(x) \ge 0 $ for all $ x $.
The first condition implies that the perturbation vector $ J_X(x) $ must satisfy:
\begin{equation} \label{eq:eit_J_sum_zero}
    \sum_{x \in \mathcal{X}} J_X(x) = 0.
\end{equation}
The non-negativity $ Q_X(x) \ge 0 $ requires $ P_X(x) + \epsilon J_X(x) \ge 0 $. For sufficiently small $ \epsilon $, this condition holds if $ J_X(x) $ is appropriately bounded.

In the context of conditional distributions $ P_{X|U}(x|u) $ that are perturbations of $ P_X(x) $ for different values of $ u \in \mathcal{U} $, we write:
\begin{align}
	\label{eq:conditional_perturbation}
	P_{X|U}(x|u) = P_X(x) + \epsilon J_X(x|u),
\end{align}
where each $ J_X(x|u) $ satisfies $ \sum_x J_X(x|u) = 0 $.
The consistency constraint \eqref{eq:slic_prob_consistency} from Problem \ref{prob:slic_original} implies:
\begin{equation} 
	\label{eq:eit_J_consistency}
    \sum_{u \in \mathcal{U}} P_U(u) J_X(x|u) = 0, \quad \forall x \in \mathcal{X}.
\end{equation}

It is often convenient to work with a scaled and weighted perturbation vector $ L_X(x) $ or $ L_X(x|u) $, for conditional distributions, defined as \cite{Zheng_2}:
\begin{equation} 
	\label{eq:eit_L_definition}
    L_X(x|u) = \frac{J_X(x|u)}{\sqrt{P_X(x)}}.
\end{equation}
Using \eqref{eq:eit_L_definition}, \eqref{eq:conditional_perturbation} is written as, 
\begin{equation} 
	\label{eq:eit_pxu_perturbation_sec3} 
	P_{X|U}(x|u) = P_X(x) + \epsilon \sqrt{P_X(x)} L_X(x|u),
\end{equation}
where $ \epsilon \in (0,1) $.
The condition \eqref{eq:eit_J_sum_zero} for $ J_X(x|u) $ translates to an orthogonality condition for $ L_X(x|u) $:
\begin{equation} 
	\label{eq:eit_L_ortho_sqrtP}
    \sum_{x \in \mathcal{X}} \sqrt{P_X(x)} L_X(x|u) = 0.
\end{equation}
This means $ L_X(\cdot|u) $ is orthogonal to the vector $ (\sqrt{P_X(x_1)}, \dots, \sqrt{P_X(x_{|\mathcal{X}|})})^T $.
Similarly, the consistency condition \eqref{eq:eit_J_consistency} becomes:
\begin{equation} 
	\label{eq:eit_L_consistency}
    \sum_{u \in \mathcal{U}} P_U(u) \sqrt{P_X(x)} L_X(x|u) = 0, \quad \forall x \in \mathcal{X}.
\end{equation}
Condition \eqref{eq:eit_L_ortho_sqrtP} translates to an orthogonality condition for each scaled perturbation vector $ \mathbf{L}_u = (L_X(x_1|u), \dots, L_X(x_{|\mathcal{X}|}|u))^T $, that is,
\begin{equation} 
	\label{eq:eit_L_ortho_sqrtP_final}
	\sum_{x \in \mathcal{X}} \sqrt{P_X(x)} L_X(x|u) = \mathbf{L}_u^T \mathbf{\sqrt{P_X}} = 0, \quad \forall u \in \mathcal{U},
\end{equation}
where $ \mathbf{\sqrt{P_X}} $ is the vector with components $ \sqrt{P_X(x)} $. This signifies that each $ \mathbf{L}_u $ must be orthogonal to $ \mathbf{\sqrt{P_X}} $ to represent a valid zero-sum perturbation $J_X(x|u)$. Furthermore, the consistency constraint~\eqref{eq:slic_prob_consistency} from the SIC problem, must hold. Substituting~\eqref{eq:eit_pxu_perturbation_sec3}:
\begin{align*}
	\sum_{u \in \mathcal{U}} P_U(u) [P_X(x) + \epsilon \sqrt{P_X(x)} L_X(x|u)] &= P_X(x) \\
	P_X(x) \sum_{u \in \mathcal{U}} P_U(u) + \epsilon \sqrt{P_X(x)} \sum_{u \in \mathcal{U}} P_U(u) L_X(x|u) &= P_X(x).
\end{align*}
Since $ \sum_{u \in \mathcal{U}} P_U(u) = 1 $ and $ \epsilon \neq 0 $, and $ P_X(x) > 0 $, this simplifies to:
\begin{equation} \label{eq:eit_L_consistency_componentwise_final}
	\sum_{u \in \mathcal{U}} P_U(u) L_X(x|u) = 0, \quad \forall x \in \mathcal{X}.
\end{equation}
In vector form, this means the expected perturbation vector must be the zero vector:
\begin{equation} 
	\label{eq:eit_L_consistency_vector_final}
	\sum_{u \in \mathcal{U}} P_U(u) \mathbf{L}_u = \mathbf{0}.
\end{equation}
The conditions \eqref{eq:eit_L_ortho_sqrtP_final} and \eqref{eq:eit_L_consistency_vector_final} are the key structural constraints imposed by the EIT framework on the perturbation vectors $ \mathbf{L}_u $ when transforming Problem \ref{prob:slic_original}. These ensure that $ P_{X|U}(x|u) $ are valid conditional PMFs for each $u$ and that they correctly average to the reference $ P_X(x) $, while $ \mathbf{L}_u $ captures the information deviations. These are precisely the constraints adopted in the formulation of the EIT-Approximated problem in Section \ref{sec:local_approx_slic}.

The KL divergence between a perturbed distribution $Q_X(x)$ in \eqref{eq:eit_perturbation} and the reference $P_X(x)$ is locally approximated using a second-order Taylor expansion \cite{CoverThomas06, Borade_1} as:
\begin{equation} 
	\label{eq:eit_kl_approx_J}
	D_{\mathrm{KL}}(Q_X || P_X) = \frac{\epsilon^2}{2} \sum_{x \in \mathcal{X}} \frac{J_X(x)^2}{P_X(x)} + O(\epsilon^3).
\end{equation}
Expressed in terms of the scaled perturbation vector $ \mathbf{L}_X $, this becomes:
\begin{equation} 
	\label{eq:eit_kl_approx_L}
	D_{\mathrm{KL}}(Q_X || P_X) = \frac{\epsilon^2}{2} \sum_{x \in \mathcal{X}} L_X(x)^2 + O(\epsilon^3) = \frac{\epsilon^2}{2} ||\mathbf{L}_X||^2 + O(\epsilon^3),
\end{equation}
where $ L_X(x) = J_X(x)/\sqrt{P_X(x)} $; and $ ||\mathbf{L}_X||^2 $ is the squared Euclidean norm of the vector $ \mathbf{L}_X $. This approximation highlights the local Euclidean geometry: KL divergence behaves locally like a squared Euclidean distance in the space of transformed perturbation vectors $ \mathbf{L}_X $. Furthermore, that \eqref{eq:eit_kl_approx_L} shares the same second-order approximation, rendering the KL divergence locally symmetric. The term $ \epsilon^2 ||\mathbf{L}_X||^2 $ is equivalent to the $ \chi^2 $-divergence in \cite{csiszar1967information}, $ \chi^2(Q_X, P_X) = \sum_x \frac{(Q_X(x)-P_X(x))^2}{P_X(x)} $.

\subsection{The Divergence Transfer Matrix} \label{subsec:eit_dtm_final_direct}

To understand how perturbations in the input distribution $P_X(x)$ propagate through a channel to affect the output distribution within the EIT framework, the concept of the Divergence Transfer Matrix (DTM) is essential \cite{Zheng_2,Huang_2, Ulukus}. Consider a DMC from $X$ to an output $Y$, defined by the conditional PMF $ P_{Y|X}(y|x) $. Let $ P_X(x) $ be the reference input distribution, where $ P_X(x) > 0 $ for all $x$, and $ P_Y(y) = \sum_x P_X(x) P_{Y|X}(y|x) $ be the corresponding reference output distribution, assuming $ P_Y(y) > 0 $ for all $y$. If the input distribution for a specific $U=u$, $ P_{X|U}(x|u) $, is a perturbation of $ P_X(x) $ given in~\eqref{eq:eit_pxu_perturbation_sec3}, with $ \mathbf{L}_u $ satisfies \eqref{eq:eit_L_ortho_sqrtP_final} and \eqref{eq:eit_L_consistency_vector_final}, this induces a perturbation in the corresponding output distribution $ P_{Y|U}(y|u) = \sum_x P_{X|U}(x|u)P_{Y|X}(y|x) $. In the EIT framework, this output perturbation, when scaled appropriately, can also be represented by a vector $ \mathbf{L}_{Y,u} $. The DTM, denoted $ B_{Y|X} $, provides the linear transformation relating these scaled perturbation vectors:
\begin{equation} 
	\label{eq:eit_LYu_equals_BYX_Lu}
	\mathbf{L}_{Y,u} = B_{Y|X} \mathbf{L}_u.
\end{equation}
The DTM $ B_{Y|X} $, which has dimensions $ |\mathcal{Y}| \times |\mathcal{X}| $, is defined as:
\begin{equation} 
	\label{eq:eit_dtm_definition_direct}
	B_{Y|X} = \mathrm{diag}(P_Y)^{-1/2} P_{Y|X} \mathrm{diag}(P_X)^{1/2},
\end{equation}
where $ \mathrm{diag}(P_A) $ denotes a diagonal matrix with the elements of the vector $ P_A $ on its main diagonal, and $ P_{Y|X} $ is the channel transition matrix.

The singular value decomposition (SVD) of the DTM is: $ B_{Y|X} = \sum_k \sigma_k \boldsymbol{\psi}_k^Y (\boldsymbol{\phi}_k^X)^T $, where $ \sigma_k $ are the singular values, and $ \boldsymbol{\psi}_k^Y $ (left singular vectors, basis for output perturbation space) and $ \boldsymbol{\phi}_k^X $ (right singular vectors, basis for input perturbation space) are the corresponding singular vectors.
For any valid DTM derived from a stochastic channel matrix and positive marginals, the largest singular value is $ \sigma_1 = 1 $. The corresponding right singular vector $ \boldsymbol{\phi}_1^X $ is proportional to $ \mathbf{\sqrt{P_X}} $, and the corresponding left singular vector $ \boldsymbol{\psi}_1^Y $ is proportional to $ \mathbf{\sqrt{P_Y}} $.
Input perturbations $ \mathbf{L}_u $ that are collinear with $ \boldsymbol{\phi}_1^X \propto \mathbf{\sqrt{P_X}} $ do not satisfy the EIT orthogonality condition $ \mathbf{L}_u^T \mathbf{\sqrt{P_X}} = 0 $ and thus do not represent new information relative to $P_X$. The information-bearing perturbations $ \mathbf{L}_u $ considered in EIT lie in the subspace orthogonal to $ \mathbf{\sqrt{P_X}} $, which is spanned by the remaining right singular vectors $ \{\boldsymbol{\phi}_k^X\}_{k \ge 2} $. The DTM describes how these specific perturbation components are transformed and scaled by the channel. Analogously, for Eve's channel, a DTM $ B_{Z|X} $ can be defined as:
\begin{equation} \label{eq:eit_dtm_bzx_direct}
	B_{Z|X} = \mathrm{diag}(P_Z)^{-1/2} P_{Z|X} \mathrm{diag}(P_X)^{1/2},
\end{equation}
where $ P_Z(z) = \sum_x P_X(x) P_{Z|X}(z|x) $. An input perturbation $ \mathbf{L}_u $ will induce an output perturbation $ \mathbf{L}_{Z,u} = B_{Z|X} \mathbf{L}_u $ for Eve. These DTMs are crucial for approximating $ I(U;Y) $ and $ I(U;Z) $ in terms of the input perturbations $ \mathbf{L}_u $, as will be detailed in the next subsection.

\subsection{EIT Approximations for Mutual Information Terms} \label{subsec:eit_mi_approximations_formal}

We now state the EIT approximations for the mutual information terms in Problem \ref{prob:slic_original}. 
\begin{lem} \label{lem:approx_IUX_formal}
	Let $ P_{X|U}(x|u) $ be defined as in \eqref{eq:eit_pxu_perturbation_sec3} with corresponding scaled perturbation vectors $ \mathbf{L}_u $. The mutual information $ I(U;X) $ can be approximated as:
	\begin{equation} \label{eq:eit_IUX_approx_lemma}
		I(U;X) = \sum_{u \in \mathcal{U}} P_U(u) D_{\mathrm{KL}}(P_{X|U=u} || P_X) \approx \frac{\epsilon^2}{2} \sum_{u \in \mathcal{U}} P_U(u) ||\mathbf{L}_u||^2 = \frac{\epsilon^2}{2} \mathbb{E}_U[||\mathbf{L}_U||^2].
	\end{equation}
\end{lem}
\begin{proof}
	The proof is presented in Appendix A.
\end{proof}
To approximate $ I(U;Y) $ and $ I(U;Z) $, we utilize the DTMs for Bob's and Eve's channel. Let $ P_Y(y) $ and $ P_Z(z)$ be the marginal output distributions under the reference input $ P_X $. The DTMs are $ B_{Y|X} $ in~\eqref{eq:eit_dtm_definition_direct} and $ B_{Z|X} $ in~\eqref{eq:eit_dtm_bzx_direct}. Given an input perturbation $ \mathbf{L}_u $ for $ P_{X|U=u} $, the EIT states that the scaled perturbation of the output distribution $ P_{Y|U=u} $ from $ P_Y $ is $ \mathbf{L}_{Y,u} = B_{Y|X} \mathbf{L}_u $. Similarly, for Eve, $ \mathbf{L}_{Z,u} = B_{Z|X} \mathbf{L}_u $.
\begin{lem}\label{lem:approx_IUY_formal}
	With $ \mathbf{L}_u $ and $ B_{Y|X} $ defined as above, the mutual information $ I(U;Y) $ can be approximated as:
	\begin{equation} \label{eq:eit_IUY_approx_lemma}
		I(U;Y) = \mathbb{E}_U [ D_{\mathrm{KL}}(P_{Y|U=u} || P_Y) ] \approx \frac{\epsilon^2}{2} \sum_{u \in \mathcal{U}} P_U(u) ||B_{Y|X} \mathbf{L}_u||^2 = \frac{\epsilon^2}{2} \mathbb{E}_U [ ||B_{Y|X} \mathbf{L}_U||^2 ].
	\end{equation}
\end{lem}
\begin{proof}
	The proof is presented in Appendix A.
\end{proof}
\begin{lem} \label{lem:approx_IUZ_formal}
	Similarly, with $ \mathbf{L}_u $ and $ B_{Z|X} $ defined as above, the mutual information $ I(U;Z) $ can be approximated as:
	\begin{equation} \label{eq:eit_IUZ_approx_lemma}
		I(U;Z) = \mathbb{E}_U [ D_{\mathrm{KL}}(P_{Z|U=u} || P_Z) ] \approx \frac{\epsilon^2}{2} \sum_{u \in \mathcal{U}} P_U(u) ||B_{Z|X} \mathbf{L}_u||^2 = \frac{\epsilon^2}{2} \mathbb{E}_U [ ||B_{Z|X} \mathbf{L}_U||^2 ].
	\end{equation}
\end{lem}
\begin{proof}
	The proof is presented in Appendix A.
\end{proof}

These EIT approximations transform the objective function and constraints of the SIC problem into quadratic forms of the perturbation vectors $ \mathbf{L}_u $. This transformation is pivotal for deriving a tractable version of the problem, which will be explored in the subsequent section.

\section{Local Approximation of the SIC Problem} \label{sec:local_approx_slic}

In this section, we apply the EIT framework, developed in Section \ref{sec:eit_preliminaries}, to transform Problem \ref{prob:slic_original} into a quadratic programming formulation. This transformation is the cornerstone of the subsequent analysis.

\subsection{Transformation of SIC Problem} \label{subsec:transform_slic_terms_eit}

Utilizing Lemmas \ref{lem:approx_IUX_formal}--\ref{lem:approx_IUZ_formal} the following approximations hold. Using the approximation from~\eqref{eq:eit_IUY_approx_lemma}, the objective in \eqref{eq:slic_prob_obj} becomes:
	\begin{equation} \label{eq:slic_obj_approx_final}
		I(U;Y) \approx \frac{\epsilon^2}{2} \sum_{u \in \mathcal{U}} P_U(u) ||B_{Y|X} \mathbf{L}_u||^2 = \frac{\epsilon^2}{2} \sum_{u \in \mathcal{U}} P_U(u) \mathbf{L}_u^T V \mathbf{L}_u,
	\end{equation}
	where $ V \triangleq B_{Y|X}^T B_{Y|X} $; and $V$ is a symmetric and positive semidefinite matrix of size $ |\mathcal{X}| \times |\mathcal{X}| $.
	
	 Applying \eqref{eq:eit_IUX_approx_lemma}, constraint in \eqref{eq:slic_prob_rate} transforms to:
	\begin{align} 
		\label{eq:slic_rate_approx_final}
		I(U:X) \approx \frac{\epsilon^2}{2} \sum_{u \in \mathcal{U}} P_U(u) ||\mathbf{L}_u||^2 &\le R \quad  \text{or} \\
		\label{eq:slic_rate_approx_final_with_I}
		\sum_{u \in \mathcal{U}} P_U(u) \mathbf{L}_u^T I \mathbf{L}_u &\le \frac{2R}{\epsilon^2}, 
	\end{align}
	where $ I $ is the identity matrix.
	
	Utilizing~\eqref{eq:eit_IUZ_approx_lemma}, constraint in \eqref{eq:slic_prob_leakage} becomes:
	\begin{align}
		 \label{eq:slic_leakage_approx_final}
		I(U;Z) \approx \frac{\epsilon^2}{2} \sum_{u \in \mathcal{U}} P_U(u) ||B_{Z|X} \mathbf{L}_u||^2 &\le \Theta \quad  \text{or} \\
		 \label{eq:slic_leakage_approx_final_with Lambda}
		 \sum_{u \in \mathcal{U}} P_U(u) \mathbf{L}_u^T \Lambda \mathbf{L}_u &\le \frac{2\Theta}{\epsilon^2},
	\end{align}
	where $ \Lambda \triangleq B_{Z|X}^T B_{Z|X} $; and $ \Lambda $ is also symmetric and positive semidefinite and has size $ |\mathcal{X}| \times |\mathcal{X}| $.

The consistency constraint~\eqref{eq:slic_prob_consistency} from Problem~\ref{prob:slic_original} is already incorporated into the EIT framework through the structural constraint~\eqref{eq:eit_L_ortho_sqrtP_final} and \eqref{eq:eit_L_consistency_vector_final}. The probability normalization and non-negativity constraints on $ P_U(u) $ and $ P_{X|U}(x|u) $ remain.

\subsection{The Approximated SIC Optimization Problem} \label{subsec:eit_approximated_slic_problem_final_corrected}

By substituting the local approximations \eqref{eq:slic_obj_approx_final}, \eqref{eq:slic_rate_approx_final}, and \eqref{eq:slic_leakage_approx_final} into the original SIC Problem 1, and incorporating the structural constraints on $ \mathbf{L}_u $ in \eqref{eq:eit_L_ortho_sqrtP_final} and \eqref{eq:eit_L_consistency_vector_final}, we arrive at the EIT-approximated SIC problem. We define the scaled rate and leakage thresholds as $ R' \triangleq \frac{2R}{\epsilon^2} $ and $ \Theta' \triangleq \frac{2\Theta}{\epsilon^2} $. The common factor $ \epsilon^2/2 $ in the objective function \eqref{eq:slic_obj_approx_final} can be omitted for the purpose of the optimization as it is a positive constant.
\begin{problem} \label{prob:eit_slic_approximated}
	The objective is to find the PMF $ \{P_U(u)\}_{u \in \mathcal{U}} $ and the set of perturbation vectors $ \{\mathbf{L}_u\}_{u \in \mathcal{U}} $, where each $ \mathbf{L}_u \in \mathbb{R}^{|\mathcal{X}|} $, that solve:
	\begin{align}
		\underset{ \{P_U(u)\}, \{\mathbf{L}_u\} }{\mathrm{max}} \quad & \sum_{u \in \mathcal{U}} P_U(u) \mathbf{L}_u^T V \mathbf{L}_u \label{eq:eit_slic_opt_obj_problem_env} \\
		\mathrm{s.t.} \quad & \sum_{u \in \mathcal{U}} P_U(u) \mathbf{L}_u^T I \mathbf{L}_u \le R' \label{eq:eit_slic_opt_rate_problem_env} \\
		& \sum_{u \in \mathcal{U}} P_U(u) \mathbf{L}_u^T \Lambda \mathbf{L}_u \le \Theta' \label{eq:eit_slic_opt_leakage_problem_env} \\
		& \mathbf{L}_u^T \mathbf{\sqrt{P_X}} = 0, \quad \forall u \in \mathcal{U} \label{eq:eit_slic_opt_ortho_problem_env} \\
		& \sum_{u \in \mathcal{U}} P_U(u) \mathbf{L}_u = \mathbf{0} \label{eq:eit_slic_opt_consistency_L_problem_env} \\
		& \sum_{u \in \mathcal{U}} P_U(u) = 1 \label{eq:eit_slic_opt_pu_sum_problem_env} \\
		& P_U(u) \ge 0, \quad \forall u \in \mathcal{U} \label{eq:eit_slic_opt_pu_nonneg_problem_env}
	\end{align}
	where $ V \triangleq B_{Y|X}^T B_{Y|X} $, $ \Lambda \triangleq B_{Z|X}^T B_{Z|X} $, $ I $ is the identity matrix, and $ \mathbf{\sqrt{P_X}} $ is the vector with components $ \sqrt{P_X(x)} $.
\end{problem}
This formulation transforms the original non-linear SIC problem into a structured quadratic program. The primary optimization is over the set of perturbation vectors $\{\mathbf{L}_u\}$, which represent the encoding channels for each message. As we will rigorously prove in Section \ref{sec:approx_capacity_solution_structure}, a key structural property of this problem is that the optimal value of the objective function is independent of the source distribution $P_U(u)$. The distribution $P_U(u)$ acts as a mixing strategy whose role is to ensure the average constraints are met, rather than being a variable to be optimized for performance. This insight simplifies the problem significantly and paves the way for the analytical solution techniques that follow.
\begin{rem} \label{rem:non_negativity_eit}
	Problem \ref{prob:eit_slic_approximated} does not explicitly include constraints to enforce the non-negativity of the underlying perturbed conditional PMFs. That is because the solutions obtained for $ \{\mathbf{L}_u^*\} $ are considered meaningful only within a regime where $ \epsilon > 0 $ can be chosen small enough such that the resulting $ P_{X|U}(x|u) $ are valid PMFs for all $x, u$. Essentially, the non-negativity condition imposes an implicit bound on the magnitude of the components of the optimal $ \mathbf{L}_u^* $, ensuring that the solution remains within the local neighborhood of the reference distribution where the quadratic approximations of mutual information are accurate.
\end{rem}
\section{Approximate Secrecy Capacity and Solution Structure} \label{sec:approx_capacity_solution_structure}

Having transformed the SIC problem into its EIT-approximated quadratic form, we now proceed to analyze its solution structure. This section is dedicated to the core theoretical results of this work. We begin by employing Lagrangian duality and KKT conditions to derive the fundamental optimality conditions for the perturbation vectors. These conditions lead to an explicit expression for an approximate local secrecy capacity, whose value is determined by the optimal Lagrange multipliers. We then establish that these optimal multipliers can be found by solving an LP, a key result that renders the problem analytically tractable. Building on this, we fully characterize its behavior by analyzing the solution of this LP across different operational regimes.

For the subsequent analysis we rely on the following mild assumption.

\begin{ass}\label{ass_1}
	The matrix $\Lambda = B_{Z|X}^T B_{Z|X}$, when restricted to the valid perturbation subspace $S^\perp = \{L | L^T\sqrt{P_X}=0\}$, is positive definite.
\end{ass}
This assumption ensures that the generalized eigenvalues of the pencil $(V, \Lambda)$ are well-defined and finite, which implies that any non-trivial perturbation of the input signal incurs a non-zero amount of information leakage to the eavesdropper. If this were not the case, one could encode information with zero leakage, making the secrecy problem trivial.
	

\subsection{Lagrangian Duality and KKT Conditions} \label{subsec:kkt_conditions_final}

To analyze Problem \ref{prob:eit_slic_approximated}, we formulate its Lagrangian. The objective is to maximize~\eqref{eq:eit_slic_opt_obj_problem_env}. We introduce Lagrange multipliers: $ \rho^* \ge 0 $ for the scaled rate constraint \eqref{eq:eit_slic_opt_rate_problem_env}, $ \nu^* \ge 0 $ for the scaled leakage constraint \eqref{eq:eit_slic_opt_leakage_problem_env}, scalar multipliers $ \xi^*(u) $ for each orthogonality constraint~\eqref{eq:eit_slic_opt_ortho_problem_env}, a vector multiplier $ \boldsymbol{\mu}^* $ for the consistency constraint~\eqref{eq:eit_slic_opt_consistency_L_problem_env}, and a scalar $ \kappa^* $ for the normalization~\eqref{eq:eit_slic_opt_pu_sum_problem_env}. The non-negativity constraints \eqref{eq:eit_slic_opt_pu_nonneg_problem_env} are also part of the KKT conditions.

The Lagrangian $ \mathcal{L}(\{P_U\}, \{\mathbf{L}_u\}, \rho, \nu, \{\xi(u)\}, \boldsymbol{\mu}, \kappa) $ is:
\begin{align} \label{eq:lagrangian_slic_final}
	\mathcal{L} = &\sum_{u \in \mathcal{U}} P_U(u) \mathbf{L}_u^T V \mathbf{L}_u \nonumber \\
	&- \rho \left( \sum_{u \in \mathcal{U}} P_U(u) \mathbf{L}_u^T I \mathbf{L}_u - R' \right) \nonumber \\
	&- \nu \left( \sum_{u \in \mathcal{U}} P_U(u) \mathbf{L}_u^T \Lambda \mathbf{L}_u - \Theta' \right) \nonumber \\
	&- \sum_{u \in \mathcal{U}} \xi(u) (\mathbf{L}_u^T \mathbf{\sqrt{P_X}}) \nonumber \\
	&- \boldsymbol{\mu}^T \left( \sum_{u \in \mathcal{U}} P_U(u) \mathbf{L}_u \right) \nonumber \\
	&- \kappa \left( \sum_{u \in \mathcal{U}} P_U(u) - 1 \right) + \sum_{u \in \mathcal{U}} \zeta(u) P_U(u) \quad \text{, with } \zeta(u) \ge 0 \text{ for } P_U(u) \ge 0.
\end{align}

The KKT conditions provide a set of necessary conditions for a solution to be optimal. A key condition is stationarity, which requires the gradient of the Lagrangian with respect to the primal variables to be zero at the optimal point. By differentiating the full Lagrangian in \eqref{eq:lagrangian_slic_final}, with respect to each perturbation vector $ \mathbf{L}_u $, we can characterize the structure that these optimal vectors must possess.

\begin{thm} \label{thm:kkt_conditions_Lu_final}
	Let $ \{P_U^*(u)\} $, $ \{\mathbf{L}_u^*\} $, $ \rho^* \ge 0 $, $ \nu^* \ge 0 $, $ \{\xi^*(u)\} $, $ \boldsymbol{\mu}^* $, $ \kappa^* $, and $ \{\zeta^*(u) \ge 0\} $ be the optimal primal and dual variables for Problem \ref{prob:eit_slic_approximated}. If a suitable constraint qualification holds, the KKT conditions include stationarity with respect to each $ \mathbf{L}_u^* $, for $P_U^*(u) > 0$. This stationarity, implies that for any $u$ where $P_U^*(u) > 0$ and for non-trivial optimal $ \mathbf{L}_u^* \neq \mathbf{0} $, it holds that:
	\begin{equation} \label{eq:k_Lu_equals_0_final}
		(-V + \rho^*I + \nu^*\Lambda) \mathbf{L}_u^* = \mathbf{0}.
	\end{equation}
	The KKT conditions also include primal feasibility, dual feasibility ($ \rho^* \ge 0, \nu^* \ge 0, \zeta^*(u) \ge 0 $), and complementary slackness for all inequality constraints. Specifically for the rate and leakage constraints:
	\begin{align}
		\rho^* \left( \sum_{u \in \mathcal{U}} P_U^*(u) (\mathbf{L}_u^*)^T I \mathbf{L}_u^* - R' \right) &= 0 \label{eq:cs_rate_final} \\
		\nu^* \left( \sum_{u \in \mathcal{U}} P_U^*(u) (\mathbf{L}_u^*)^T \Lambda \mathbf{L}_u^* - \Theta' \right) &= 0 \label{eq:cs_leakage_final}
	\end{align}
\end{thm}
\begin{proof}
	The proof is presented in Appendix B.
\end{proof}

Condition \eqref{eq:k_Lu_equals_0_final} is pivotal as it indicates that any optimal non-zero perturbation vector $ \mathbf{L}_u^* $ must lie in the null space of the matrix $ K(\rho^*, \nu^*) \triangleq -V + \rho^*I + \nu^*\Lambda $. This connects the optimal perturbations to a generalized eigenvalue problem structure. This also clarifies the role of the source distributions $\{P_U(u)\}$. Although Problem \ref{prob:eit_slic_approximated} is formulated as a joint optimization over both $\{P_U(u)\}$ and $\{\mathbf{L}_u\}$, these variables are not chosen independently. In problems of this nature, the choice of the conditional mapping implicitly determines the output marginal distribution via the consistency relation $P_U(u) = \sum_x P_X(x) P_{U|X}(u|x)$. The following proposition formalizes a more powerful consequence: the optimal value of the problem is itself independent of the specific source distribution that achieves the optimum.
\begin{prop}\label{prop:pu_invariance}
	The optimal value of the objective function of the EIT-Approximated SIC Problem 2, is independent of the choice of the source distributions $\{P_U(u)\}$, provided that a distribution exists that can satisfy the problem's average constraints using the optimal set of perturbation vectors. 
\end{prop}
\begin{proof}
	The proof is presented in Appendix C.
\end{proof}

This proposition implies that the optimal performance is fundamentally determined by the channel properties and the rate-leakage budgets, rather than the statistics of the source message $U$. The role of the source distribution $P_U(u)$ is to serve as a mixing distribution that allows a set of optimal perturbation vectors $\{\mathbf{L}_u^*\}$ to satisfy the average constraints of the problem, particularly the rate, leakage, and consistency constraints \eqref{eq:eit_slic_opt_rate_problem_env}, \eqref{eq:eit_slic_opt_leakage_problem_env}, and \eqref{eq:eit_slic_opt_consistency_L_problem_env}. The existence of such a distribution supported on a small number of points is guaranteed by Caratheodory's theorem and its information-theoretic counterpart, the Support Lemma \cite{ElGamalKim11, csiszar2011}. Specifically, a cardinality of $|\mathcal{U}|$ on the order of the number of active constraints is sufficient to construct the optimal average. This theoretical result is further supported by the numerical experiments shown in Table \ref{tab:varying_cardinality_U_appendix} in Appendix N which show that the achievable utility from the primal optimization stabilizes for a small, finite cardinality $|\mathcal{U}|$. This insight justifies our focus on determining the optimal Lagrange multipliers $(\rho^*, \nu^*)$ and the corresponding set of active perturbation directions, which is the subject of the following subsections.

\subsection{Local Secrecy Capacity} \label{subsec:c_lic_final}

The KKT conditions, particularly \eqref{eq:k_Lu_equals_0_final}, allow us to derive an expression for the maximum achievable utility in the EIT-Approximated SIC problem. This maximum utility, when scaled by $ \epsilon^2/2 $ to reflect an actual information rate, is termed the approximate local secrecy capacity, $ C_{\mathrm{SIC}} $.

\begin{thm} \label{thm:c_lic_final}
	Let $ (\rho^*, \nu^*) $ be the optimal Lagrange multipliers for the scaled rate and leakage constraints \eqref{eq:eit_slic_opt_rate_problem_env} and \eqref{eq:eit_slic_opt_leakage_problem_env}, respectively, as determined by the solution to the dual problem in Theorem \ref{thm:general_lp_for_multipliers}. The approximate local secrecy capacity, $ C_{\mathrm{LIC}} $, is given by the linear combination:
	\begin{equation} \label{eq:c_lic_formula_final}
		C_{\mathrm{SIC}} = \rho^* R + \nu^* \Theta
	\end{equation}
	where $ R $ and $ \Theta $ are the original unscaled constraint parameters from Problem \ref{prob:slic_original}. 
\end{thm}
\begin{proof}
	The proof is presented in Appendix D.
\end{proof}
The parameters $ \rho^* $ and $ \nu^* $ in \eqref{eq:c_lic_formula_final} act as prices for the rate and leakage constraints, respectively, indicating the sensitivity of $ C_{\mathrm{SIC}} $ to marginal changes in $R$ and $\Theta$. Furthermore, the KKT stationarity condition $ V \mathbf{L}_u^* = (\rho^*I + \nu^*\Lambda) \mathbf{L}_u^* $ can be rearranged to $ (V - \nu^*\Lambda)\mathbf{L}_u^* = \rho^* \mathbf{L}_u^* $. This reveals that for a given optimal $\nu^*$, the corresponding optimal $\rho^*$ is the generalized eigenvalue of the pencil $ (V - \nu^*\Lambda, I) $ associated with the optimal perturbation direction $\mathbf{L}_u^*$. For maximization, one is typically interested in solutions corresponding to the largest such feasible eigenvalue. This highlights the connection between the optimal multipliers and the spectral properties of the channel matrices.

\subsection{Optimal Lagrange Multipliers $\rho^*$ and $\nu^*$} \label{subsec:optimal_multipliers_lp_general}
The approximate local secrecy capacity $C_{\mathrm{SIC}}$, as given by Theorem \ref{thm:c_lic_final}, is determined by the optimal non-negative Lagrange multipliers $\rho^*$ and $\nu^*$. These multipliers are not arbitrary; they are intrinsically linked to the KKT conditions of Problem \ref{prob:eit_slic_approximated}. A central KKT condition for optimality is that the dual function, obtained from the Lagrangian \eqref{eq:lagrangian_slic_final} by maximizing over the primal variables $\{\mathbf{L}_u\}$, must be well-defined and finite. This requires that the matrix $K(\rho, \nu)$ be positive semidefinite on the valid perturbation subspace, $\mathcal{S}^{\perp} = \{\mathbf{L} | \mathbf{L}^T \mathbf{\sqrt{P_X}} = 0\}$. Remarkably, it can be shown that this condition, which is a linear matrix inequality, can be transformed into a finite set of linear inequalities. This allows the problem of finding the optimal multipliers to be formulated as a general LP.
\begin{thm} \label{thm:general_lp_for_multipliers}
	The optimal Lagrange multipliers $\rho^* \ge 0$ and $\nu^* \ge 0$ that determine $C_{\mathrm{SIC}}$ are the solution to the following LP:
		\begin{align}
			\underset{\rho, \nu}{\mathrm{max}} \quad & \rho R + \nu \Theta \label{eq:lp_general_obj_final_v3} \\
			\mathrm{s.t.} \quad & \rho + \nu \lambda_j \ge d_j \lambda_j, \quad \forall j \in \mathcal{J}_{\perp} \label{eq:lp_general_constr_eig_final_v3} \\
			& \rho R + \nu \Theta \le C_{\mathrm{max}} \label{eq:lp_general_constr_obj_upper_final_v3} \\
			& \rho \ge 0, \quad \nu \ge 0 \label{eq:lp_general_constr_nonneg_final_v3}					
		\end{align}
			where $\mathcal{J}_{\perp}$ indexes the $M = |\mathcal{X}|-1$ modes spanning the perturbation subspace $\mathcal{S}^{\perp}$. For each mode $j \in \mathcal{J}_{\perp}$, $d_j$ is the $j$-th generalized eigenvalue of the pencil $(V, \Lambda)$ restricted to $\mathcal{S}^{\perp}$, and $\lambda_j$ is the corresponding eigenvalue of $\Lambda$ that simultaneously diagonalizes $V$ and $\Lambda$. $R, \Theta > 0$, and $C_{\max} = \lambda_{\max}^{\perp}(V)R$.
\end{thm}
\begin{proof}
	The proof is presented in Appendix E.
\end{proof}

This theorem is a significant result as it makes the computation of optimal multipliers analytically and numerically tractable for any general wiretap channel within the EIT framework. It establishes that the complex, non-convex problem of local secrecy capacity can be understood through the tractable framework of LP. Moreover, it reveals that the constraints of this LP are determined by the generalized eigenvalues of the channel matrices $V$ and $\Lambda$. This provides us with the insight that the local trade-offs between utility, secrecy, and encoding complexity are governed by the spectral properties of the legitimate and eavesdropper channels. The validity of solving this LP via standard solvers is confirmed numerically in Appendix N and Table \ref{tab:lp_vertex_search_appendix}, where the solver's result is shown to match a direct exhaustive search over the feasible vertices.

For the LP in Theorem \ref{thm:general_lp_for_multipliers} to be meaningful, its feasible region must be non-empty. The following lemma provides a sufficient condition for the existence of a feasible solution.
\begin{lem} \label{thm:lp_feasibility_condition_final_v2}
	A sufficient condition for the LP in Theorem \ref{thm:general_lp_for_multipliers} to have a feasible solution is:
	\begin{equation} \label{eq:lp_feasibility_condition_formula_final_v2}
		\lambda_{\max}^{\perp}(V) > \frac{\Theta}{R} \cdot d_{\max}^{\perp}(V, \Lambda),
	\end{equation}
	where $\lambda_{\max}^{\perp}(V)$ and $d_{\max}^{\perp}(V, \Lambda)$ are the largest eigenvalue of $V$ and generalized eigenvalue of the pencil $(V, \Lambda)$,respectively, restricted to the perturbation subspace $\mathcal{S}^{\perp}$.
\end{lem}
\begin{proof}
	The proof is presented in Appendix F.
\end{proof}

Condition \eqref{eq:lp_feasibility_condition_formula_final_v2} can be rewritten as $R \cdot \lambda_{\max}^{\perp}(V) > \Theta \cdot d_{\max}^{\perp}(V, \Lambda)$. This provides an insightful relationship between the channel properties and the SIC problem parameters. It requires that the maximum achievable utility in a purely rate-dominant regime must be strictly greater than the maximum achievable utility in a purely leakage-dominant regime. This ensures that the feasible region for $(\rho, \nu)$ is non-empty and allows for a non-trivial solution where positive utility can be achieved. It essentially requires that Bob's best channel gain is sufficiently high to overcome the cost imposed by the channel's best leakage efficiency for the given rate and leakage thresholds.

\subsection{Local Secrecy Capacity Regimes} 
\label{subsec:clic_regimes_final_v3_bounded}
The LP established in Theorem \ref{thm:general_lp_for_multipliers} provides a systematic method for determining the optimal Lagrange multipliers $(\rho^*, \nu^*)$, and thus the value of $C_{\mathrm{SIC}}$. The solution to this LP is not static; it adapts to the specific problem parameters, namely the rate budget $R$, the leakage allowance $\Theta$, and the channel characteristics encapsulated by the eigenvalues. Specifically, the optimal vertex of the LP's feasible region shifts as these parameters change. This behavior gives rise to distinct operational regimes for the approximate local secrecy capacity. We can characterize these regimes by analyzing which set of constraints becomes active at the optimum, which in turn dictates the optimal signaling strategy. For instance, in a leakage-dominant regime, the solution provides a clear principle: the system should prioritize perturbation directions (i.e., signaling dimensions) that align with the channel's most leakage-efficient eigenmodes. The following theorem formalizes this characterization by describing the solution for $C_{\mathrm{SIC}}$ within each of these fundamental regimes. 
\begin{thm} \label{thm_c_lic_regimes}	
The approximate local secrecy capacity is given by $C_{\mathrm{SIC}} = \rho^*R + \nu^*\Theta$, where $(\rho^*, \nu^*)$ is the solution to the LP in Theorem \ref{thm:general_lp_for_multipliers}. The solution can be expressed as the maximum of candidate capacities arising from the vertices of the LP's feasible region:
\begin{equation} \label{eq:clic_max_over_regimes_v2}
	C_{\mathrm{SIC}} = \max \{C_{R}, C_{\Theta}, C_{\mathrm{inter}}\},	
\end{equation}
where $C_{R} = \min\{\lambda_{\max}^{\perp}(V)R, C_{\mathrm{max}}\}$ is the capacity from the optimal vertex on the $\rho$-axis (Rate-dominant regime); $C_{\Theta} = \min\{d_{\max}^{\perp}(V, \Lambda)\Theta, C_{\mathrm{max}}\}$ is the capacity from the optimal vertex on the $\nu$-axis (Leakage-dominant regime); and $C_{\mathrm{inter}}$ is the capacity achieved if the optimal vertex lies in the interior of the first quadrant ($\rho^*>0, \nu^*>0$), and is given by the maximum objective value over all such feasible vertices:
	\begin{equation} \label{eq:c_inter_definition_v2}
		C_{\mathrm{inter}} = \max_{(\rho_{jk}^*, \nu_{jk}^*) \in \mathcal{V}_{\mathrm{inter}}} \left\{ \rho_{jk}^* R + \nu_{jk}^* \Theta \right\},
	\end{equation}
	where $\mathcal{V}_{\mathrm{inter}}$ is the set of all feasible vertices not on the axes. A point $(\rho_{jk}^*, \nu_{jk}^*)$ belongs to $\mathcal{V}_{\mathrm{inter}}$ if it is the solution to the system of equations for two active constraints, and it satisfies the feasibility conditions:
	\begin{subequations} \label{eq:valid_pair_conditions}
		\begin{align}
			& \rho_{jk}^* > 0, \quad \nu_{jk}^* > 0, \\
			& \rho_{jk}^* + \nu_{jk}^* (d_{\Lambda})_m \ge (d_V)_m, \quad \forall m \in \mathcal{J}_{\perp}, \\
			& \rho_{jk}^* R + \nu_{jk}^* \Theta \le C_{\mathrm{max}}.
		\end{align}
	\end{subequations}

If the set $\mathcal{V}_{\mathrm{inter}}$ is empty, $C_{\mathrm{inter}}$ is taken to be $-\infty$.
\end{thm}
\begin{proof}
		The proof is presented in Appendix G.
\end{proof}
Theorem \ref{thm_c_lic_regimes} provides a complete characterization of the EIT-approximated secrecy capacity, partitioning the solution space into distinct operational regimes. The dominant regimes have clear closed-form expressions corresponding to EIT approximations of the IB\cite{TishbyPereiraBialek00} and a privacy funnel\cite{MakhdoumiSalamatianFawazMedard14} like problem, respectively. The intermediate regime reflects a direct trade-off where the optimal strategy is determined by the intersection of multiple performance limits, as captured by the active constraints of the LP in Theorem 8.  For instance, the BSWC, having only one relevant perturbation mode as will be shown, provides a clear example where this structure results in a simple piecewise formula for $C_{\mathrm{SIC}}$. The behavior of these regimes across different channels, including the BSWC and more complex multi-mode channels with larger alphabets, will be numerically demonstrated in Section \ref{sec:numerical_illustations}.

\subsection{The Commuting Matrices Case} \label{subsec:commuting_matrices_final_v5}

The general analysis of the EIT-approximated SIC problem simplifies significantly under the special condition that the channel matrices $V$ and $\Lambda$ commute. While the condition of commutativity implies a specific alignment between the principal perturbation directions of Bob's and Eve's effective channels and is not generally expected to hold for arbitrary wiretap channels, its analysis provides deep, tractable insights into the optimal solution structure. First, it serves as an essential analytical benchmark where the problem becomes fully tractable, allowing for a clear and intuitive understanding of the optimal strategies. Second, it is not purely a mathematical abstraction; it provides the exact EIT solution for important canonical models, such as the BSWC with a uniform input.

When $V$ and $\Lambda$ commute, they are simultaneously diagonalizable by a single orthogonal matrix $Q$ whose columns $\{\mathbf{e}_j\}_{j=1}^{|\mathcal{X}|}$ are common orthonormal eigenvectors \cite{horn2013matrix,cookbook} spanning the subspace orthogonal to $\sqrt{\mathbf{P_X}}$ the following simplifications hold.
\begin{thm} \label{thm:commuting_matrices_unified}
	If the matrices $V$ and $\Lambda$ commute and the eigenvectors $\{\mathbf{e}_j\}$ are ordered such that $\{\mathbf{e}_j\}_{j \in \mathcal{J}_{\perp}}$ then:
	
	 The LP from Theorem \ref{thm:general_lp_for_multipliers} simplifies as its constraints \eqref{eq:lp_general_constr_eig_final_v3} can be expressed directly using the standard eigenvalues of $V$ and $\Lambda$:
		\begin{equation} \label{eq:lp_commuting_constr}
			\rho + \nu (d_{\Lambda})_j \ge (d_V)_j, \quad \forall j \in \mathcal{J}_{\perp},
		\end{equation}
		where $\mathcal{J}_{\perp} = \{2, \dots, |\mathcal{X}|\}$ form an orthonormal basis for the EIT perturbation subspace $\mathcal{S}^{\perp}$; and $(d_V)_j, (d_{\Lambda})_j,$ and $(d_I)_j=1$ are the respective eigenvalues of $V, \Lambda,$ and $I$ for each eigenmode $j \in \mathcal{J}_{\perp}$.
		
	 Moreover, letting $(\rho^*, \nu^*)$ be the optimal solution from this LP, the KKT stationarity condition \eqref{eq:k_Lu_equals_0_final} implies that if an optimal perturbation $ \mathbf{L}_u^* = \sum_{j \in \mathcal{J}_{\perp}} (\tilde{L}_u^*)_j \mathbf{e}_j $ has a non-zero component along a common eigenmode $\mathbf{e}_j$ (i.e., $(\tilde{L}_u^*)_j \neq 0$), then the eigenvalues for that mode must satisfy:
		\begin{equation} \label{eq:commuting_eigenvalue_condition}
			(d_V)_j = \rho^* + \nu^* (d_{\Lambda})_j.
		\end{equation}
\end{thm}
\begin{proof}
			The proof is presented in Appendix H.
\end{proof}

Theorem \ref{thm:commuting_matrices_unified} reveals that in the commuting case, the EIT-approximated SIC problem decouples across the shared eigenmodes. This can be intuitively understood as an optimization over the total expected squared magnitudes of the perturbation coefficients, $E_j \triangleq \sum_u P_U(u)((\tilde{L}_u)_j)^2$, for each available eigenmode $j \in \mathcal{J}_{\perp}$. The problem becomes maximizing the total utility $\sum_j (d_V)_j E_j$ subject to a rate constraint on $\sum_j E_j$ and a leakage constraint on $\sum_j (d_{\Lambda})_j E_j$. This is an LP in terms of the variables $\{E_j\}$. A conceptual illustration for a 2-mode system is provided in Appendix H alongside the proof.

\section{Numerical Illustrations and Comparisons} \label{sec:numerical_illustations}

In this section, we provide numerical examples to validate the proposed EIT framework and illustrate the key theoretical results from Section \ref{sec:approx_capacity_solution_structure}. We begin by benchmarking the EIT approximation methodology against the well-known IB problem. We then demonstrate the validity and utility of the general LP for determining the optimal Lagrange multipliers in a complex, multi-mode channel scenario. Finally, we explore the behavior of the approximate local secrecy capacity and its operational regimes, using both a general multi-mode system and the canonical BSWC as illustrative examples.

\subsection{Validation of the EIT framework} \label{subsubsec:ib_validation}
To assess the efficacy of the EIT approximation methodology we consider its application to the well-known IB problem \cite{TishbyPereiraBialek00}. The IB problem seeks to maximize $I(U;Y)$ subject to $I(U;X) \le R$. An EIT-approximated IB (EIT-IB) problem can be formulated by omitting the leakage-related terms from Problem \ref{prob:eit_slic_approximated}.

\begin{figure}[!htbp]
	\centering
	\includegraphics[width=0.8\textwidth]{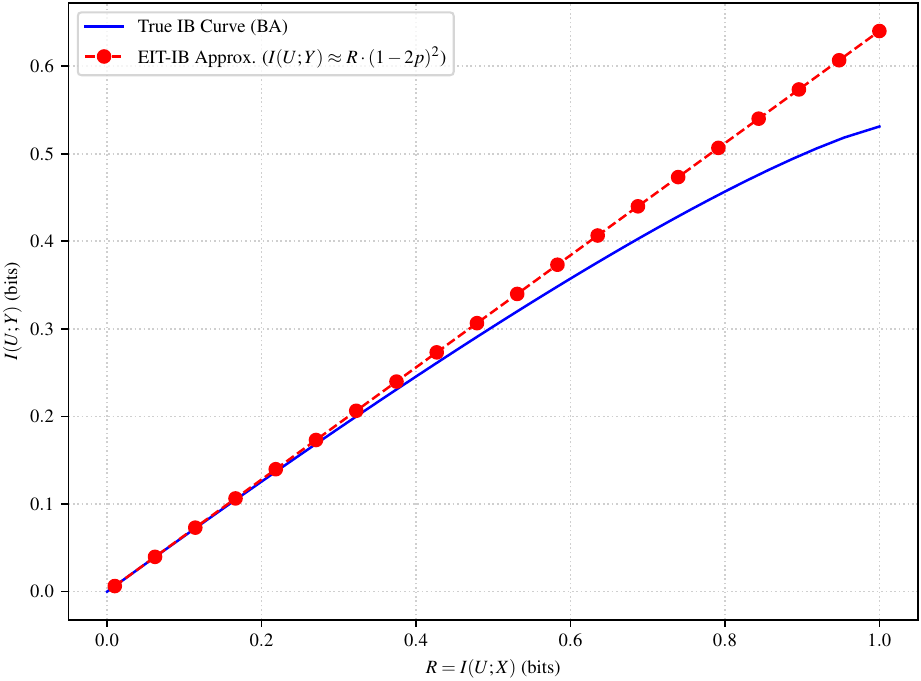} 
	\caption{IB Comparison for a BSC($p_{\mathrm{bob}}=0.1$) with uniform $P_X$ and binary $U$. Solid blue represents the true IB curve (Blahut-Arimoto solution) and red circles dashed the analytical EIT-IB.}
	\label{fig:EIT_IB_vs_BA_final}
\end{figure}

As observed in Figure \ref{fig:EIT_IB_vs_BA_final}, the analytical EIT-IB solution accurately capture the initial slope of the true IB curve, demonstrating EIT's strength as a local approximation for small rate constraints $R$. For larger $R$, the EIT approximation, which can be linear in $R$ for simplified cases like the BSWC, diverges from the true concave IB curve, as the latter accounts for global information-theoretic saturation limits. This comparison underscores the regime of validity for the approximations and confirms the utility of the EIT-IB in finding optima.

To demonstrate the validity of the general LP formulation of Theorem \ref{thm:general_lp_for_multipliers}, beyond simple or commuting cases, we consider a channel with input alphabet size $|\mathcal{X}|=5$, generated by quantizing an AWGN channel. This results in a $4$-dimensional perturbation subspace and non-commuting matrices $V$ and $\Lambda$. Figure \ref{fig:lp_solution_general_channel_final} illustrates the solution to the LP for a specific set of parameters $R$ and $\Theta$.

\begin{figure}[!htbp]
	\centering
	 \includegraphics[width=0.8\textwidth]{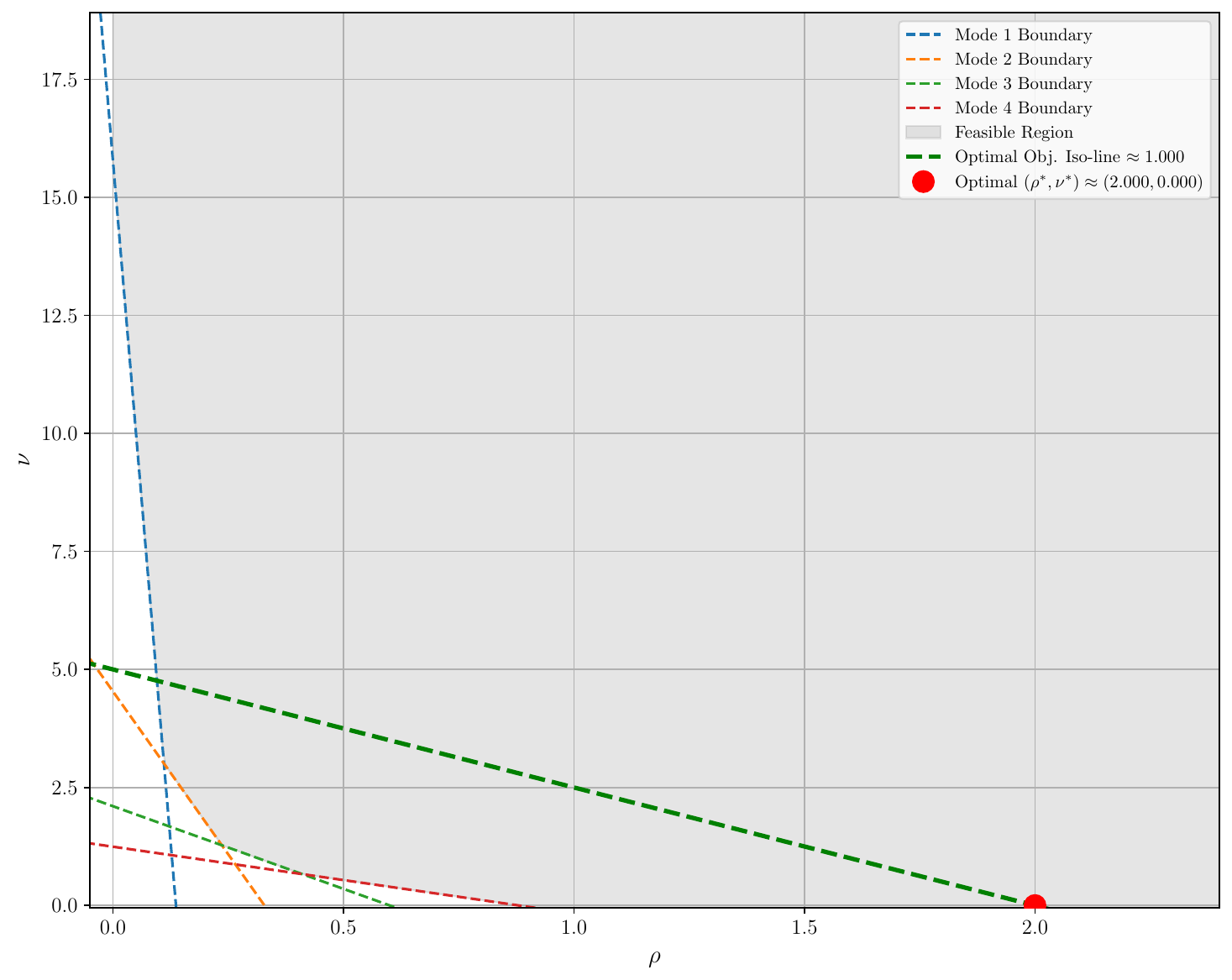}
	\caption{The LP solution for optimal multipliers for a numerically generated channel with $|\mathcal{X}|=5$. The optimal vertex is found at the boundary of the feasible region defined by the four eigenmode constraint lines.}
	\label{fig:lp_solution_general_channel_final}
\end{figure}

The figure depicts the feasible region for the $(\rho, \nu)$ for a general channel with $|\mathcal{X}|=5$. The dashed lines are the linear boundaries imposed by each of the four ($|\mathcal{X}|-1$) eigenmodes of the LP problem and the shaded gray area represents the feasible set of all $(\rho, \nu)$ pairs satisfying all constraints simultaneously. The thick green dashed line is an iso-objective line, representing a set of points where the objective function is constant. In an LP, all iso-objective lines are parallel; this green line is the one that passes through the optimal solution.
The optimal solution $(\rho^*, \nu^*)$, marked by the red circle, is found at a vertex of this feasible region. In this specific example, the solution is $(\rho^* \approx 0.889, \nu^* \approx 0.000)$, indicating a rate-dominant regime where the leakage constraint is not active for the chosen $R$ and $\Theta$. This numerical experiment validates that the general LP formulation correctly identifies the optimal multipliers at an extreme point of the feasible set defined by the KKT-derived constraints. As shown in Table \ref{tab:lp_vertex_search_appendix} in Appendix N, the results from both methods match perfectly, confirming the correctness of the LP framework for general channels.

Theorem \ref{thm_c_lic_regimes} provides a complete analytical characterization of the value of $C_{\mathrm{SIC}}$. To visualize these operational regimes, Figure \ref{fig:clic_normalized_vs_ratio} plots the normalized capacity as a function of the constraint ratio. This curve reveals the intrinsic behavior of the optimal solution which is governed by the solution to the general LP from Theorem \ref{thm:general_lp_for_multipliers}.
\begin{figure}[!htbp]
	\centering
	 \includegraphics[width=0.8\textwidth]{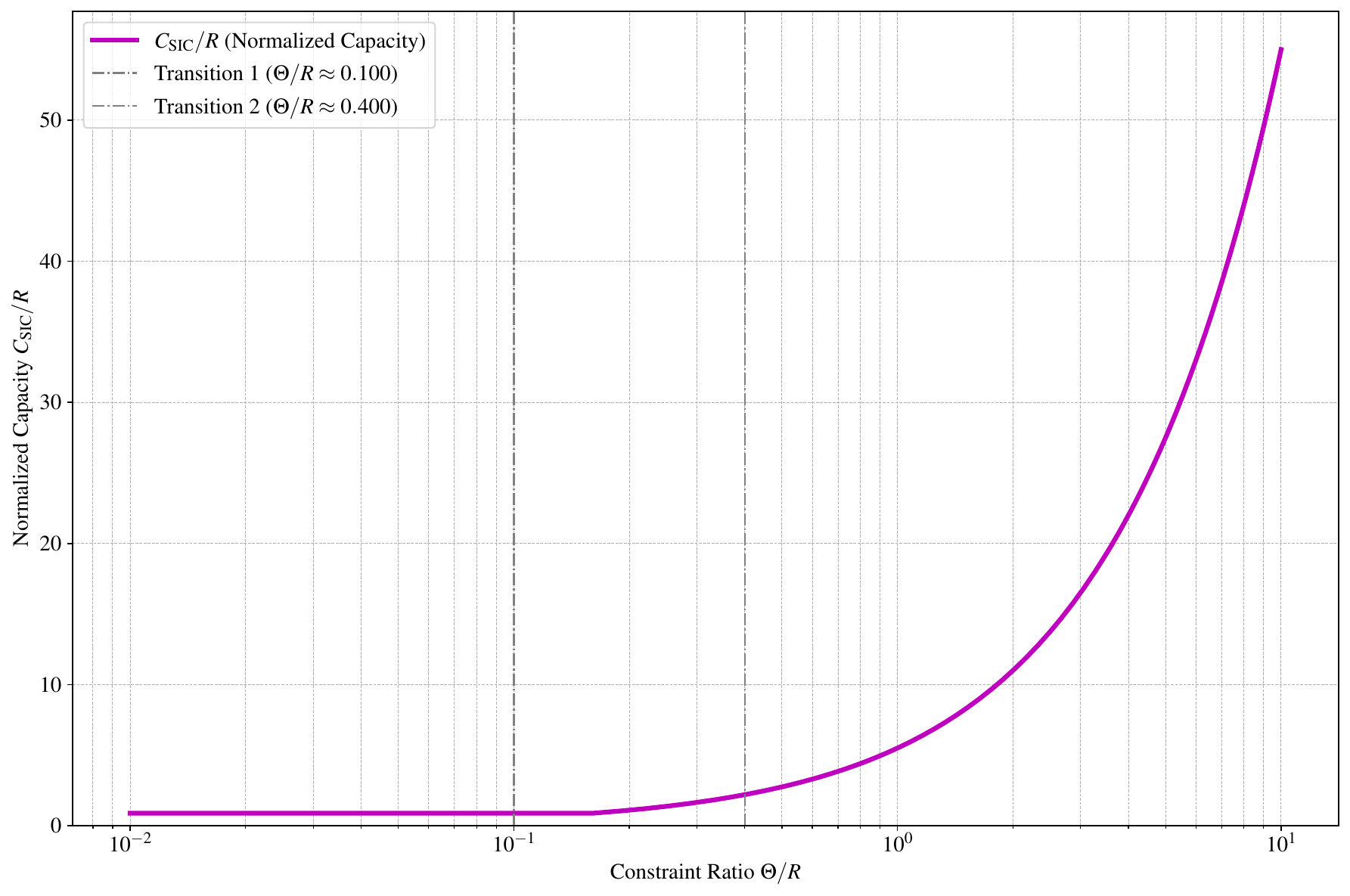} 
	\caption{Normalized approximate local secrecy capacity as a function of the constraint ratio. }
	\label{fig:clic_normalized_vs_ratio}
\end{figure}

As shown in Figure \ref{fig:clic_normalized_vs_ratio}, the behavior of the normalized capacity clearly delineates the three fundamental regimes. For small values of $\Theta/R$, the system is secrecy-constrained and operates in the leakage-dominant regime. Here, $C_{\mathrm{SIC}} \approx d_{\max}^{\perp}(V, \Lambda) \cdot \Theta$, and thus $C_{\mathrm{SIC}}/R$ grows linearly with the ratio $\Theta/R$, with a slope given by $d_{\max}^{\perp}(V, \Lambda)$. Conversely, for large values of $\Theta/R$, the system becomes resource-constrained and enters the rate-dominant regime. In this region, $C_{\mathrm{SIC}}$ saturates at $ \lambda_{\max}^{\perp}(V) R$, causing the normalized capacity to become constant at the value $\lambda_{\max}^{\perp}(V)$. Between these two extremes lies the intermediate regime, where the normalized capacity provides a smooth transition, reflecting a balanced trade-off where the optimal solution to the LP is determined by multiple active constraints. This plot demonstrates how the balance between secrecy and encoding resources dictates the ultimate efficiency of the secure communication scheme. 

\begin{figure}[H]
	\centering
	 \includegraphics[width=0.75\textwidth]{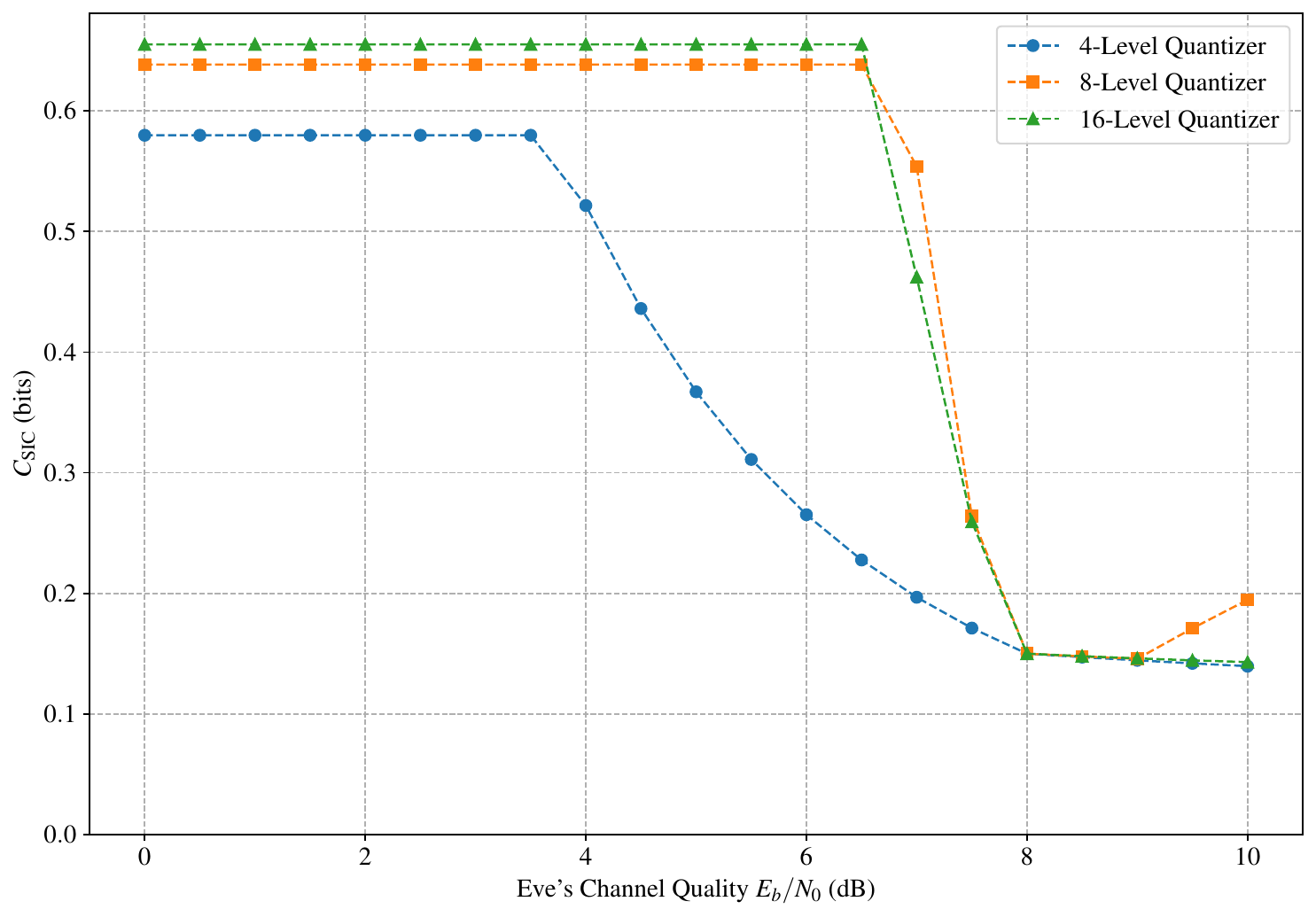} 
	\caption{$C_{\mathrm{SIC}}$ as a function of Eve's channel quality for different output quantization levels $|\mathcal{Z}|$, with $|X|=8$, Bob's $E_{b}/N_{0}=8.0 \text{dB}, R=0.5, \Theta=0.1$}
	\label{fig:clic_vs_eve_snr_quantization}
\end{figure}
The practical impact of system parameters is further explored in Figure \ref{fig:clic_vs_eve_snr_quantization}, which illustrates how $C_{\mathrm{SIC}}$ varies with Eve's channel quality for different output quantization levels of $|\mathcal{Z}|$. As expected, $C_{\mathrm{SIC}}$ decreases as Eve's channel improves. More interestingly, the figure shows that a coarser quantizer at the eavesdropper significantly limits Eve's ability to resolve information, resulting in a higher secrecy capacity for Alice. This demonstrates the utility of the framework for analyzing the impact of practical receiver characteristics on secure performance.

\subsection{Binary-Input Channels and the BSWC}\label{subsec:binary_channels}

A key insight of the EIT analysis is that the framework simplifies greatly for the broad and practical class of channels with a binary input alphabet, such as those modeling BPSK signaling. For any such channel, regardless of the output alphabet sizes $|\mathcal{Y}|$ and $|\mathcal{Z}|$, the matrices $V$ and $\Lambda$ are always $2 \times 2$.
Consequently, the subspace of valid perturbations $\mathcal{S}^{\perp}$ is always one-dimensional. This means that for any binary-input channel, there is only a single effective eigenmode, $|\mathcal{J}_{\perp}|=1$, for perturbations to consider. This simplifies the general LP for finding optimal multipliers, as it now contains only a single eigenmode constraint $\rho + \nu \lambda_j \ge d_j \lambda_j$, for $j=2$. The solution to this simplified one constraint LP always leads to a clean, piece-wise closed-form expression for $C_{\mathrm{SIC}}$, as characterized by Theorem \ref{thm_c_lic_regimes}, where the intermediate regime collapses to a single transition point. We use the BSWC as a canonical example of such a binary-input channel to demonstrate these results.

To solve the EIT-Approximated SIC problem numerically for this channel, we fix the cardinality of $|\mathcal{U}|$ following the result from Proposition \ref{prop:pu_invariance}. For problems involving variables in an $(|\mathcal{X}|-1)$-dimensional space and a set of linear and quadratic constraints, a cardinality of $|\mathcal{U}|$ on the order of $|\mathcal{X}|$ is typically sufficient.

We now provide a comprehensive case study using the BSWC to concretely illustrate the theoretical results of Section \ref{sec:approx_capacity_solution_structure}. The BSWC is an ideal illustrative platform for two primary reasons. First, as an instance of a binary-input channel, its EIT analysis simplifies significantly due to its one-dimensional perturbation subspace, as discussed above. Second, with a uniform input distribution, its EIT matrices $V$ and $\Lambda$ commute, making it a perfect example for the commuting matrices case analysis. The detailed derivation of the BSWC's specific EIT parameters and its piecewise closed-form formula for $C_{\mathrm{SIC}}$ are provided in Appendix I. For all subsequent BSWC illustrations, a uniform input distribution $P_X = [0.5, 0.5]^T$ is assumed.

\subsubsection{Approximate vs. True Secrecy Capacity for BSWC}
Figure \ref{fig:CLIC_vs_Cs_BSWC_final_sec6_full} compares the EIT-derived $ C_{\mathrm{SIC}} $ with the true secrecy capacity $ C_s = \max(0, H_b(q_{\mathrm{eve}}) - H_b(p_{\mathrm{bob}})) $.

\begin{figure}[!htbp]
	\centering
	 \includegraphics[width=0.85\textwidth]{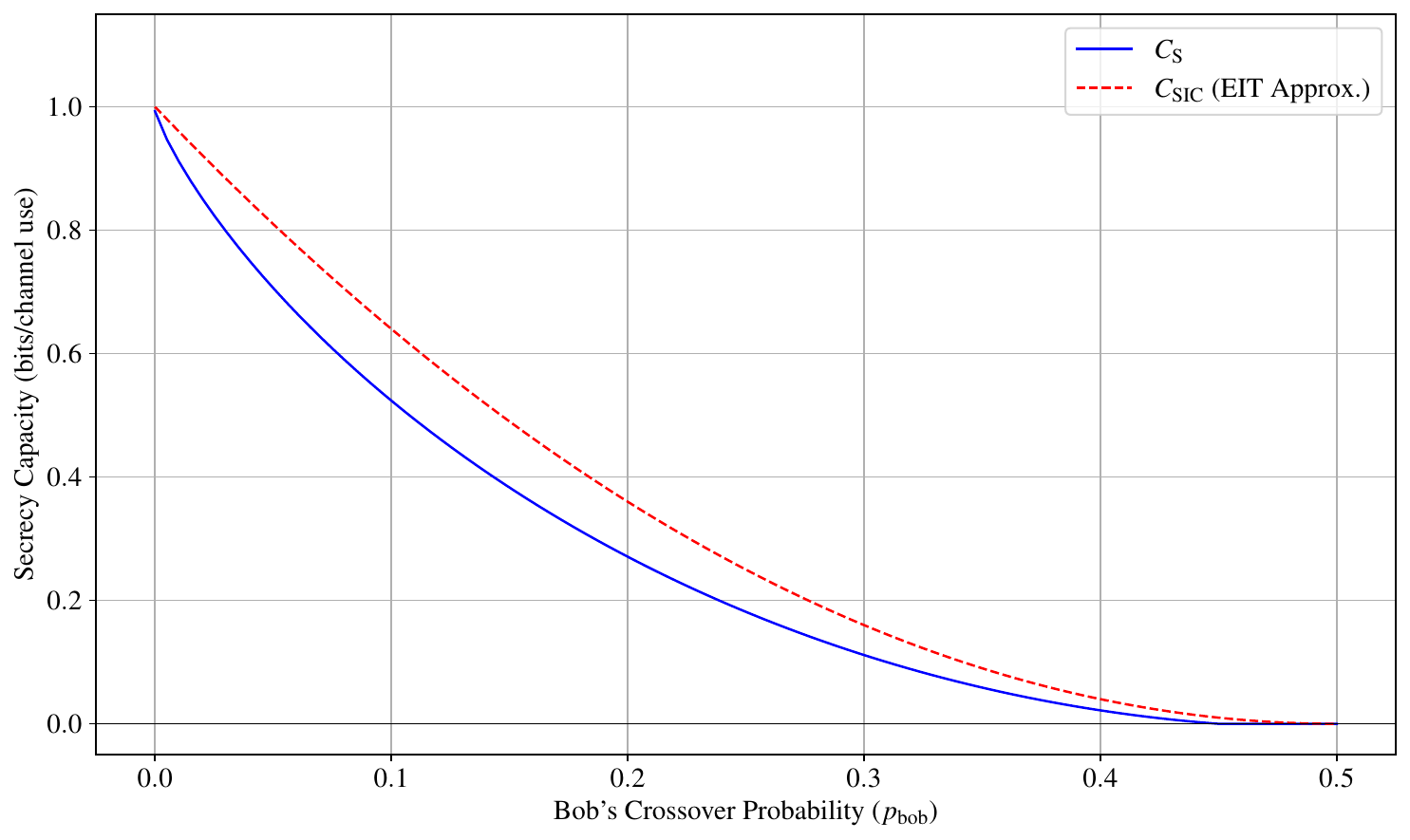} 
	\caption{Comparison of the EIT-approximated local secrecy capacity (dashed red line) for $q_{\mathrm{eve}}= 0.45, R=1.0, \Theta/R=0.085$, with the true secrecy capacity (solid blue line) for the BSWC. }
	\label{fig:CLIC_vs_Cs_BSWC_final_sec6_full}
\end{figure}

The comparison in Figure \ref{fig:CLIC_vs_Cs_BSWC_final_sec6_full} illustrates the performance of $ C_{\mathrm{SIC}} $ against the true secrecy capacity $ C_s $ as $ p_{\mathrm{bob}} $ varies. Both capacities exhibit the expected qualitative behavior: they are highest when Bob's channel is perfect ($ p_{\mathrm{bob}}=0 $) and degrade as $ p_{\mathrm{bob}} $ increases, eventually reaching zero. The $ C_{\mathrm{SIC}} $ curve provides an approximation to $ C_s $, whose accuracy is influenced by the chosen $R$ and $\Theta$ values and how well the operating point aligns with the EIT local perturbation regime. 

\subsubsection{Operational Regimes}

Figure \ref{fig:clic_bswc_regimes_final} illustrates the operational regimes of $C_{\mathrm{SIC}}$ for the BSWC, providing a concrete example of the general theory from Theorem \ref{thm_c_lic_regimes}. Due to the BSWC's one-dimensional perturbation subspace, the intermediate regime collapses to a single transition point, resulting in a clean piecewise structure. The solid black line, representing the true $C_{\mathrm{SIC}}$ for the BSWC, is precisely the lower envelope of the capacities from the two dominant regimes: $C_{\mathrm{SIC}} = \min(C_{\mathrm{R-dom}}, C_{\mathrm{\Theta-dom}})$. For small $\Theta$ ($\Theta \le R \lambda_{\Lambda}$), the system is leakage-dominant, and $C_{\mathrm{SIC}}$ grows linearly with $\Theta$. For larger $\Theta$ ($\Theta > R \lambda_{\Lambda}$), the system becomes rate-dominant, and $C_{\mathrm{SIC}}$ saturates at a constant value determined by the rate budget $R$.

\begin{figure}[!htbp]
	\centering
	 \includegraphics[width=0.8\textwidth]{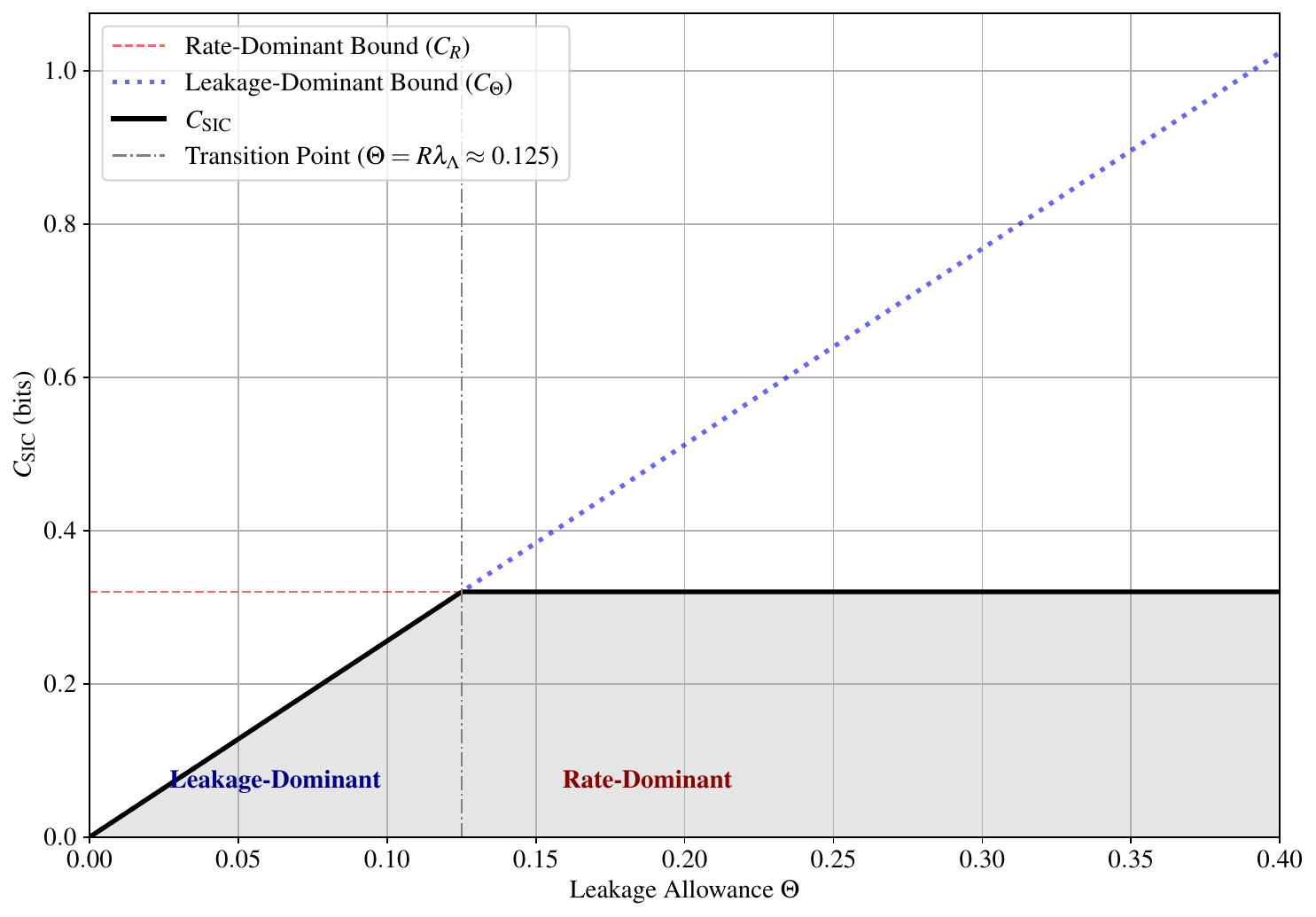} 
	\caption{Operational regimes of $C_{\mathrm{SIC}}$ for the BSWC, as characterized by Theorem \ref{thm_c_lic_regimes}. The solid black line shows $C_{\mathrm{SIC}}$ as a function of the leakage allowance $\Theta$, for fixed $p_{\mathrm{bob}}=0.1, q_{\mathrm{eve}}=0.25,$ and $R=0.5$. The curve shows the clear transition from the leakage-dominant regime to the rate-dominant regime.}
	\label{fig:clic_bswc_regimes_final}
\end{figure}

\subsubsection{The Perfect Secrecy Case} \label{subsubsec:bswc_perfect_secrecy_final_v2}

An important limiting case for evaluating any secrecy framework is that of perfect secrecy, corresponding to $\Theta \to 0$ in the SIC problem. For the EIT approximation, this requires the EIT-approximated leakage, $(\epsilon^2/2)\mathbb{E}_U[||\boldsymbol{B}_{Z|X} \mathbf{L}_U||^2]$, to be zero. For this to hold with $\mathbf{L}_u \neq \mathbf{0}$, the term $||\boldsymbol{B}_{Z|X} \mathbf{L}_u||^2 = \mathbf{L}_u^T \Lambda \mathbf{L}_u$ must be zero for all active perturbations. For the BSWC, where all perturbations are proportional to a single vector $\boldsymbol{\tau}$(Appendix H), this condition reduces to requiring the single relevant eigenvalue of $\Lambda$, $\lambda_{\Lambda} = (1-2q_{\mathrm{eve}})^2$, to be zero. The outcome therefore depends critically on the quality of Eve's channel. 
First, if Eve's channel is not pure noise ($q_{\mathrm{eve}} \neq 0.5$), then its leakage factor $\lambda_{\Lambda}$ is strictly positive. Enforcing perfect local secrecy then requires all perturbations to be zero, leading to zero utility, and thus $C_{\mathrm{SIC}}(\Theta=0) = 0$.

Second, in the special case where Eve's channel is pure noise ($q_{\mathrm{eve}} = 0.5$), we have $\lambda_{\Lambda} = 0$. The leakage constraint is trivially satisfied for any perturbation. The SIC problem with $\Theta=0$ thus reduces to maximizing utility subject only to the rate constraint, which is an EIT-IB problem. This aligns with \cite{SreekumarGunduz19}, whose PUT framework also simplifies under these conditions. While their general approach uses exact mutual information, they perform a local analysis to characterize the utility-leakage trade-off for infinitesimal leakage, defining a slope metric as the supremum of a ratio of KL divergences which is conceptually parallel to the EIT approach. However, for the specific $q_{\mathrm{eve}}=0.5$ case, their perfect privacy constraint $I(S;U)=0$ is automatically satisfied globally, causing their problem to become identical to the standard, globally optimal IB problem.

The comparison in this specific scenario is therefore between the local EIT approximation and the true global optimum of the IB problem in Figure \ref{fig:EIT_IB_vs_BA_final}, which illustrates this comparison. That confirms that the local EIT result correctly captures the initial rate-utility trade-off of the global optimum, with the divergence at higher rates quantifying the inherent difference between a local approximation and a global solution. In the case where $q_{\mathrm{eve}} \neq 0.5$, both our framework and the methodology in \cite{SreekumarGunduz19} correctly predict zero utility, underscoring the fundamental difficulty of achieving perfect secrecy against a competent eavesdropper.

\subsubsection{BSWC KKT Condition Verification}
The BSWC with uniform $P_X$ directly illustrates Theorem \ref{thm:commuting_matrices_unified}, as its EIT matrices commute and the relevant perturbation subspace is one-dimensional. Figure \ref{fig:BSWC_KKT_Verification_final_sec6_full} verifies the KKT condition $\lambda_V = \rho^* + \nu^*\lambda_{\Lambda}$.

\begin{figure}[H]
	\centering
	 \includegraphics[width=0.8\textwidth]{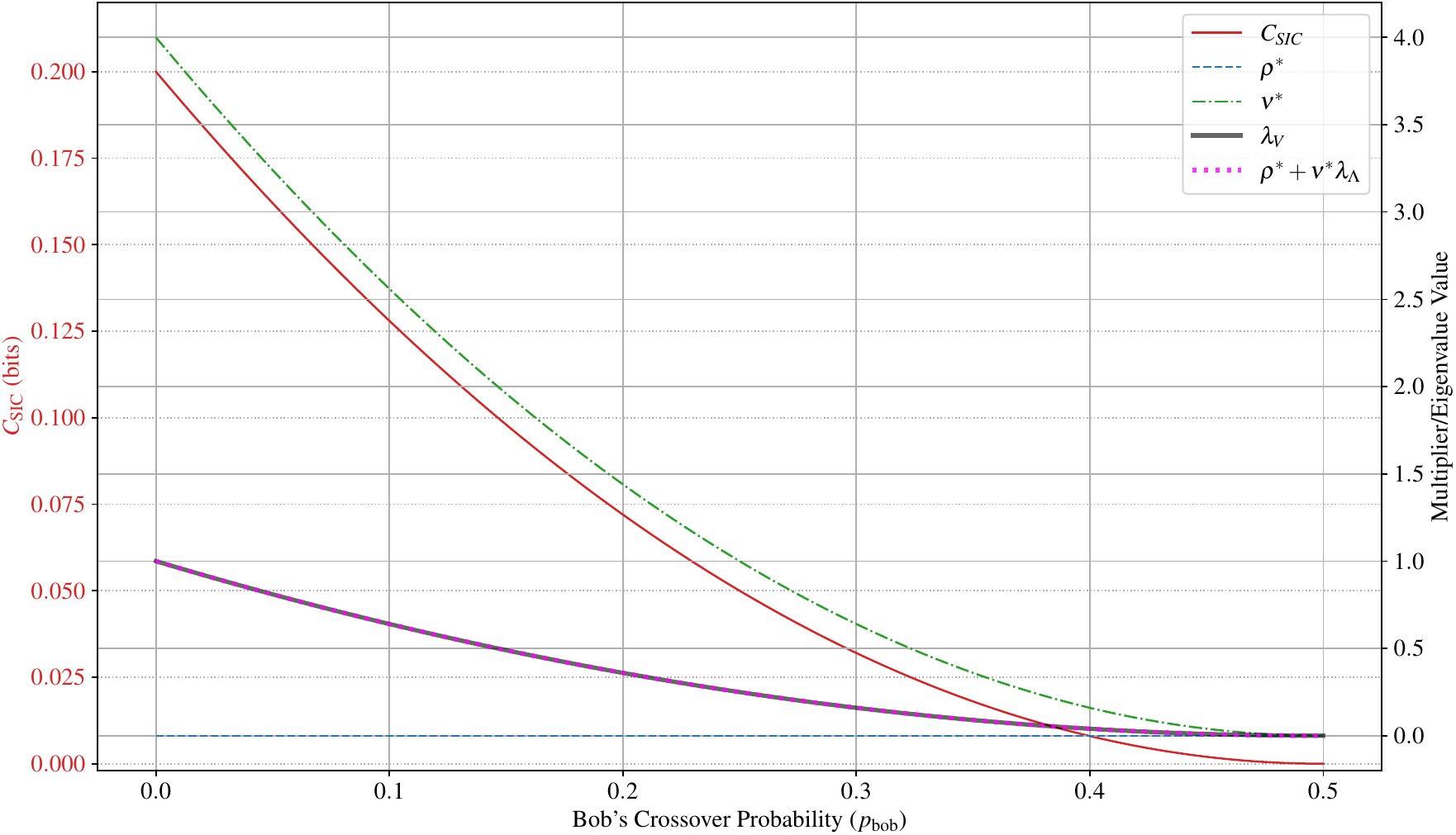}
	\caption{Verification of KKT condition for BSWC, with $q_{{eve}}={0.25}, R={0.5}, \Theta={0.05}$. The left y-axis (red solid) is $C_{\mathrm{SIC}}$; the right y-axis shows $\rho^*$ (blue dashed); $\nu^*$ (green dash-dot); $\lambda_V = (1-2p_{\mathrm{bob}})^2$ (black solid); and $\rho^* + \nu^*\lambda_{\Lambda}$ (magenta dotted).}
	\label{fig:BSWC_KKT_Verification_final_sec6_full}
\end{figure}

Figure \ref{fig:BSWC_KKT_Verification_final_sec6_full} shows $ C_{\mathrm{SIC}} $ (left y-axis) and the components of the KKT condition (right y-axis) as $p_{\mathrm{bob}}$ varies. The effective Lagrange multipliers $\rho^*$ and $\nu^*$ transition based on the active constraints. For $p_{\mathrm{bob}} < 0.18$, the rate constraint is active, so $\rho^* \approx \lambda_V$ and $\nu^* \approx 0$. For $p_{\mathrm{bob}} > 0.18$, the leakage constraint is active, so $\rho^* \approx 0$ and $\nu^* \approx \lambda_V / \lambda_{\Lambda}$. Critically, the plot demonstrates a perfect overlap between $\lambda_V$ (solid black line) and the sum $\rho^* + \nu^*\lambda_{\Lambda}$ (dotted magenta line) across all $p_{\mathrm{bob}}$, numerically confirming that the optimal multipliers adapt to satisfy the condition from Theorem \ref{thm:commuting_matrices_unified}.

The comprehensive set of numerical results presented in this section serves to validate and illustrate the EIT-based framework for local secrecy analysis. The detailed case study on the BSWC provided concrete illustrations for all key aspects of the theory. We showed that the approximate capacity, provides a meaningful estimate for the true secrecy capacity and correctly transitions between rate-dominant and leakage-dominant operational regimes. Furthermore, we explored the impact of system parameters, such as Eve's channel quality, as a function of her SNR and output quantization, and validated the underlying KKT conditions for the commuting case. The analysis of the perfect secrecy limit for the BSWC further clarified the conditions under which non-trivial local utility is achievable. Collectively, these numerical experiments provide strong support for the EIT-SIC framework as a robust and insightful tool for local secrecy analysis.

\section{Secret Local Contraction Coefficients} \label{sec:secret_contraction_coeffs}

Beyond approximating the secrecy capacity, the EIT framework allows for the definition of new coefficients that characterize the efficiency of secure information transfer in the local domain. This section introduces the \textit{secret local contraction coefficient}. We motivate its definition by analogy to standard contraction coefficients \cite{Anathram_13_2, Anantharam_12, Anantharam_13, Polyanskiy_notes, makur2018linear}, provide its mathematical characterization using the EIT-derived matrices, discuss its operational meaning for bounding local utility versus leakage, relate it to the parameters of the SIC problem solution, and finally, establish bounds connecting this local coefficient to its global counterpart defined over exact mutual information measures.

\subsection{Quantifying Secure Information Transfer Efficiency} \label{subsec:contraction_motivation}

In communication theory, contraction coefficients often quantify how much information about a source $ U $, present in an intermediate variable $ X $, is preserved or contracted when $ X $ is processed to produce an output $ Y $. For example, the square of the second largest singular value of the channel's DTM acts as such a coefficient for $ I(U;Y)/I(U;X) $ in EIT \cite{Zheng_2,Makur_1,Makur_20}. 

In the context of secrecy, a pertinent question is: for a given amount of information leaked to an eavesdropper, $ I(U;Z) $, what is the maximum utility, $ I(U;Y) $, that can be reliably conveyed to the legitimate receiver? The ratio $ I(U;Y)/I(U;Z) $ serves as a measure of this leakage efficiency, that is the utility achieved per unit of leakage. This concept is central to PUT problems \cite{geng2020tight} and is related to alternative secrecy metrics like maximal leakage \cite{Khisti19} and maximal correlation secrecy \cite{li2017maximal, tahmasbi2018information}. Our goal is to define and analyze a local version of this efficiency using EIT approximations.

For a standard channel $ U \to X \to Y $, where $ P_{X|U}(x|u) = P_X(x) + \epsilon \sqrt{P_X(x)} L_X(x|u) $, the EIT approximations are $ I(U;X) \approx (\epsilon^2/2) \mathbb{E}_U[||\mathbf{L}_U||^2] $ and $ I(U;Y) \approx (\epsilon^2/2) \mathbb{E}_U [ ||B_{Y|X} \mathbf{L}_U||^2 ] $. The standard local contraction coefficient $ \eta_{\mathrm{loc}} $ is defined as \cite{Zheng_2, Polyanskiy_notes}:
\begin{equation} \label{eq:eta_loc_standard}
	\eta_{\mathrm{loc}} = \sup_{\mathbf{L}: \mathbf{L}^T \mathbf{\sqrt{P_X}}=0, \mathbf{L}\neq\mathbf{0}} \frac{||\boldsymbol{B}_{Y|X} \mathbf{L}||^2}{||\mathbf{L}||^2}.
\end{equation}
This $ \eta_{\mathrm{loc}} $ is equal to $ \sigma_{\max, \perp}^2(B_{Y|X}) $, which denotes the squared largest singular value of $ B_{Y|X} $ when restricted to the subspace orthogonal to $ \mathbf{\sqrt{P_X}} $, often the second largest singular value of the full $ B_{Y|X} $ matrix, as $ \sigma_1=1 $ typically corresponds to $ \mathbf{\sqrt{P_X}} $.

\subsection{Definition and Characterization of the Secret Local Contraction Coefficient} \label{subsec:define_eta_loc_sec_final}

We extend this concept to the wiretap channel $ U \to X \to (Y,Z) $. Using the EIT-approximated utility $ I(U;Y)$ and leakage $ I(U;Z)$, given in \eqref{eq:slic_obj_approx_final} and \eqref{eq:slic_leakage_approx_final_with Lambda}, respectively.
To quantify the intrinsic efficiency of secure transmission in this local regime, we define a coefficient based on the ratio of these quadratic forms for an effective single perturbation vector $ \mathbf{L} $. The term $ \mathbf{L}^T V \mathbf{L} $ can be interpreted as the signal power for Bob due to the input perturbation $ \mathbf{L} $, while $ \mathbf{L}^T \Lambda \mathbf{L} $ represents the corresponding leakage power or information acquired by Eve from the same perturbation. The ratio $ (\mathbf{L}^T V \mathbf{L}) / (\mathbf{L}^T \Lambda \mathbf{L}) $ thus measures a form of signal-to-leakage power ratio for the specific perturbation direction. We seek the perturbation direction that maximizes this ratio, representing the most leakage-efficient way to convey information.
\begin{defn} \label{def:secret_local_contraction_coefficient_final}
	The secret local contraction coefficient is defined as the maximum ratio of the EIT-approximated utility to the EIT-approximated leakage, over all valid non-trivial perturbation vectors $ \mathbf{L} $:
	\begin{equation} \label{eq:eta_loc_sec_definition_final}
		\eta_{\mathrm{loc}}^{\mathrm{sec}} = \sup_{\mathbf{L}} \frac{\mathbf{L}^T V \mathbf{L}}{\mathbf{L}^T \Lambda \mathbf{L}}
	\end{equation}
	subject to $ \mathbf{L}^T \mathbf{\sqrt{P_X}} = 0 $, $ \mathbf{L} \neq \mathbf{0} $, and $ \mathbf{L}^T \Lambda \mathbf{L} > 0 $. An implicit constraint $ \mathbf{L}^T \Lambda \mathbf{L} \le \delta' $, for some small $ \delta' > 0 $, can be considered to ensure the perturbation remains within the local validity regime, though the ratio itself is scale-invariant.
\end{defn}
\begin{thm} \label{thm:char_eta_loc_sec_final}
	The secret local contraction coefficient is equal to the largest generalized eigenvalue, of the matrix pencil $ (V, \Lambda) $, when restricted to the subspace of perturbation vectors $ \mathbf{L} $ satisfying $ \mathbf{L}^T \mathbf{\sqrt{P_X}} = 0 $, that is
	\begin{equation} \label{eq:eta_loc_sec_is_lambda_max_final}
		\eta_{\mathrm{loc}}^{\mathrm{sec}} = d_{\max}^{\perp}(V, \Lambda)
	\end{equation}
	where $ d_{\max}^{\perp}(V, \Lambda) $ denotes the largest generalized eigenvalue.
\end{thm}
\begin{proof}
	The proof is presented in Appendix J.
\end{proof}
Intuitively, $ \eta_{\mathrm{loc}}^{\mathrm{sec}} $ identifies the optimal EIT perturbation direction $ \mathbf{L} $ that maximizes Bob's information gain relative to Eve's. Thus, is an intrinsic characteristic of the EIT representations of channels $V$ and $\Lambda$, reflecting the optimal local trade-off achievable.
\begin{rem} \label{rem:clic_eta_connection}
	 Given \eqref{eq:eta_loc_sec_is_lambda_max_final} the leakage-dominant capacity in \eqref{eq:clic_max_over_regimes_v2} can be written as:
	\begin{equation} \label{eq:clic_with_eta_loc_sec}
		C_{\mathrm{SIC}} = \max \left( \lambda_{\max}^{\perp}(V) R, \quad \eta_{\mathrm{loc}}^{\mathrm{sec}} \Theta, \quad C_{\mathrm{inter}} \right).
	\end{equation}
\end{rem}

\subsection{Operational Interpretation and Relation to the SIC Problem Solution} \label{subsec:eta_loc_sec_operational_final}
Having defined and characterized the $ \eta_{\mathrm{loc}}^{\mathrm{sec}} $, we now explore its operational meaning and its connection to the solution of the EIT-approximated SIC problem. While $ \eta_{\mathrm{loc}}^{\mathrm{sec}} $ is defined as the supremum of a ratio of quadratic forms, its true utility lies in its ability to provide a concrete bound on the achievable performance and to shed light on the structure of the optimal trade-off. The following lemma establishes a direct operational interpretation of $ \eta_{\mathrm{loc}}^{\mathrm{sec}} $ as a bounding constant that links the achievable local utility to the incurred local leakage for any valid EIT perturbation strategy.
\begin{lem} \label{lem:eta_loc_sec_bound_final}
	For any perturbation vector $ \mathbf{L} $ satisfying \eqref{eq:eit_slic_opt_ortho_problem_env} and resulting in $ I(U;Z) = (\epsilon^2/2)\mathbf{L}^T\Lambda\mathbf{L} > 0 $, the approximated utility for Bob is bounded:
	\begin{equation} \label{eq:utility_leakage_bound_eta_loc_sec_final}
		I(U;Y) \le \eta_{\mathrm{loc}}^{\mathrm{sec}} \cdot I(U;Z).
	\end{equation}
	Equality holds if $ \mathbf{L} $ aligns with the principal generalized eigenvector of $ (V, \Lambda) $ within the valid subspace.
\end{lem}
\begin{proof}
	The proof is presented in Appendix K.
\end{proof}
The bound in \eqref{eq:utility_leakage_bound_eta_loc_sec_final} is tight, as it is achieved by the perturbation direction corresponding to the principal generalized eigenvector. Figures \ref{fig:lemma71_scatter_general_channel} and \ref{fig:thm71_histogram_general_channel} illustrate Lemma \ref{lem:eta_loc_sec_bound_final} and Theorem \ref{thm:char_eta_loc_sec_final} respectively.
\begin{figure}[H]
	\centering
	\includegraphics[width=0.8\textwidth]{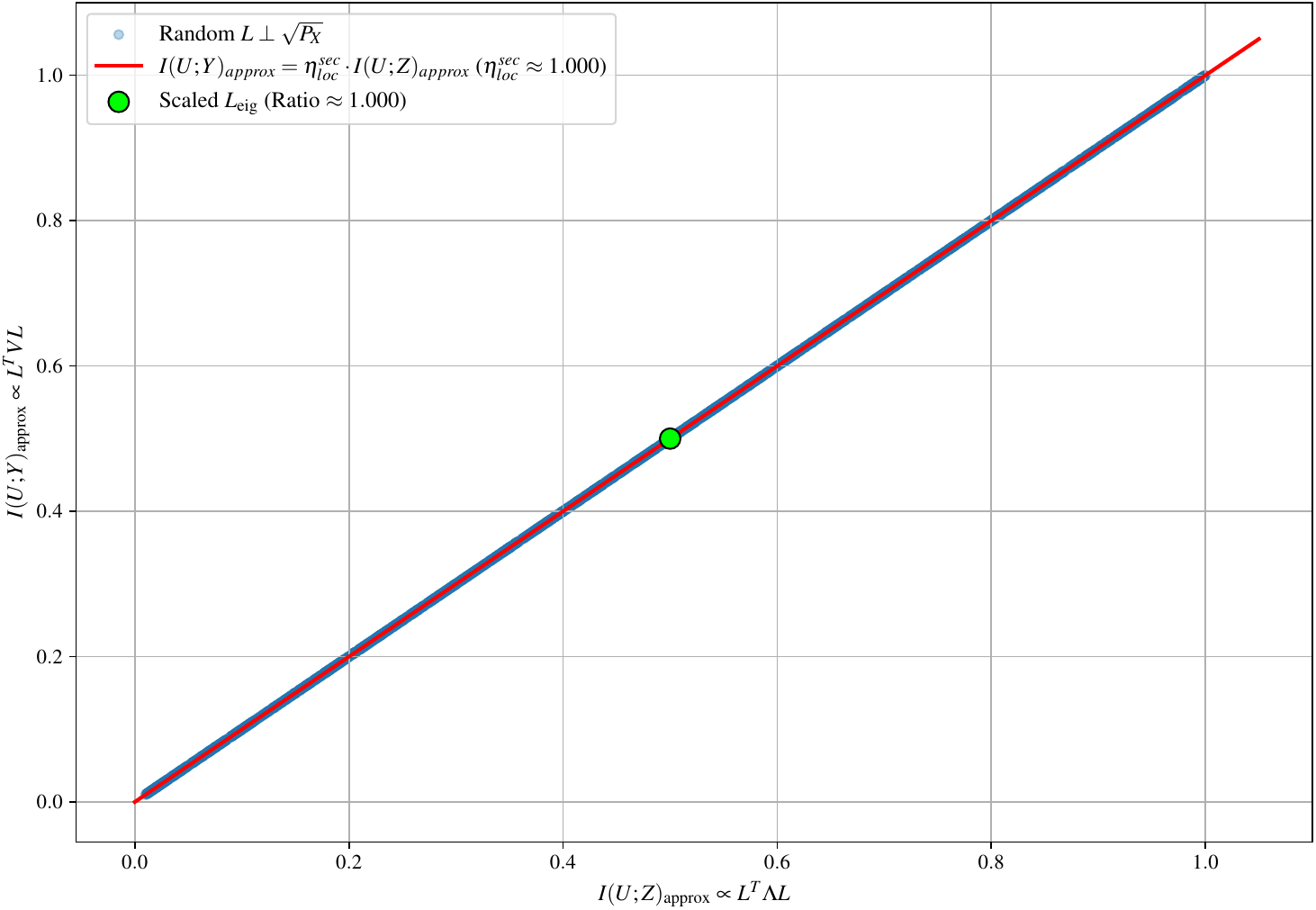} 
	\caption{Utility vs. Leakage for general DM-WTC with $|\mathcal{X}|=3, P_X = [0.4, 0.3, 0.3]$. The solid red line represents the bound $I(U;Y) = \eta_{\mathrm{loc}}^{\mathrm{sec}} \cdot I(U;Z)$, with the numerically calculated $\eta_{\mathrm{loc}}^{\mathrm{sec}} \approx 1.000$ for this example. The green circle indicates the point achieved by a scaled version of the principal generalized eigenvector $L_{\mathrm{eig}}$.}
	\label{fig:lemma71_scatter_general_channel}
\end{figure}
\begin{figure}[H]
	\centering
	 \includegraphics[width=0.8\textwidth]{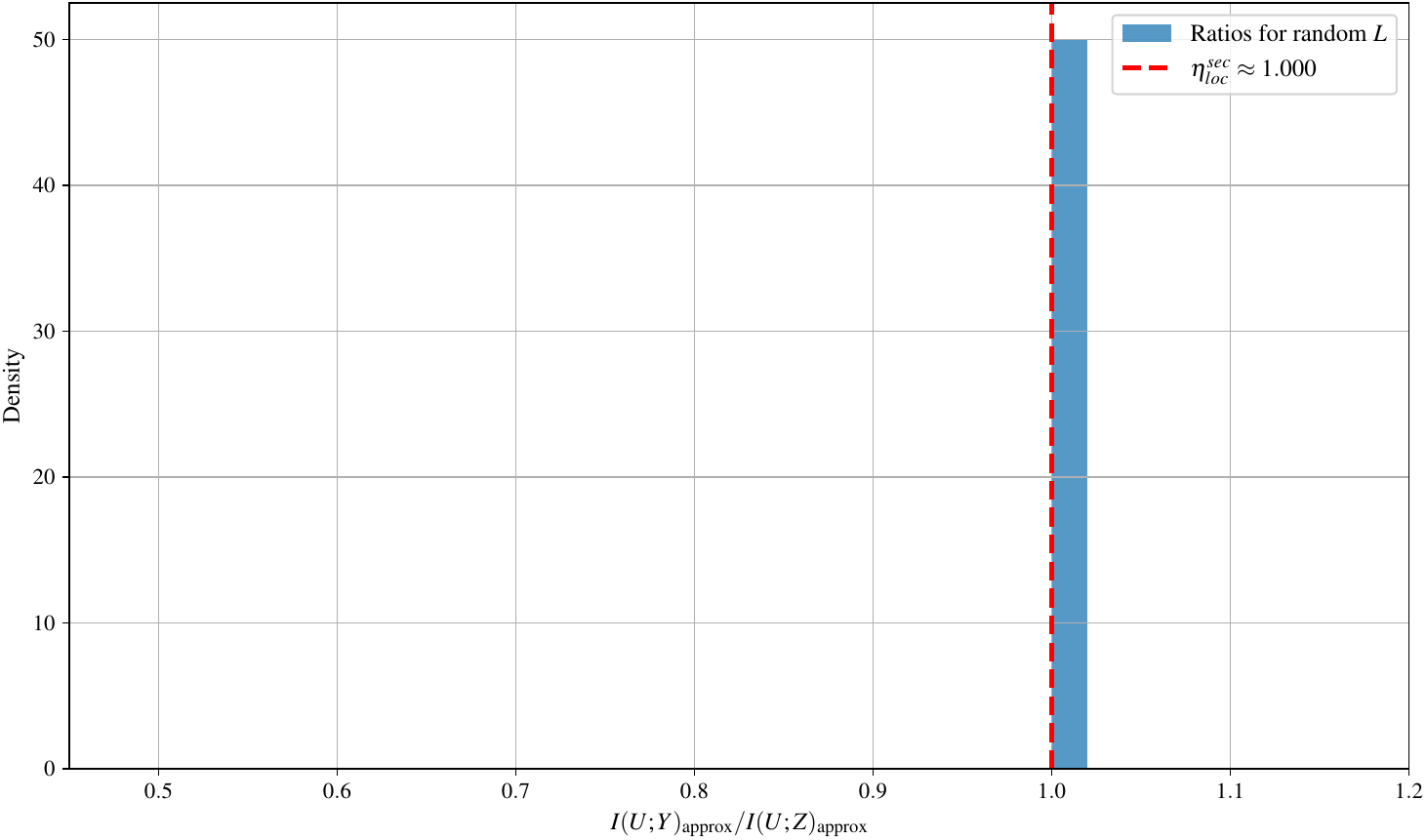}
	\caption{Distribution of Utility/Leakage Ratios with $|\mathcal{X}|=3, P_X = [0.4, 0.3, 0.3]$. The dashed red line indicates the calculated $ \eta_{\mathrm{loc}}^{\mathrm{sec}} \approx 1.000 $, which is the supremum of these ratios.}
	\label{fig:thm71_histogram_general_channel}
\end{figure}
The scatter plot in Figure \ref{fig:lemma71_scatter_general_channel} shows that for numerous random perturbation vectors $ \mathbf{L} $, the resulting pairs $ (I(U;Z), I(U;Y)) $ all lie on or below the line defined by $ \eta_{\mathrm{loc}}^{\mathrm{sec}} $. In this particular example, the channel parameters were chosen such that $V_{\mathrm{proj}} \approx \Lambda_{\mathrm{proj}}$ in the EIT perturbation subspace, leading to $ \eta_{\mathrm{loc}}^{\mathrm{sec}} \approx 1.0 $. Consequently, all valid perturbations yield approximately the same utility-to-leakage ratio, hence the points closely follow the bound line. The histogram in Figure \ref{fig:thm71_histogram_general_channel} confirms this, with the distribution of ratios sharply peaked at $ \eta_{\mathrm{loc}}^{\mathrm{sec}} \approx 1.0 $. These figures effectively demonstrate the operational meaning of $ \eta_{\mathrm{loc}}^{\mathrm{sec}} $ as the maximum local leakage efficiency.

While $ \eta_{\mathrm{loc}}^{\mathrm{sec}} $ represents the intrinsic leakage efficiency of the channel, the actual utility-to-leakage ratio achieved by the solution to the full SIC problem is also influenced by the rate constraint $R$. The following theorem connects the parameters of the SIC problem's optimal solution, $ \rho^* $ and $ \nu^* $, to the ratio achieved by the optimal perturbation vector $ \mathbf{L}^* $.
\begin{thm} \label{thm:slic_achieved_ratio_final}
	Let $ \mathbf{L}^* $ be an optimal perturbation vector representative of the solution to Problem \ref{prob:eit_slic_approximated}, associated with optimal $ \nu^* $ and $ \rho^* $. The ratio of approximated utility to leakage by $ \mathbf{L}^* $ is:
	\begin{equation} \label{eq:slic_ratio_formula_final}
		\frac{(\mathbf{L}^*)^T V \mathbf{L}^*}{(\mathbf{L}^*)^T \Lambda \mathbf{L}^*} = \nu^* + \rho^* \frac{(\mathbf{L}^*)^T I \mathbf{L}^*}{(\mathbf{L}^*)^T \Lambda \mathbf{L}^*},
	\end{equation}
	where $(\mathbf{L}^*)^T \Lambda \mathbf{L}^*> 0$. This achieved ratio is generally $ \le \eta_{\mathrm{loc}}^{\mathrm{sec}} $ due to the additional rate constraint $R'$ in the SIC problem.
\end{thm}
\begin{proof}
	The proof is presented in Appendix L.
\end{proof}
Theorem \ref{thm:slic_achieved_ratio_final} provides an insight into the performance of the EIT-approximated SIC problem by distinguishing between the intrinsic leakage efficiency of the channel and the actual efficiency achieved by the optimal solution. The achieved ratio, given by \eqref{eq:slic_ratio_formula_final}, is a direct consequence of the KKT conditions and is determined by the optimal Lagrange multipliers $ \rho^* $ and $ \nu^* $. The values of these multipliers are, in turn, determined by the solution to the LP in Theorem \ref{thm:general_lp_for_multipliers}, which leads to the operational regimes characterized in Theorem \ref{thm_c_lic_regimes}.

\begin{figure}[H]
	\centering
	\includegraphics[width=0.8\columnwidth]{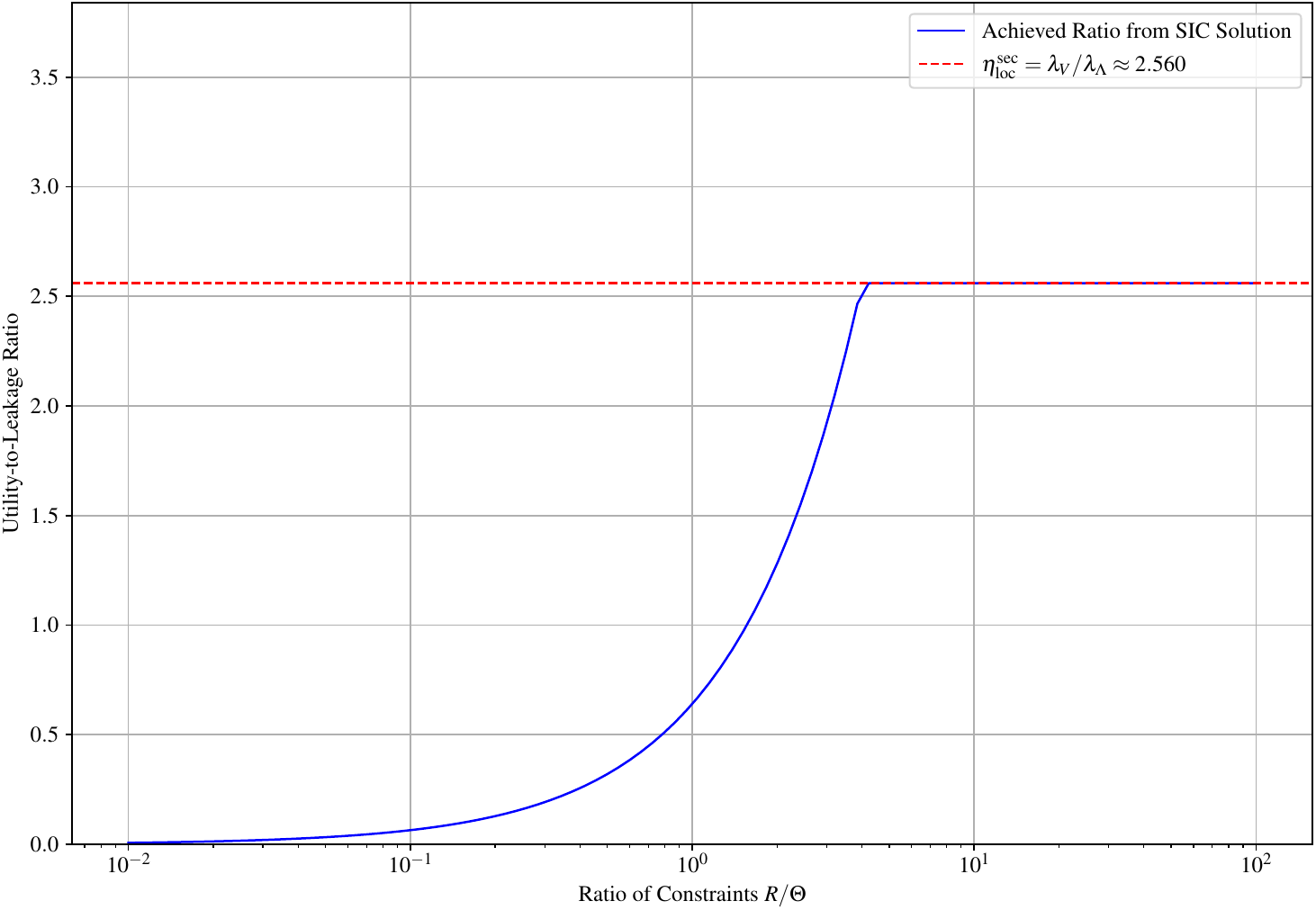} 
	\caption{Illustration of the achieved utility-to-leakage ratio from the SIC solution (solid blue) versus the intrinsic channel efficiency, $\eta_{\mathrm{loc}}^{\mathrm{sec}}$ (dashed red), for a BSWC. The x-axis represents the ratio of the rate budget to the leakage budget. For large $R/\Theta$, the rate constraint is inactive, and the achieved ratio converges to the channel's intrinsic limit. For small $R/\Theta$, the achieved efficiency is strictly less than the intrinsic limit, visualizing the trade-off described in Theorem \ref{thm:slic_achieved_ratio_final}.}
	\label{fig:achieved_ratio}
\end{figure}

In the leakage-dominant regime, the LP solution yields $ \rho^* = 0 $ and $ \nu^* = \eta_{\mathrm{loc}}^{\mathrm{sec}} $. Substituting these into \eqref{eq:slic_ratio_formula_final}, the achieved ratio becomes exactly $ \eta_{\mathrm{loc}}^{\mathrm{sec}} $. This means that when the rate budget $R$ is effectively infinite, the SIC problem's only goal is to maximize utility per unit of leakage, and it does so by choosing a perturbation strategy that achieves the intrinsic maximum channel efficiency. In the rate-dominant regime, the LP solution yields $ \nu^* = 0 $ and $ \rho^* = \lambda_{\max}^{\perp}(V) $. The achieved ratio from \eqref{eq:slic_ratio_formula_final} becomes $ \lambda_{\max}^{\perp}(V) \frac{||\mathbf{L}^*||^2}{(\mathbf{L}^*)^T \Lambda \mathbf{L}^*} $. The optimal perturbation $ \mathbf{L}^* $ in this case is the eigenvector corresponding to $ \lambda_{\max}^{\perp}(V) $, so the achieved ratio is determined by the properties of this specific eigenvector. Finally, in the intermediate regime, both $ \rho^* > 0 $ and $ \nu^* > 0 $. The achieved ratio is given by the full expression in \eqref{eq:slic_ratio_formula_final}. It is this regime that most clearly shows the trade-off: the achieved efficiency is a blend of the direct leakage efficiency priced by $ \nu^* $ and the rate-driven efficiency priced by $ \rho^* $.

In all cases where $ \rho^* > 0 $, the achieved ratio is generally less than or equal to $ \eta_{\mathrm{loc}}^{\mathrm{sec}} $. This inequality arises precisely because the additional rate constraint $R'$ forces the SIC solution to potentially choose a perturbation direction $ \mathbf{L}^* $ that is not perfectly aligned with the principal generalized eigenvector of $(V, \Lambda)$ if that optimal direction is too costly in terms of its squared Euclidean norm, $ ||\mathbf{L}^*||^2 $. This is illustrated in Figure \ref{fig:achieved_ratio}.

\subsection{Bounds Connecting Global and Local Secret Contraction Coefficients} \label{subsec:global_local_bounds_contraction_final}
While the secret local contraction coefficient in \eqref{eq:eta_loc_sec_definition_final} provides a tractable, EIT-derived measure of leakage efficiency, it is crucial to understand its relationship to the true efficiency defined by exact mutual information terms. To this end, we define the global secret contraction coefficient which is the supremum of the ratio of true mutual information terms:
\begin{equation} \label{eq:eta_glo_sec_definition_sec7_v2}
	\eta_{\mathrm{glo}}^{\mathrm{sec}} \triangleq \sup_{P_U, P_{X|U}} \frac{I(U;Y)}{I(U;Z)},
\end{equation}
subject to the constraints of Problem 1, including $ 0 < I(U;Z) \le \Theta $ and $ I(U;X) \le R $.
The following lemma establishes a fundamental two-sided bound connecting this global, often intractable, quantity to the analytically derived local coefficient. This connection is analogous in spirit to bounds established between mutual information and other statistical measures of dependence, such as maximal correlation \cite{Asoodeh, asoodeh2016information}.
\begin{lem} \label{lem:glo_ge_loc_sec_final}
	The global and local secret contraction coefficients are related by the following inequalities:
	\begin{equation}		
	\eta_{\mathrm{loc}}^{\mathrm{sec}} \le \eta_{\mathrm{glo}}^{\mathrm{sec}} \le \frac{2}{P_{\mathrm{min}}} \eta_{\mathrm{loc}}^{\mathrm{sec}},
\end{equation}
where $ P_{\mathrm{min}} = \min_{x \in \mathrm{supp}(P_X)} P_X(x) > 0 $.
\end{lem}
\begin{proof}
The proof is presented in Appendix M. 
\end{proof}
These bounds connect the EIT-derived local coefficient to its global counterpart. Figure \ref{fig:BSWC_contraction_bounds_final} illustrates these bounds for the BSWC. This figure effectively demonstrates the utility of $ \eta_{\mathrm{loc}}^{\mathrm{sec}} $ not only as a standalone measure of local leakage efficiency but also as an analytical tool to constrain the range of its more general, but harder to compute, global counterpart $ \eta_{\mathrm{glo}}^{\mathrm{sec}} $. The observed behaviors across different channel conditions for Bob and Eve align with information-theoretic intuition regarding the principles of secure communication.
Figure \ref{fig:BSWC_contraction_bounds_final} illustrates the bounds on $ \eta_{\mathrm{glo}}^{\mathrm{sec}} $ for the BSWC, as established in Lemma \ref{lem:glo_ge_loc_sec_final}. The simulations assume a uniform input distribution, for which $ P_{\mathrm{min}}=0.5 $, yielding an upper bound of $ 4 \eta_{\mathrm{loc}}^{\mathrm{sec}} $. The solid lines represent the lower bound, $ \eta_{\mathrm{loc}}^{\mathrm{sec}} $, which is calculated as $ (1-2p_{\mathrm{bob}})^2 / (1-2q_{\mathrm{eve}})^2 $ for the BSWC. The behavior of these bounds as Bob's and Eve's channel qualities vary reveals several key insights.
\begin{figure}[H]
	\centering
	\includegraphics[width=0.8\textwidth]{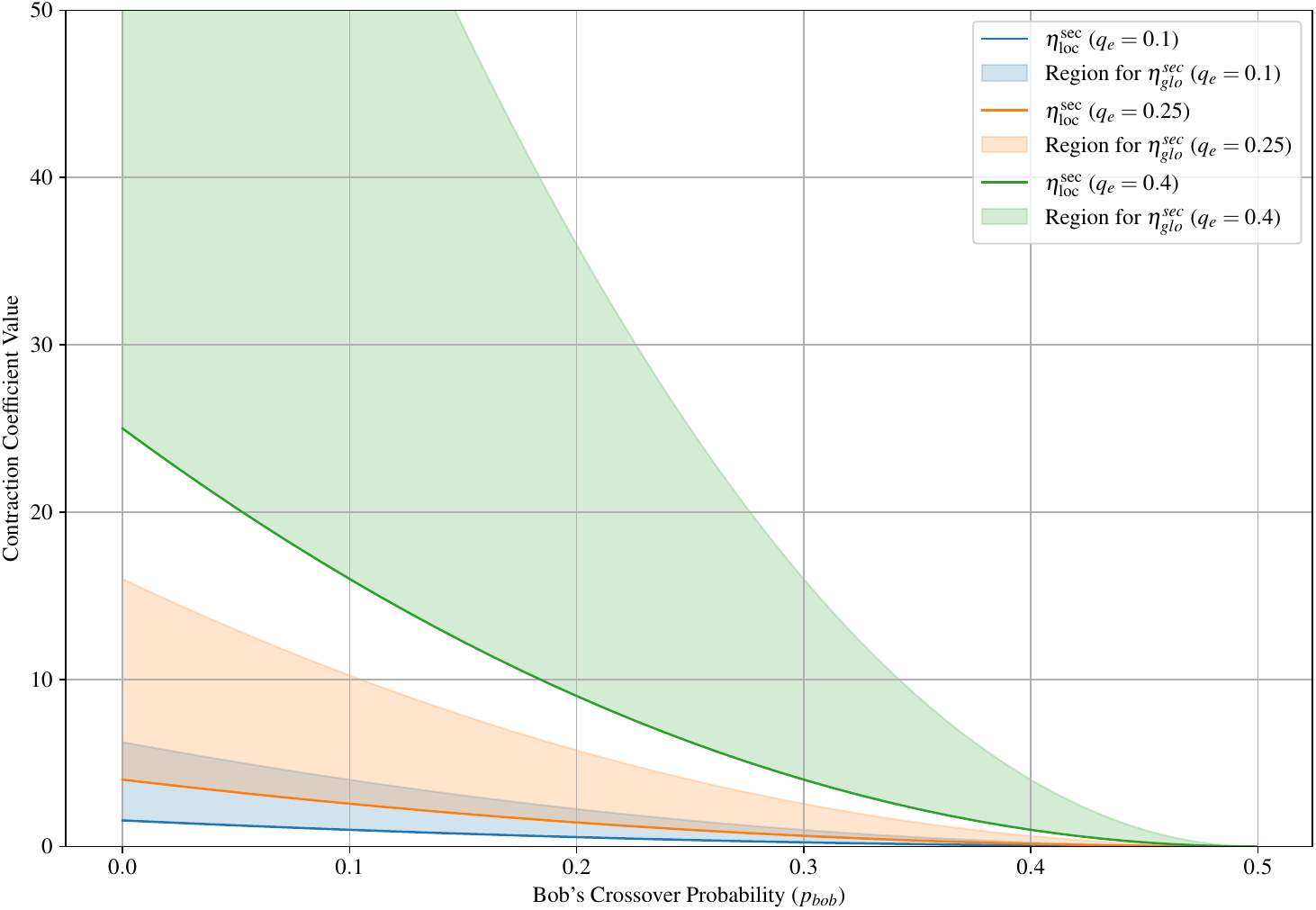} 
	\caption{Bounds on the global secret contraction coefficient ($\eta_{\mathrm{glo}}^{\mathrm{sec}}$) for the BSWC as a function of Bob's crossover probability ($p_{\mathrm{bob}}$). Input $P_X$ is uniform ($P_{\mathrm{min}} = 0.5$). Solid lines are $\eta_{\mathrm{loc}}^{\mathrm{sec}}$ (lower bound) for $q_e = [0.1, 0.25, 0.4]$. Shaded regions extend to the upper bound $4 \eta_{\mathrm{loc}}^{\mathrm{sec}}$.}
	\label{fig:BSWC_contraction_bounds_final}
\end{figure}
First, the lower bound $ \eta_{\mathrm{loc}}^{\mathrm{sec}} $ behaves intuitively: for a fixed Eve, it is maximized when Bob's channel is perfect ($ p_{\mathrm{bob}}=0 $) and decreases as Bob's channel degrades. Conversely, for a fixed Bob, $ \eta_{\mathrm{loc}}^{\mathrm{sec}} $ increases as Eve's channel becomes noisier, indicating greater potential for leakage-efficient communication. Second, the shaded regions depict the guaranteed interval where the true global coefficient $ \eta_{\mathrm{glo}}^{\mathrm{sec}} $ must lie. This visually represents the established connection between the analytically tractable local coefficient and the often intractable global one. The multiplicative gap between the bounds, a factor of $ (2/P_{\mathrm{min}}) - 1 = 3 $ in this example, indicates the potential difference between the EIT local approximation's efficiency measure and the true global efficiency. Furthermore, the figure demonstrates the impact of the eavesdropper's channel quality. When Eve's channel is strong (blue curve), the entire possible range for $ \eta_{\mathrm{glo}}^{\mathrm{sec}} $ is constrained to lower values. As Eve's channel degrades (green curve), the potential for high utility-to-leakage ratios, both locally and globally, increases significantly. These bounds collectively provide valuable, albeit not always tight, constraints on the performance of any possible secure communication scheme as measured by $ \eta_{\mathrm{glo}}^{\mathrm{sec}} $. 

\section{Conclusion and Future Work} \label{sec:conclusion}
In this paper, we introduced a framework for Local Information-Theoretic Security, moving beyond the classical asymptotic analysis of the wiretap channel to a local regime characterized by small perturbations around a fixed operating point. By applying the framework of EIT, we transformed the non-convex SIC problem into an analytically tractable domain. This approach yielded several contributions that constitute a new perspective on the structure of secure communication.

Our central discovery is that the local secrecy capacity is governed by an LP, whose constraints are determined by the generalized eigenvalues of the channel matrices. This LP formulation provides a new geometric perspective on secrecy, revealing that local performance trade-offs are dictated by the spectral theory of the channel matrices. This framework is constructive. It directly led to an explicit analytical formula for the local secrecy capacity, expressing it as a weighted sum of the rate and leakage budgets, where the weights are the optimal dual variables of the LP.

Building on this, we fully characterized the behavior of the capacity across distinct operational regimes by analyzing the vertex solutions of the LP. This analysis provides clear design principles by showing how the optimal signaling strategy shifts depending on whether the system is constrained by encoding resources or by the secrecy requirement. A key consequence of this analysis is the finding that the local capacity is a universal characteristic of the channel, independent of the secret message's source statistics.

Furthermore, our method enabled the introduction of the secret local contraction coefficient, a new computable metric defined as the largest generalized eigenvalue of a channel-derived matrix pencil. This coefficient quantifies the channel's intrinsic leakage efficiency and was shown to be linked to the solution in the leakage-dominant regime. The validity of all our theoretical results was demonstrated through numerical experiments, confirming the accuracy of the EIT approximation and the correctness of the LP formulation for general channels.

The primary advantage of this work is its provision of analytical tractability for a problem that is often opaque. The connection between local secrecy performance and spectral theory offers a new paradigm for the design and evaluation of secure systems. A compelling direction for future research is to formally bridge the gap between our local, geometric analysis and the non-asymptotic guarantees of finite-blocklength information theory.

		
		%
		%
	\appendices
	
\section*{Appendix A: Proofs for Mutual Information Approximation Lemmas}		
\label{app:proofs_mi_lemmas}

\subsection{Proof of Lemma \ref{lem:approx_IUX_formal} (Approximation of $I(U;X)$)}
\begin{proof}
	The mutual information $ I(U;X) $ can be expressed as:
	\begin{equation} \label{eq:proof_iux_def_app}
		I(U;X) = \sum_{u \in \mathcal{U}} P_U(u) D_{\mathrm{KL}}(P_{X|U=u} || P_X).
	\end{equation}
	We are given the EIT perturbation for the conditional distribution $ P_{X|U=u}(x) $ (denoted as $ Q_X^{(u)}(x) $ for clarity in this proof step) around the reference marginal distribution $ P_X(x) $:
	$ Q_X^{(u)}(x) \triangleq P_{X|U=u}(x) = P_X(x) + \epsilon \sqrt{P_X(x)} L_X(x|u) $.
	Let $ J_X(x|u) = \sqrt{P_X(x)} L_X(x|u) $. Then $ Q_X^{(u)}(x) = P_X(x) + \epsilon J_X(x|u) $.
	
	From \eqref{eq:eit_kl_approx_J}, the KL divergence $ D_{\mathrm{KL}}(Q_X^{(u)} || P_X) $ is approximated by a second-order Taylor expansion around $ \epsilon = 0 $.
	Let $ f(\epsilon) = D_{\mathrm{KL}}(P_X + \epsilon J_X(\cdot|u) || P_X) = \sum_x (P_X(x) + \epsilon J_X(x|u)) \log \frac{P_X(x) + \epsilon J_X(x|u)}{P_X(x)} $.
	Then,
	\begin{equation*}
		f(\epsilon) = \sum_x (P_X(x) + \epsilon J_X(x|u)) \log \left(1 + \epsilon \frac{J_X(x|u)}{P_X(x)}\right).
	\end{equation*}
	We use the Taylor expansion $ \log(1+z) = z - \frac{z^2}{2} + O(z^3) $. Let $ z_x = \epsilon \frac{J_X(x|u)}{P_X(x)} $.
	\begin{align*}
		f(\epsilon) &= \sum_x (P_X(x) + \epsilon J_X(x|u)) \left( \epsilon \frac{J_X(x|u)}{P_X(x)} - \frac{\epsilon^2}{2} \left(\frac{J_X(x|u)}{P_X(x)}\right)^2 + O(\epsilon^3) \right) \\
		&= \sum_x \left[ P_X(x) \left( \epsilon \frac{J_X(x|u)}{P_X(x)} - \frac{\epsilon^2}{2} \frac{J_X(x|u)^2}{P_X(x)^2} \right) + \epsilon J_X(x|u) \left( \epsilon \frac{J_X(x|u)}{P_X(x)} \right) \right] + O(\epsilon^3) \\
		&= \sum_x \left[ \epsilon J_X(x|u) - \frac{\epsilon^2}{2} \frac{J_X(x|u)^2}{P_X(x)} + \epsilon^2 \frac{J_X(x|u)^2}{P_X(x)} \right] + O(\epsilon^3) \\
		&= \epsilon \sum_x J_X(x|u) + \frac{\epsilon^2}{2} \sum_x \frac{J_X(x|u)^2}{P_X(x)} + O(\epsilon^3).
	\end{align*}
	Since $ P_{X|U=u}(x) $ is a valid PMF for each $u$, $ \sum_x P_{X|U=u}(x) = 1 $. Given $ \sum_x P_X(x) = 1 $, it implies $ \sum_x \epsilon J_X(x|u) = 0 $, so $ \sum_x J_X(x|u) = 0 $.
	Therefore, the first term $ \epsilon \sum_x J_X(x|u) $ vanishes.
	So, for each $u$:
	\begin{equation} \label{eq:proof_dkl_pxu_px_app}
		D_{\mathrm{KL}}(P_{X|U=u} || P_X) = \frac{\epsilon^2}{2} \sum_{x \in \mathcal{X}} \frac{J_X(x|u)^2}{P_X(x)} + O(\epsilon^3).
	\end{equation}
	Substituting $ J_X(x|u) = \sqrt{P_X(x)} L_X(x|u) $:
	\begin{equation*}
		\frac{J_X(x|u)^2}{P_X(x)} = \frac{(\sqrt{P_X(x)} L_X(x|u))^2}{P_X(x)} = \frac{P_X(x) L_X(x|u)^2}{P_X(x)} = L_X(x|u)^2.
	\end{equation*}
	Thus,
	\begin{equation*}
		D_{\mathrm{KL}}(P_{X|U=u} || P_X) = \frac{\epsilon^2}{2} \sum_{x \in \mathcal{X}} L_X(x|u)^2 + O(\epsilon^3) = \frac{\epsilon^2}{2} ||\mathbf{L}_u||^2 + O(\epsilon^3).
	\end{equation*}
	Substituting this back into the definition of $ I(U;X) $ from \eqref{eq:proof_iux_def_app}:
	\begin{align*}
		I(U;X) &= \sum_{u \in \mathcal{U}} P_U(u) \left( \frac{\epsilon^2}{2} ||\mathbf{L}_u||^2 + O(\epsilon^3) \right) \\
		&\approx \frac{\epsilon^2}{2} \sum_{u \in \mathcal{U}} P_U(u) ||\mathbf{L}_u||^2 \\
		&= \frac{\epsilon^2}{2} \mathbb{E}_U[||\mathbf{L}_U||^2].
	\end{align*}
	This completes the proof, neglecting terms of order $ O(\epsilon^3) $ and higher.
\end{proof}

\subsection{Proof of Lemma \ref{lem:approx_IUY_formal} (Approximation of $I(U;Y)$)}
\begin{proof}
	The mutual information $ I(U;Y) $ can be expressed using the conditional output distributions $ P_{Y|U=u}(y) $ and the marginal output distribution $ P_Y(y) $:
	\begin{equation} \label{eq:proof_iuy_def_app}
		I(U;Y) = \sum_{u \in \mathcal{U}} P_U(u) D_{\mathrm{KL}}(P_{Y|U=u} || P_Y).
	\end{equation}
	The reference marginal output distribution is $ P_Y(y) = \sum_x P_X(x) P_{Y|X}(y|x) $.
	The output distribution is $ P_{Y|U=u}(y) = \sum_x P_{X|U}(x|u) P_{Y|X}(y|x) $.
	We are given the EIT perturbation for the input: $ P_{X|U}(x|u) = P_X(x) + \epsilon \sqrt{P_X(x)} L_X(x|u) $.
	Let $ J_X(x|u) = \sqrt{P_X(x)} L_X(x|u) $. So $ P_{X|U}(x|u) = P_X(x) + \epsilon J_X(x|u) $.
	Then, the output distribution can be written as:
	\begin{align*}
		P_{Y|U=u}(y) &= \sum_x (P_X(x) + \epsilon J_X(x|u)) P_{Y|X}(y|x) \\
		&= \sum_x P_X(x) P_{Y|X}(y|x) + \epsilon \sum_x J_X(x|u) P_{Y|X}(y|x) \\
		&= P_Y(y) + \epsilon J_Y(y|u),
	\end{align*}
	where $ J_Y(y|u) \triangleq \sum_x J_X(x|u) P_{Y|X}(y|x) $ is the unscaled perturbation of the output distribution $P_Y(y)$ when $U=u$.
	
	Let $ \mathbf{L}_u $ be the vector form of $ L_X(x|u) $. From the definition of the DTM $ B_{Y|X} $, the scaled perturbation vector $ \mathbf{L}_{Y,u} $ for the output distribution $ P_{Y|U=u} $ (relative to $ P_Y $) is given by $ \mathbf{L}_{Y,u} = B_{Y|X} \mathbf{L}_u $.
	This means $ P_{Y|U=u}(y) $ can be written in the EIT form as $ P_{Y|U=u}(y) = P_Y(y) + \epsilon \sqrt{P_Y(y)} L_Y(y|u) $, where $ L_Y(y|u) $ are the components of $ \mathbf{L}_{Y,u} $.
	
	Now, we apply the KL divergence approximation in \eqref{eq:eit_kl_approx_L} to $ D_{\mathrm{KL}}(P_{Y|U=u} || P_Y) $:
	\begin{equation*}
		D_{\mathrm{KL}}(P_{Y|U=u} || P_Y) \approx \frac{\epsilon^2}{2} ||\mathbf{L}_{Y,u}||^2.
	\end{equation*}
	Substituting $ \mathbf{L}_{Y,u} = B_{Y|X} \mathbf{L}_u $:
	\begin{equation*}
		D_{\mathrm{KL}}(P_{Y|U=u} || P_Y) \approx \frac{\epsilon^2}{2} ||B_{Y|X} \mathbf{L}_u||^2.
	\end{equation*}
	Substituting this into the definition of $ I(U;Y) $ from \eqref{eq:proof_iuy_def_app}:
	\begin{align*}
		I(U;Y) &= \sum_{u \in \mathcal{U}} P_U(u) \left( \frac{\epsilon^2}{2} ||B_{Y|X} \mathbf{L}_u||^2 + O(\epsilon^3) \right) \\
		&\approx \frac{\epsilon^2}{2} \sum_{u \in \mathcal{U}} P_U(u) ||B_{Y|X} \mathbf{L}_u||^2 \\
		&= \frac{\epsilon^2}{2} \mathbb{E}_U [||B_{Y|X} \mathbf{L}_U||^2].
	\end{align*}
	This completes the proof, neglecting terms of order $ O(\epsilon^3) $ and higher.
\end{proof}

\subsection{Proof of Lemma \ref{lem:approx_IUZ_formal} (Approximation of $I(U;Z)$)}
\begin{proof}
	The proof is entirely analogous to the proof of Lemma \ref{lem:approx_IUY_formal}.
\end{proof}

	\section*{Appendix B: Proof of Theorem \ref{thm:kkt_conditions_Lu_final}}\label{app:proof_of_theorem_6}
	\begin{proof}
	We analyze the EIT-Approximated SIC Problem (Problem \ref{prob:eit_slic_approximated}) by considering its Lagrangian $ \mathcal{L} $. For a fixed optimal distribution $ \{P_U^*(u)\} $ with support $ \mathcal{U}^* $, the KKT stationarity condition, $ \nabla_{\mathbf{L}_u} \mathcal{L} = \mathbf{0} $, must hold for each $ u \in \mathcal{U}^* $. The full Lagrangian is given in \eqref{eq:lagrangian_slic_final}. Differentiating with respect to $ \mathbf{L}_u $ for a specific $ u \in \mathcal{U}^* $ yields:
	\begin{equation} \label{eq:app_full_stationarity_final}
		2 P_U^*(u) V \mathbf{L}_u^* - 2 \rho^* P_U^*(u) I \mathbf{L}_u^* - 2 \nu^* P_U^*(u) \Lambda \mathbf{L}_u^* - \xi^*(u) \mathbf{\sqrt{P_X}} - P_U^*(u) \boldsymbol{\mu}^* = \mathbf{0},
	\end{equation}
	where $ \xi^*(u) $ is the scalar Lagrange multiplier for the constraint $ \mathbf{L}_u^T \mathbf{\sqrt{P_X}} = 0 $ in \eqref{eq:eit_slic_opt_ortho_problem_env} and $ \boldsymbol{\mu}^* $ is the vector multiplier for $ \sum_{u} P_U(u) \mathbf{L}_u = \mathbf{0} $ in \eqref{eq:eit_slic_opt_consistency_L_problem_env}. Let $ K(\rho^*, \nu^*) \triangleq -V + \rho^*I + \nu^*\Lambda $. Dividing \eqref{eq:app_full_stationarity_final} by $ -2P_U^*(u) $ (since $P_U^*(u) > 0$):
	\begin{equation} \label{eq:app_stationarity_intermediate_final}
		K(\rho^*, \nu^*) \mathbf{L}_u^* + \frac{\xi^*(u)}{2 P_U^*(u)} \mathbf{\sqrt{P_X}} + \frac{1}{2} \boldsymbol{\mu}^* = \mathbf{0}.
	\end{equation}
	Summing \eqref{eq:app_full_stationarity_final} over all $ u \in \mathcal{U}^* $ and using the constraints $ \sum_{u \in \mathcal{U}^*} P_U^*(u) \mathbf{L}_u^* = \mathbf{0} $ and $ \sum_{u \in \mathcal{U}^*} P_U^*(u) = 1 $, we find the relationship between the multipliers:
	\begin{equation*}
		- \left( \sum_{u \in \mathcal{U}^*} \xi^*(u) \right) \mathbf{\sqrt{P_X}} - \boldsymbol{\mu}^* = \mathbf{0} \quad \Rightarrow \quad \boldsymbol{\mu}^* = -\left( \sum_{u \in \mathcal{U}^*} \xi^*(u) \right) \mathbf{\sqrt{P_X}}.
	\end{equation*}
	This shows that $ \boldsymbol{\mu}^* $ is collinear with $ \mathbf{\sqrt{P_X}} $. Substituting this back into \eqref{eq:app_stationarity_intermediate_final} yields:
	\begin{equation} \label{eq:app_KLu_gamma_sqrtP}
		K \mathbf{L}_u^* = \gamma_u^* \mathbf{\sqrt{P_X}},
	\end{equation}
	where $ \gamma_u^* $ is a scalar constant defined as $ \gamma_u^* \triangleq \frac{1}{2} \left( \sum_{k \in \mathcal{U}^*} \xi^*(k) - \frac{\xi^*(u)}{P_U^*(u)} \right) $.
	
	We now prove by contradiction that $ \gamma_u^* $ must be zero for any meaningful optimal solution $ \mathbf{L}_u^* \neq \mathbf{0} $. Assume $ \gamma_u^* \neq 0 $.
	The matrix $ K $ is symmetric, thus it is orthogonally diagonalizable as $ K = Q \Delta Q^T $, where $ Q $ is an orthogonal matrix and $ \Delta $ is a diagonal matrix of the eigenvalues of $ K $.
	Let $ \tilde{\mathbf{L}}_u^* \triangleq Q^T \mathbf{L}_u^* $ and $ \tilde{\mathbf{p}} \triangleq Q^T \mathbf{\sqrt{P_X}} $.
	The equation \eqref{eq:app_KLu_gamma_sqrtP} becomes $ Q \Delta Q^T (Q \tilde{\mathbf{L}}_u^*) = \gamma_u^* (Q \tilde{\mathbf{p}}) $, which simplifies to:
	\begin{equation} \label{eq:app_diagonalized_stationarity}
		\Delta \tilde{\mathbf{L}}_u^* = \gamma_u^* \tilde{\mathbf{p}}.
	\end{equation}
	The orthogonality constraint $ (\mathbf{L}_u^*)^T \mathbf{\sqrt{P_X}} = 0 $ transforms to:
	\begin{equation} \label{eq:app_diagonalized_orthogonality}
		(\tilde{\mathbf{L}}_u^*)^T \tilde{\mathbf{p}} = \sum_j (\tilde{\mathbf{L}}_u^*)_{j} \tilde{p}_j = 0.
	\end{equation}
	From \eqref{eq:app_diagonalized_stationarity}, for each component $j = 1, \dots, |\mathcal{X}|$:
	\begin{itemize}
		\item If $ \Delta_{jj} \neq 0 $, then $ (\tilde{\mathbf{L}}_u^*)_{j} = \frac{\gamma_u^* \tilde{p}_j}{\Delta_{jj}} $.
		\item If $ \Delta_{jj} = 0 $, then for $ \Delta_{jj} (\tilde{\mathbf{L}}_u^*)_{j} = \gamma_u^* \tilde{p}_j $ to hold, we must have $ \gamma_u^* \tilde{p}_j = 0 $. Since we assumed $ \gamma_u^* \neq 0 $, this implies $ \tilde{p}_j = 0 $.
	\end{itemize}
	Substituting these into the orthogonality condition \eqref{eq:app_diagonalized_orthogonality}:
	\begin{equation*}
		\sum_{j: \Delta_{jj} \neq 0} \left(\frac{\gamma_u^* \tilde{p}_j}{\Delta_{jj}}\right) \tilde{p}_j + \sum_{j: \Delta_{jj} = 0} (\tilde{\mathbf{L}}_u^*)_{j} \tilde{p}_j = 0.
	\end{equation*}
	As established above, for every term in the second sum, $ \tilde{p}_j = 0 $ if $ \gamma_u^* \neq 0 $. Therefore, the second sum is zero. The equation reduces to:
	\begin{equation*}
		\gamma_u^* \sum_{j: \Delta_{jj} \neq 0} \frac{\tilde{p}_j^2}{\Delta_{jj}} = 0.
	\end{equation*}
	Since we assumed $ \gamma_u^* \neq 0 $, this forces the condition:
	\begin{equation} \label{eq:app_sum_p_sq_over_delta_zero}
		\sum_{j: \Delta_{jj} \neq 0} \frac{\tilde{p}_j^2}{\Delta_{jj}} = 0.
	\end{equation}
	However, by pre-multiplying \eqref{eq:app_KLu_gamma_sqrtP} by $ (\mathbf{L}_u^*)^T $, we obtain $ (\mathbf{L}_u^*)^T K \mathbf{L}_u^* = \gamma_u^* (\mathbf{L}_u^*)^T \mathbf{\sqrt{P_X}} $. Due to the orthogonality constraint \eqref{eq:eit_slic_opt_ortho_problem_env}, the right-hand side is zero, so $ (\mathbf{L}_u^*)^T K \mathbf{L}_u^* = 0 $. In the transformed coordinates, this is $ (\tilde{\mathbf{L}}_u^*)^T \Delta \tilde{\mathbf{L}}_u^* = \sum_j \Delta_{jj} ((\tilde{\mathbf{L}}_u^*)_{j})^2 = 0 $.
	Substituting $ (\tilde{\mathbf{L}}_u^*)_{j} = \frac{\gamma_u^* \tilde{p}_j}{\Delta_{jj}} $ for the non-singular part:
	\begin{equation*}
		\sum_{j: \Delta_{jj} \neq 0} \Delta_{jj} \left(\frac{\gamma_u^* \tilde{p}_j}{\Delta_{jj}}\right)^2 + \sum_{j: \Delta_{jj} = 0} 0 \cdot ((\tilde{\mathbf{L}}_u^*)_{j})^2 = 0,
	\end{equation*}
	which implies $ (\gamma_u^*)^2 \sum_{j: \Delta_{jj} \neq 0} \frac{\tilde{p}_j^2}{\Delta_{jj}} = 0 $. With $ \gamma_u^* \neq 0 $, this again leads to condition \eqref{eq:app_sum_p_sq_over_delta_zero}.
	
	This condition, $ \sum_{j: \Delta_{jj} \neq 0} \tilde{p}_j^2/\Delta_{jj} = 0 $, leads to a contradiction. If the matrix $K$ restricted to the subspace where $ \tilde{\mathbf{p}} $ has support were definite (i.e., all relevant $ \Delta_{jj} $ have the same sign), this would imply that all relevant $ \tilde{p}_j $ must be zero. Combined with the finding that $ \tilde{p}_j=0 $ for components where $ \Delta_{jj}=0 $, this would force $ \tilde{\mathbf{p}} = \mathbf{0} $. This implies $ Q^T \mathbf{\sqrt{P_X}} = \mathbf{0} $, and since $ Q $ is invertible, it would require $ \mathbf{\sqrt{P_X}} = \mathbf{0} $, which contradicts the assumption that $ P_X(x) > 0 $. Even if $K$ is indefinite, this condition imposes a non-generic constraint on the vector $ \mathbf{\sqrt{P_X}} $ relative to the eigensystem of $K$. For a robust solution, this path is not viable. Therefore, the initial assumption $ \gamma_u^* \neq 0 $ must be false for any optimal $ \mathbf{L}_u^* \neq \mathbf{0} $.
	
	Thus, with $ \gamma_u^* = 0 $, the stationarity condition with respect to $ \mathbf{L}_u $ simplifies to $ K(\rho^*, \nu^*) \mathbf{L}_u^* = \mathbf{0} $, which is the pivotal result \eqref{eq:k_Lu_equals_0_final} in Section \ref{sec:approx_capacity_solution_structure}. This is equivalent to $ V \mathbf{L}_u^* = (\rho^*I + \nu^*\Lambda) \mathbf{L}_u^* $.
	
	The full set of KKT conditions for Problem \ref{prob:eit_slic_approximated} requires that this stationarity condition holds in conjunction with primal and dual feasibility, and complementary slackness. While the detailed derivation is standard in constrained optimization theory \cite{Boyd}, we briefly outline these conditions for the optimal primal variables $ (\{P_U^*(u)\}, \{\mathbf{L}_u^*\}) $ and dual variables $ (\rho^*, \nu^*, \{\xi^*(u)\}, \boldsymbol{\mu}^*, \kappa^*, \{\zeta^*(u)\}) $:

	Primal feasibility implies that the optimal solution must satisfy all constraints of Problem \ref{prob:eit_slic_approximated}, namely constraints \eqref{eq:eit_slic_opt_rate_problem_env} through \eqref{eq:eit_slic_opt_pu_nonneg_problem_env}. For dual feasibility, the Lagrange multipliers associated with inequality constraints must be non-negative:
	\begin{align*}
		\rho^* &\ge 0, \\
		\nu^* &\ge 0, \\
		\zeta^*(u) &\ge 0, \quad \forall u \in \mathcal{U},
	\end{align*}
	where $ \zeta^*(u) $ are the multipliers for the constraints $ P_U^*(u) \ge 0 $. And the complementary slackness condition implies that the product of a Lagrange multiplier and its corresponding inequality constraint must be zero at the optimum, i.e.,
	\begin{align}
			\rho^* \left( \sum_{u \in \mathcal{U}} P_U^*(u) (\mathbf{L}_u^*)^T I \mathbf{L}_u^* - R' \right) &= 0, \label{eq:app_cs_rate_final} \\
			\nu^* \left( \sum_{u \in \mathcal{U}} P_U^*(u) (\mathbf{L}_u^*)^T \Lambda \mathbf{L}_u^* - \Theta' \right) &= 0, \label{eq:app_cs_leakage_final} \\
			\zeta^*(u) P_U^*(u) &= 0, \quad \forall u \in \mathcal{U}. \label{eq:app_cs_pu_nonneg_final}
	\end{align}
	Conditions \eqref{eq:app_cs_rate_final} and \eqref{eq:app_cs_leakage_final} are particularly important, as they mean that if a multiplier is strictly positive, its corresponding constraint must be active, i.e., hold with equality. Together, these conditions fully characterize the properties of an optimal solution to the EIT-approximated problem.
	\end{proof}

\section*{Appendix C: Proof of Proposition \ref{prop:pu_invariance}}
\label{app:proof_pu_invariance}

\begin{proof}
	
	 Let us rearrange the Lagrangian $ \mathcal{L} $ for Problem \ref{prob:eit_slic_approximated} given in \eqref{eq:lagrangian_slic_final} by grouping terms involving the primal variables:
	\begin{equation} \label{eq:app_rearranged_lagrangian}
		\mathcal{L} = \sum_{u \in \mathcal{U}} P_U(u) \left[ \mathbf{L}_u^T (V - \rho I - \nu \Lambda) \mathbf{L}_u - \boldsymbol{\mu}^T \mathbf{L}_u - \kappa + \zeta(u) \right] - \sum_{u \in \mathcal{U}} \xi(u) (\mathbf{L}_u^T \mathbf{\sqrt{P_X}}) + \rho R' + \nu \Theta' + \kappa.
	\end{equation}
	The dual function is $g(\rho, \nu, \{\xi(u)\}, \boldsymbol{\mu}, \kappa, \{\zeta(u)\}) = \sup_{\{\mathbf{L}_u\}, \{P_U(u) \ge 0\}} \mathcal{L}$. For the supremum to be finite, the KKT stationarity conditions must hold at the optimum.	

	Let $(\{\mathbf{L}_u^*\}, \{P_U^*(u)\})$ be an optimal primal solution and $(\rho^*, \nu^*, \dots)$ be an optimal dual solution.
	The stationarity condition with respect to $\mathbf{L}_u$ for a given $u$ where $P_U^*(u) > 0$ is $ \nabla_{\mathbf{L}_u} \mathcal{L} = \mathbf{0} $. This leads to the condition $K(\rho^*, \nu^*) \mathbf{L}_u^* = \mathbf{0}$.
	Crucially, we also consider the stationarity condition with respect to $P_U(u)$. For any $u$  with $P_U^*(u) > 0$, the KKT conditions require that the derivative of the Lagrangian with respect to $P_U(u)$ is zero. From \eqref{eq:app_rearranged_lagrangian}, this condition is:
	\begin{equation} \label{eq:app_kkt_wrt_pu}
		\frac{\partial \mathcal{L}}{\partial P_U(u)} \bigg|_* = (\mathbf{L}_u^*)^T (V - \rho^* I - \nu^* \Lambda) \mathbf{L}_u^* - (\boldsymbol{\mu}^*)^T \mathbf{L}_u^* - \kappa^* + \zeta^*(u) = 0.
	\end{equation}
	By complementary slackness for the constraint $P_U(u) \ge 0$, we have $\zeta^*(u)P_U^*(u)=0$. Since we are considering $u$ for which $P_U^*(u) > 0$, it must be that $\zeta^*(u)=0$.
	Therefore, for any $u$ in the support of $P_U^*$, the term in the square brackets of the Lagrangian \eqref{eq:app_rearranged_lagrangian} must satisfy the following:
	\begin{equation*}
		(\mathbf{L}_u^*)^T K(\rho^*,\nu^*)^T \mathbf{L}_u^* - (\boldsymbol{\mu}^*)^T \mathbf{L}_u^* - \kappa^* = 0.
	\end{equation*}
	
	The dual problem is to minimize the dual function $g$ over all dual variables. The optimal value of the primal problem is equal to the optimal value of the dual problem, $g^* = \min g$.
	At the optimal point $(\{\mathbf{L}_u^*\}, \{P_U^*(u)\}, \rho^*, \nu^*, \dots)$, the value of the Lagrangian \eqref{eq:app_rearranged_lagrangian} is:
	\begin{equation*}
		\mathcal{L}^* = \sum_{u \in \mathrm{supp}(P_U^*)} P_U^*(u) \underbrace{\left[ \dots \right]}_{=0 \text{ by \eqref{eq:app_kkt_wrt_pu}}} - \sum_{u \in \mathcal{U}} \xi^*(u) \underbrace{(\mathbf{L}_u^*)^T \mathbf{\sqrt{P_X}}}_{=0 \text{ by primal feasibility}} + \rho^* R' + \nu^* \Theta' + \kappa^*.
	\end{equation*}
	Due to the KKT condition \eqref{eq:app_kkt_wrt_pu}, the entire first sum becomes zero, regardless of the specific values of $P_U^*(u) > 0$. The second sum is zero because the optimal $\mathbf{L}_u^*$ must be feasible and thus satisfy the orthogonality constraint. The term $\kappa^*$ can be shown to be part of the objective from the dual perspective. The optimal value of the dual function is thus found to be:
	\begin{equation*}
		g^* = \rho^* R' + \nu^* \Theta'.
	\end{equation*}

	This final expression for the optimal value depends only on the fixed budgets $R, \Theta$ and the optimal multipliers $(\rho^*, \nu^*)$. Therefore, the optimal value of the objective function is independent of the specific source distribution.
\end{proof}

\section*{Appendix D: Proof of Theorem \ref{thm:c_lic_final}}\label{app:proof_of_theorem_7}
\begin{proof}
	The approximate local secrecy capacity, $ C_{\mathrm{SIC}} $, is defined as the optimal value of the objective function of the original SIC Problem (Problem \ref{prob:slic_original}) under the EIT approximation. The objective is to maximize $I(U;Y)$. Using the EIT approximation from Lemma \ref{lem:approx_IUY_formal}, this is equivalent to maximizing the objective of the EIT-Approximated SIC Problem (Problem \ref{prob:eit_slic_approximated}) scaled by the factor $ \epsilon^2/2 $:
	\begin{equation} \label{eq:app_clic_def}
		C_{\mathrm{SIC}} = \frac{\epsilon^2}{2} \max \left( \sum_{u \in \mathcal{U}} P_U(u) \mathbf{L}_u^T V \mathbf{L}_u \right),
	\end{equation}
	where the maximization is over the feasible set defined in Problem \ref{prob:eit_slic_approximated}.
	
	Let $ (\{P_U^*(u)\}, \{\mathbf{L}_u^*\}) $ be an optimal primal solution, and $ (\rho^*, \nu^*) $ be the optimal dual variables for the rate and leakage constraints. From the simplified KKT stationarity condition derived in Theorem \ref{thm:kkt_conditions_Lu_final}, we have for any $u$ in the support of $P_U^*$ and for any non-trivial optimal perturbation $ \mathbf{L}_u^* \neq \mathbf{0} $:
	\begin{equation} \label{eq:app_kkt_for_clic_proof}
		(-V + \rho^*I + \nu^*\Lambda) \mathbf{L}_u^* = \mathbf{0} \quad \Rightarrow \quad V \mathbf{L}_u^* = \rho^* \mathbf{L}_u^* + \nu^* \Lambda \mathbf{L}_u^*.
	\end{equation}
	We substitute this expression for $ V \mathbf{L}_u^* $ into the term inside the maximization in \eqref{eq:app_clic_def}. For the optimal solution, the objective value is:
	\begin{align}
		\sum_{u \in \mathcal{U}} P_U^*(u) (\mathbf{L}_u^*)^T V \mathbf{L}_u^* &= \sum_{u \in \mathcal{U}} P_U^*(u) (\mathbf{L}_u^*)^T (\rho^* \mathbf{L}_u^* + \nu^* \Lambda \mathbf{L}_u^*) \nonumber \\
		&= \rho^* \sum_{u \in \mathcal{U}} P_U^*(u) (\mathbf{L}_u^*)^T I \mathbf{L}_u^* + \nu^* \sum_{u \in \mathcal{U}} P_U^*(u) (\mathbf{L}_u^*)^T \Lambda \mathbf{L}_u^*, \label{eq:app_clic_intermediate_sum}
	\end{align}
	where we have used the linearity of sums and expectations.
	
	We now invoke the KKT complementary slackness conditions for the inequality constraints of Problem \ref{prob:eit_slic_approximated}:
	\begin{align}
		\rho^* \left( \sum_{u \in \mathcal{U}} P_U^*(u) ||\mathbf{L}_u^*||^2 - R' \right) &= 0 \label{eq:app_cs_rate_clic_proof} \\
		\nu^* \left( \sum_{u \in \mathcal{U}} P_U^*(u) (\mathbf{L}_u^*)^T \Lambda \mathbf{L}_u^* - \Theta' \right) &= 0 \label{eq:app_cs_leakage_clic_proof}
	\end{align}
	where $ R' = 2R/\epsilon^2 $ and $ \Theta' = 2\Theta/\epsilon^2 $.
	
	There are three possibilities for the optimal multipliers $ \rho^* $ and $ \nu^* $.  First, if $ \rho^* > 0 $, then by \eqref{eq:app_cs_rate_clic_proof}, the rate constraint must be active: $ \sum_{u} P_U^*(u) ||\mathbf{L}_u^*||^2 = R' $. Second, if $ \nu^* > 0 $, then by \eqref{eq:app_cs_leakage_clic_proof}, the leakage constraint must be active: $ \sum_{u} P_U^*(u) (\mathbf{L}_u^*)^T \Lambda \mathbf{L}_u^* = \Theta' $. Finally, if a multiplier is zero (e.g., $ \rho^*=0 $), its corresponding sum term in \eqref{eq:app_clic_intermediate_sum} vanishes.

	In all cases, we can substitute the implications of complementary slackness into \eqref{eq:app_clic_intermediate_sum}. For instance, if both $ \rho^* > 0 $ and $ \nu^* > 0 $, both constraints are active. If $ \rho^* > 0 $ and $ \nu^* = 0 $, the rate constraint is active and the second term in \eqref{eq:app_clic_intermediate_sum} is zero. In every case, the expression for the sum becomes:
	\begin{equation} \label{eq:app_clic_sum_simplified}
		\sum_{u \in \mathcal{U}} P_U^*(u) (\mathbf{L}_u^*)^T V \mathbf{L}_u^* = \rho^* R' + \nu^* \Theta'.
	\end{equation}
	This is because if, for example, $ \rho^*=0 $, then $ \sum P_U^* ||\mathbf{L}_u^*||^2 \le R' $, but the term $ \rho^* \sum P_U^* ||\mathbf{L}_u^*||^2 $ in the sum is $ 0 \cdot \sum P_U^* ||\mathbf{L}_u^*||^2 = 0 $, which equals $ \rho^* R' = 0 \cdot R' = 0 $. The equality \eqref{eq:app_clic_sum_simplified} thus holds regardless of whether the constraints are active or not, due to the complementary slackness conditions.
	
	Finally, we substitute \eqref{eq:app_clic_sum_simplified} back into the definition of $ C_{\mathrm{SIC}} $ from \eqref{eq:app_clic_def}:
	\begin{align*}
		C_{\mathrm{SIC}} &= \frac{\epsilon^2}{2} (\rho^* R' + \nu^* \Theta') \\
		&= \frac{\epsilon^2}{2} \left[ \rho^* \left(\frac{2R}{\epsilon^2}\right) + \nu^* \left(\frac{2\Theta}{\epsilon^2}\right) \right] \\
		&= \rho^* R + \nu^* \Theta.
	\end{align*}
	This establishes the formula \eqref{eq:c_lic_formula_final}.
	
	The interpretation of $ \rho^* $ as a generalized eigenvalue follows directly from rearranging the KKT stationarity condition \eqref{eq:app_kkt_for_clic_proof} to $ (V - \nu^*\Lambda)\mathbf{L}_u^* = \rho^* \mathbf{L}_u^* $. This is the generalized eigenvalue equation for the pencil $ (V - \nu^*\Lambda, I) $, where $ \rho^* $ is the specific eigenvalue that, in conjunction with $ \nu^* $, satisfies all KKT conditions for optimality.
\end{proof}
				
\section*{Appendix E: Proof of Theorem \ref{thm:general_lp_for_multipliers}}
\label{app:proof_general_lp}

\begin{proof}
	The Lagrangian for the EIT-Approximated SIC problem is given in \eqref{eq:lagrangian_slic_final}. A necessary KKT condition for optimality is that the dual function must be bounded. This requires that the quadratic part of the Lagrangian in $\mathbf{L}_u$, remains bounded from above when maximized over $\mathbf{L}_u$. This is only possible if the matrix $K(\rho, \nu) \triangleq -V + \rho I + \nu \Lambda$ is positive semidefinite (PSD) on the subspace of valid perturbations, $\mathcal{S}^{\perp} = \{\mathbf{L} | \mathbf{L}^T \mathbf{\sqrt{P_X}} = 0\}$. This fundamental condition is expressed as:
	\begin{equation} \label{eq:app_k_psd_main_lmi}
		K(\rho, \nu) \succeq_{\mathcal{S}^{\perp}} \mathbf{0}.
	\end{equation}
	Since the entries of the matrix $K(\rho, \nu)$ are affine functions of the variables $\rho$ and $\nu$, the condition \eqref{eq:app_k_psd_main_lmi} is a Linear Matrix Inequality (LMI). The general problem of finding the optimal multipliers is to maximize the linear objective $Z_D = \rho R + \nu \Theta$ subject to the LMI \eqref{eq:app_k_psd_main_lmi} and the linear constraints \eqref{eq:lp_general_constr_obj_upper_final_v3} and \eqref{eq:lp_general_constr_nonneg_final_v3}. Such a problem is a Semidefinite Program (SDP).
	
	We now show that this SDP is equivalent to a standard LP. The PSD condition \eqref{eq:app_k_psd_main_lmi} is equivalent to:
	\begin{equation} \label{eq:app_k_psd_rearranged_proof}
		\rho (\mathbf{L}^T I \mathbf{L}) + \nu (\mathbf{L}^T \Lambda \mathbf{L}) \ge \mathbf{L}^T V \mathbf{L}, \quad \forall \mathbf{L} \in \mathcal{S}^{\perp}.
	\end{equation}
	We assume $\Lambda$ is positive definite when restricted to the subspace $\mathcal{S}^{\perp}$, which is a reasonable condition implying that any non-trivial perturbation incurs some leakage to Eve. Let $\Lambda|_{\mathcal{S}^{\perp}}$ and $V|_{\mathcal{S}^{\perp}}$ be the matrices restricted to operate on vectors in $\mathcal{S}^{\perp}$. Since $\Lambda|_{\mathcal{S}^{\perp}}$ is symmetric and positive definite, its symmetric square root $\Lambda^{1/2}|_{\mathcal{S}^{\perp}}$ and its inverse $\Lambda^{-1/2}|_{\mathcal{S}^{\perp}}$ exist and are also positive definite.
	
	The condition \eqref{eq:app_k_psd_rearranged_proof} can be analyzed using a congruence transformation that simplifies the pencil $(V, \Lambda)$, \cite{horn2013matrix, Golub}. Let $\mathbf{L} \in \mathcal{S}^{\perp}$ and define a new vector $\mathbf{z}$ via the invertible mapping $\mathbf{L} = \Lambda^{-1/2}|_{\mathcal{S}^{\perp}} \mathbf{z}$. Substituting this into \eqref{eq:app_k_psd_rearranged_proof}:
	\begin{equation*}
		\rho (\mathbf{z}^T \Lambda^{-1/2} I \Lambda^{-1/2} \mathbf{z}) + \nu (\mathbf{z}^T \Lambda^{-1/2} \Lambda \Lambda^{-1/2} \mathbf{z}) \ge \mathbf{z}^T \Lambda^{-1/2} V \Lambda^{-1/2} \mathbf{z},
	\end{equation*}
	where for brevity we omit the explicit restriction to $\mathcal{S}^{\perp}$. This simplifies to:
	\begin{equation} \label{eq:app_transformed_psd_condition_proof}
		\rho (\mathbf{z}^T \Lambda^{-1} \mathbf{z}) + \nu ||\mathbf{z}||^2 \ge \mathbf{z}^T (\Lambda^{-1/2} V \Lambda^{-1/2}) \mathbf{z}, \quad \forall \mathbf{z} \in \text{range}(\Lambda^{1/2}|_{\mathcal{S}^{\perp}}).
	\end{equation}
	Let $\tilde{V} \triangleq \Lambda^{-1/2} V \Lambda^{-1/2}$. The matrix $\tilde{V}$ is symmetric, and its eigenvalues are precisely the generalized eigenvalues of the pencil $(V, \Lambda)$ \cite{Golub}. Let these generalized eigenvalues be $\{d_j\}_{j \in \mathcal{J}_{\perp}}$, and let the corresponding orthonormal eigenvectors of $\tilde{V}$ be $\{\mathbf{q}_j\}_{j \in \mathcal{J}_{\perp}}$. These $\{\mathbf{q}_j\}$ form an orthonormal basis for the transformed space.
	
	Since the inequality \eqref{eq:app_transformed_psd_condition_proof} must hold for any vector $\mathbf{z}$ in this space, it is sufficient to ensure it holds for the basis vectors $\mathbf{q}_j$. For each $j \in \mathcal{J}_{\perp}$, we substitute $\mathbf{z} = \mathbf{q}_j$:
	\begin{equation*}
		\rho (\mathbf{q}_j^T \Lambda^{-1} \mathbf{q}_j) + \nu (\mathbf{q}_j^T I \mathbf{q}_j) \ge \mathbf{q}_j^T \tilde{V} \mathbf{q}_j.
	\end{equation*}
	By the properties of orthonormal eigenvectors, $\mathbf{q}_j^T I \mathbf{q}_j = ||\mathbf{q}_j||^2 = 1$, and $\mathbf{q}_j^T \tilde{V} \mathbf{q}_j = d_j$. The term $\mathbf{q}_j^T \Lambda^{-1} \mathbf{q}_j$ is the Rayleigh quotient\cite{Croot} of $\Lambda^{-1}$ with respect to $\mathbf{q}_j$. Let us denote $\lambda_j^{-1} \triangleq \mathbf{q}_j^T \Lambda^{-1} \mathbf{q}_j$. The value $\lambda_j$ can be interpreted as the eigenvalue of $\Lambda$ with respect to the corresponding generalized eigenvector of $(V, \Lambda)$ in the original basis. The inequality for each basis vector becomes:
	\begin{equation*}
		\rho \lambda_j^{-1} + \nu \ge d_j, \quad \forall j \in \mathcal{J}_{\perp}.
	\end{equation*}
	Multiplying by $\lambda_j$ (which is positive since $\Lambda$ is positive definite on $\mathcal{S}^{\perp}$) yields the final set of linear inequalities:
	\begin{equation} \label{eq:app_final_lp_constraints_deriv_proof}
		\rho + \nu \lambda_j \ge d_j \lambda_j, \quad \forall j \in \mathcal{J}_{\perp}.
	\end{equation}
	This establishes that the infinite set of constraints implicitly defined by the LMI in \eqref{eq:app_k_psd_main_lmi} is equivalent to the finite set of $M=|\mathcal{X}|-1$ linear inequalities in \eqref{eq:app_final_lp_constraints_deriv_proof}.
	
	With the objective function $\rho R + \nu \Theta$ being linear and all constraints now shown to be equivalent to a finite set of linear inequalities, the optimization problem for the multipliers is an LP. The bounding constraint \eqref{eq:lp_general_constr_obj_upper_final_v3} ensures a finite solution exists and that completes the proof.
\end{proof}
		
\section*{Appendix F: Proof of Lemma \ref{thm:lp_feasibility_condition_final_v2}}	\label{app:proof_thm_feasibility}		
	\begin{proof}
	A solution $(\rho, \nu)$ is feasible if it satisfies the constraints:
	\begin{align}
		\rho + \nu \lambda_j &\ge d_j \lambda_j, \quad \forall j \in \mathcal{J}_{\perp} \label{eq:app_feas_constr1_v2} \\
		\rho R + \nu \Theta &\le C_{\mathrm{max}} \label{eq:app_feas_constr2_v2} \\
		\rho \ge 0, \quad \nu &\ge 0 \label{eq:app_feas_constr3_v2}
	\end{align}
	Since the LP objective is bounded by \eqref{eq:app_feas_constr2_v2}, a solution fails to exist if and only if the feasible set is empty. We use Farkas' Lemma \cite{Boyd} to establish a sufficient condition for this set to be non-empty.
	
	Let us rewrite the system of inequalities in the standard form $A\mathbf{x} \le \mathbf{b}$, where $\mathbf{x} = [\nu, \rho]^T \ge \mathbf{0}$. The constraints \eqref{eq:app_feas_constr1_v2} are equivalent to $d_j \lambda_j - \lambda_j \nu - \rho \le 0$ for each $j \in \mathcal{J}_{\perp}$. Let $M = |\mathcal{J}_{\perp}|$. The system is:
	\begin{equation*}
		\begin{pmatrix}
			-\lambda_1 & -1 \\
			\vdots & \vdots \\
			-\lambda_M & -1 \\
			\Theta & R
		\end{pmatrix}
		\begin{pmatrix} \nu \\ \rho \end{pmatrix}
		\le
		\begin{pmatrix}
			-d_1 \lambda_1 \\
			-d_2 \lambda_2 \\
			\vdots \\
			-d_M \lambda_M \\
			C_{\mathrm{max}}
		\end{pmatrix}
		, \quad \text{and} \quad
		\begin{pmatrix} \nu \\ \rho \end{pmatrix} \ge \mathbf{0}.
	\end{equation*}
	
	Farkas' Lemma states that a solution $\mathbf{x} \ge \mathbf{0}$ to $A\mathbf{x} \le \mathbf{b}$ exists if and only if the alternative system, which seeks a vector $\mathbf{y} \ge \mathbf{0}$ such that $A^T\mathbf{y} \ge \mathbf{0}$ and $\mathbf{b}^T\mathbf{y} < 0$, has no solution. We will show that under the theorem's condition, this alternative system indeed has no solution.
	
	Let $\mathbf{y} = [y_1, \dots, y_M, y_{M+1}]^T \ge \mathbf{0}$. The condition $A^T\mathbf{y} \ge \mathbf{0}$ translates to:
	\begin{align}
		-\sum_{j=1}^{M} y_j + R y_{M+1} &\ge 0 \label{eq:app_farkas_concise_1} \\
		-\sum_{j=1}^{M} \lambda_j y_j + \Theta y_{M+1} &\ge 0 \label{eq:app_farkas_concise_2}
	\end{align}
	And the condition $\mathbf{b}^T\mathbf{y} < 0$ is:
	\begin{equation} \label{eq:app_farkas_concise_3}
		-\sum_{j=1}^{M} d_j \lambda_j y_j + C_{\mathrm{max}} y_{M+1} < 0
	\end{equation}
	
	Assume that a non-zero solution $\mathbf{y} \ge \mathbf{0}$ to this alternative system exists. If $y_{M+1}=0$, then from \eqref{eq:app_farkas_concise_1}, $-\sum y_j \ge 0$, which for $y_j \ge 0$ implies all $y_j=0$. This is the trivial solution $\mathbf{y}=\mathbf{0}$, which violates $\mathbf{b}^T\mathbf{y} < 0$. Thus, any non-trivial solution must have $y_{M+1} > 0$.
	
	From \eqref{eq:app_farkas_concise_2}, we have $\Theta y_{M+1} \ge \sum_{j=1}^{M} \lambda_j y_j$.
	Setting $C_{\mathrm{max}} = R \cdot \lambda_{\max}^{\perp}(V)=R \lambda_V^*$, and use it in \eqref{eq:app_farkas_concise_3} yields:
	\begin{align*}
		\mathbf{b}^T\mathbf{y} &= R \lambda_V^* y_{M+1} - \sum_{j=1}^{M} d_j \lambda_j y_j \\
		&\ge \frac{R}{\Theta} \lambda_V^* (\Theta y_{M+1}) - \sum_{j=1}^{M} d_j \lambda_j y_j \\
		&\ge \frac{R}{\Theta} \lambda_V^* \left( \sum_{j=1}^{M} \lambda_j y_j \right) - \sum_{j=1}^{M} d_j \lambda_j y_j \\
		&= \sum_{j=1}^{M} \lambda_j y_j \left( \frac{R}{\Theta} \lambda_V^* - d_j \right).
	\end{align*}
	The condition given in Theorem \ref{thm:lp_feasibility_condition_final_v2} is $\lambda_V^* > \frac{\Theta}{R} d_{\max}^{\perp}(V, \Lambda)$, which is equivalent to $\frac{R}{\Theta} \lambda_V^* > \max_{j}\{d_j\}$. This implies that for every $j \in \mathcal{J}_{\perp}$, the term $( \frac{R}{\Theta} \lambda_V^* - d_j )$ is strictly positive.
	Since $\lambda_j > 0$ and at least one $y_j$ must be positive for a non-trivial solution, from \eqref{eq:app_farkas_concise_1} or \eqref{eq:app_farkas_concise_2} with $y_{M+1}>0$, every term in the sum is non-negative, and at least one is strictly positive. Therefore:
	\begin{equation*}
		\sum_{j=1}^{M} \lambda_j y_j \left( \frac{R}{\Theta} \lambda_V^* - d_j \right) > 0.
	\end{equation*}
	This implies $\mathbf{b}^T\mathbf{y} > 0$. This contradicts the Farkas condition \eqref{eq:app_farkas_concise_3}, which requires $\mathbf{b}^T\mathbf{y} < 0$. Since the assumption of a non-zero solution $\mathbf{y} \ge \mathbf{0}$ for the alternative system leads to a contradiction, no such solution exists. By Farkas' Lemma, this guarantees that the original system for $(\rho, \nu)$ has a solution, thus the feasible set of the LP is non-empty and that completes the proof.
	\end{proof}

\section*{Appendix G: Proof of Theorem \ref{thm_c_lic_regimes}}	\label{app:proof_thm_regimes}

\begin{proof}
The fundamental theorem of Linear Programming \cite{BertsimasTsitsiklisLP} states that if an optimal solution exists, it must be achieved at one or more vertices of the feasible region. We analyze the objective function $Z_D(\rho, \nu) = \rho R + \nu \Theta$ at each type of vertex to derive the overall solution $C_{\mathrm{SIC}} = \max \{C_{R}, C_{\Theta}, C_{\mathrm{inter}}\}$.

First we examine the rate-dominant regime. This case corresponds to optimal vertices on the $\rho$-axis, where $\nu^* = 0$ and $\rho^* \ge 0$. With $\nu=0$, the LP constraints \eqref{eq:lp_general_constr_eig_final_v3}-\eqref{eq:lp_general_constr_nonneg_final_v3} reduce to finding $\rho$ such that:
\begin{align*}
	\text{(a)} \quad & \rho \ge (d_V)_j, \quad \forall j \in \mathcal{J}_{\perp} \quad (\text{assuming } (d_I)_j=1) \\
	\text{(b)} \quad & \rho R \le C_{\mathrm{max}} \\
	\text{(c)} \quad & \rho \ge 0.
\end{align*}
Constraint (a) implies that any feasible $\rho$ must satisfy $\rho \ge \max_{j \in \mathcal{J}_{\perp}} (d_V)_j = \lambda_{\max}^{\perp}(V)$.
Constraint (b) implies $\rho \le C_{\mathrm{max}}/R$.
Thus, the feasible segment for $\rho$ on this axis is $[\lambda_{\max}^{\perp}(V), C_{\mathrm{max}}/R]$. For this segment to be non-empty, we require $\lambda_{\max}^{\perp}(V) \le C_{\mathrm{max}}/R$. The objective is to maximize $\rho R$. This is achieved at the largest feasible $\rho$, i.e., $\rho^* = C_{\mathrm{max}}/R$.
However, a vertex solution on this axis must be defined by at least two active constraints. The active constraints are $\nu=0$ and one of the eigenmode constraints $\rho = (d_V)_k$ or the bounding constraint $\rho R = C_{\mathrm{max}}$. To maximize $\rho R$, the optimal vertex on this axis will be at $\rho^* = \min(\lambda_{\max}^{\perp}(V), C_{\mathrm{max}}/R)$, assuming $\lambda_{\max}^{\perp}(V)$ represents the most binding of the $(d_V)_j$ constraints that a feasible solution must satisfy while aiming to maximize $\rho$.
A more direct interpretation is that the best possible candidate capacity from this regime is when $\rho$ takes its maximum feasible value on the axis, which leads to $C = \rho^* R$. The LP will choose this type of vertex if the objective slope $R/\Theta$ is sufficiently high. The maximum value is thus:
\begin{equation} \label{eq:app_c_r_dom_final}
	C_{R} = \min(\lambda_{\max}^{\perp}(V)R, C_{\mathrm{max}}).
\end{equation}

The case of leakage-dominant regime corresponds to optimal vertices on the $\nu$-axis, where $\rho^* = 0$ and $\nu^* \ge 0$. With $\rho=0$, the LP constraints become:
\begin{align*}
	\text{(a)} \quad & \nu (d_{\Lambda})_j \ge (d_V)_j, \quad \forall j \in \mathcal{J}_{\perp} \\
	\text{(b)} \quad & \nu \Theta \le C_{\mathrm{max}} \\
	\text{(c)} \quad & \nu \ge 0.
\end{align*}
Constraint (a) requires $\nu \ge (d_V)_j / (d_{\Lambda})_j$ for all $j$ where $(d_{\Lambda})_j > 0$. Thus, any feasible $\nu$ must satisfy $\nu \ge \max_{j \in \mathcal{J}_{\perp}, (d_{\Lambda})_j > 0} \{ \frac{(d_V)_j}{(d_{\Lambda})_j} \} = d_{\max}^{\perp}(V, \Lambda)$.
Constraint (b) implies $\nu \le C_{\mathrm{max}}/\Theta$.
The feasible segment for $\nu$ on this axis is $[d_{\max}^{\perp}(V, \Lambda), C_{\mathrm{max}}/\Theta]$. The objective is to maximize $\nu \Theta$.
Following the same logic as the rate-dominant regime, the maximum value achievable on this axis is:
\begin{equation} \label{eq:app_c_theta_dom_final}
	C_{\Theta} = \min(d_{\max}^{\perp}(V, \Lambda)\Theta, C_{\mathrm{max}}).
\end{equation}

Finally, the intermediate regime case corresponds to optimal vertices in the interior of the first quadrant, where $\rho^* > 0$ and $\nu^* > 0$. Such a vertex must be formed by the intersection of at least two linearly independent, active constraints.

	Consider a vertex $(\rho_{jk}^*, \nu_{jk}^*)$ formed by the intersection of two distinct eigenmode constraint boundaries for modes $j,k \in \mathcal{J}_{\perp}$ (assuming $(d_I)_j=1$):
	\begin{align*}
		\rho + \nu (d_{\Lambda})_j &= (d_V)_j \\
		\rho + \nu (d_{\Lambda})_k &= (d_V)_k.
	\end{align*}
	Assuming $(d_{\Lambda})_j \neq (d_{\Lambda})_k$, solving this $2 \times 2$ system yields the unique intersection point:
	\begin{align}
		\nu_{jk}^* &= \frac{(d_V)_j - (d_V)_k}{(d_{\Lambda})_j - (d_{\Lambda})_k}, \label{eq:app_nu_jk_formula} \\
		\rho_{jk}^* &= (d_V)_j - \nu_{jk}^* (d_{\Lambda})_j. \label{eq:app_rho_jk_formula}
	\end{align}
	
	 Consider a vertex $(\rho_{j,b}^*, \nu_{j,b}^*)$ formed by the intersection of an eigenmode constraint for mode $j$ and the overall bounding constraint:
	\begin{align*}
		\rho + \nu (d_{\Lambda})_j &= (d_V)_j \\
		\rho R + \nu \Theta &= C_{\mathrm{max}}.
	\end{align*}
	Solving this system yields another set of candidate multipliers.

The set of all feasible interior vertices, $\mathcal{V}_{\mathrm{inter}}$, is the collection of all such intersection points that satisfy the feasibility conditions outlined in \eqref{eq:valid_pair_conditions}.
The candidate capacity from this regime, $C_{\mathrm{inter}}$, is the maximum objective value $\rho R + \nu \Theta$ evaluated over all vertices in this set:
\begin{equation} \label{eq:app_c_inter_final}
	C_{\mathrm{inter}} = \max_{(\rho, \nu) \in \mathcal{V}_{\mathrm{inter}}} (\rho R + \nu \Theta).
\end{equation}
If no such feasible interior vertex exists, $\mathcal{V}_{\mathrm{inter}}$ is empty and $C_{\mathrm{inter}}$ is taken as $-\infty$.

The overall optimal value of the LP, $C_{\mathrm{SIC}}$, must be the maximum value achieved at any feasible vertex. The set of all feasible vertices is the union of the feasible axial vertices and the feasible interior vertices. Therefore,
\begin{equation*}
	C_{\mathrm{SIC}} = \max \{C_{R}, C_{\Theta}, C_{\mathrm{inter}}\}.
\end{equation*}
This completes the proof.
\end{proof}

\section*{Appendix H: Proof of Theorem \ref{thm:commuting_matrices_unified}}	\label{app:proof_thm_commuting}
	
	\begin{proof}	
		The first part of the theorem states that for commuting matrices, the general linear constraints of the LP in Theorem \ref{thm:general_lp_for_multipliers}, given by $\rho + \nu \lambda_j \ge d_j \lambda_j$ for each mode $j \in \mathcal{J}_{\perp}$, simplify to $\rho + \nu (d_{\Lambda})_j \ge (d_V)_j$.
		
		If $V$ and $\Lambda$ commute and are symmetric, they are simultaneously diagonalizable by a single orthogonal matrix $Q$ whose columns $\{\mathbf{e}_j\}$ are common orthonormal eigenvectors \cite{horn2013matrix}. Let $\{\mathbf{e}_j\}_{j \in \mathcal{J}_{\perp}}$ be the subset of these eigenvectors that form an orthonormal basis for the perturbation subspace $\mathcal{S}^{\perp}$.
		For any such common eigenvector $\mathbf{e}_j$, we have:
		\begin{align*}
			V \mathbf{e}_j &= (d_V)_j \mathbf{e}_j \\
			\Lambda \mathbf{e}_j &= (d_{\Lambda})_j \mathbf{e}_j,
		\end{align*}
		where $(d_V)_j$ and $(d_{\Lambda})_j$ are the standard eigenvalues of $V$ and $\Lambda$, respectively.
		
		By definition, the generalized eigenvalues $d_j$ of the pencil $(V, \Lambda)$ corresponding to a generalized eigenvector $\mathbf{u}_j$ satisfy $V\mathbf{u}_j = d_j \Lambda \mathbf{u}_j$.
		If we set the generalized eigenvector $\mathbf{u}_j$ to be the common eigenvector $\mathbf{e}_j$, we can write:
		\begin{equation*}
			V \mathbf{e}_j = d_j \Lambda \mathbf{e}_j.
		\end{equation*}
		Substituting the standard eigenvalue relations:
		\begin{equation*}
			(d_V)_j \mathbf{e}_j = d_j (d_{\Lambda})_j \mathbf{e}_j.
		\end{equation*}
		Since $\mathbf{e}_j \neq \mathbf{0}$, we can equate the scalar coefficients:
		\begin{equation*}
			(d_V)_j = d_j (d_{\Lambda})_j.
		\end{equation*}
		If $(d_{\Lambda})_j > 0$, the generalized eigenvalue is simply the ratio of the standard eigenvalues: $d_j = (d_V)_j / (d_{\Lambda})_j$.
		The parameter $\lambda_j$ in the general LP constraint \eqref{eq:lp_general_constr_eig_final_v3} is defined as the eigenvalue of $\Lambda$ corresponding to the $j$-th mode, which in the commuting case is simply $(d_{\Lambda})_j$.
		
		Substituting $d_j = (d_V)_j / (d_{\Lambda})_j$ and $\lambda_j = (d_{\Lambda})_j$ into the general LP constraint $\rho + \nu \lambda_j \ge d_j \lambda_j$:
		\begin{equation*}
			\rho + \nu (d_{\Lambda})_j \ge \left( \frac{(d_V)_j}{(d_{\Lambda})_j} \right) (d_{\Lambda})_j.
		\end{equation*}
		This directly simplifies to:
		\begin{equation*}
			\rho + \nu (d_{\Lambda})_j \ge (d_V)_j,
		\end{equation*}
		which is the simplified linear constraint \eqref{eq:lp_commuting_constr}. This holds for all $j \in \mathcal{J}_{\perp}$ where $(d_{\Lambda})_j > 0$. If $(d_{\Lambda})_j = 0$, the general PSD condition $\rho + \nu(0) \ge (d_V)_j$ must still hold. The simplified form is thus valid. This completes the proof of the first part.
		
		The second part of the theorem states that an optimal perturbation $\mathbf{L}_u^*$ can only have components along common eigenmodes $\mathbf{e}_j$ for which the condition $(d_V)_j = \rho^* + \nu^*(d_{\Lambda})_j$ is met.
		
		This follows directly from the KKT stationarity condition from Theorem \ref{thm:kkt_conditions_Lu_final}:
		\begin{equation} \label{eq:app_kkt_cond_commuting}
			(-V + \rho^*I + \nu^*\Lambda) \mathbf{L}_u^* = \mathbf{0}.
		\end{equation}
		Any optimal perturbation vector $\mathbf{L}_u^*$ must lie in the subspace $\mathcal{S}^{\perp}$, and can thus be expressed in the common eigenbasis as:
		\begin{equation*}
			\mathbf{L}_u^* = \sum_{j \in \mathcal{J}_{\perp}} (\tilde{L}_u^*)_j \mathbf{e}_j,
		\end{equation*}
		where $(\tilde{L}_u^*)_j = (\mathbf{L}_u^*)^T \mathbf{e}_j$ are the coefficients.
		Substituting this expansion into \eqref{eq:app_kkt_cond_commuting}:
		\begin{equation*}
			(-V + \rho^*I + \nu^*\Lambda) \left( \sum_{j \in \mathcal{J}_{\perp}} (\tilde{L}_u^*)_j \mathbf{e}_j \right) = \mathbf{0}.
		\end{equation*}
		By linearity, and since $\mathbf{e}_j$ are eigenvectors of $V, I,$ and $\Lambda$:
		\begin{align*}
			\sum_{j \in \mathcal{J}_{\perp}} (\tilde{L}_u^*)_j (-V\mathbf{e}_j + \rho^*I\mathbf{e}_j + \nu^*\Lambda\mathbf{e}_j) &= \mathbf{0} \\
			\sum_{j \in \mathcal{J}_{\perp}} (\tilde{L}_u^*)_j (-(d_V)_j\mathbf{e}_j + \rho^*\mathbf{e}_j + \nu^*(d_{\Lambda})_j\mathbf{e}_j) &= \mathbf{0} \\
			\sum_{j \in \mathcal{J}_{\perp}} (\tilde{L}_u^*)_j (-(d_V)_j + \rho^* + \nu^*(d_{\Lambda})_j) \mathbf{e}_j &= \mathbf{0}.
		\end{align*}
		Since the eigenvectors $\{\mathbf{e}_j\}_{j \in \mathcal{J}_{\perp}}$ form an orthogonal and thus linearly independent set, for the above sum to be the zero vector, the scalar coefficient of each eigenvector $\mathbf{e}_j$ must be zero. Therefore, for every $j \in \mathcal{J}_{\perp}$:
		\begin{equation*}
			(\tilde{L}_u^*)_j \cdot (-(d_V)_j + \rho^* + \nu^*(d_{\Lambda})_j) = 0.
		\end{equation*}
		This implies that if a particular eigenmode $j$ is active in the optimal perturbation $\mathbf{L}_u^*$, i.e.,  $(\tilde{L}_u^*)_j \neq 0$, then its corresponding eigenvalues must satisfy:
		\begin{equation*}
			-(d_V)_j + \rho^* + \nu^*(d_{\Lambda})_j = 0,
		\end{equation*}
		which rearranges to the condition \eqref{eq:commuting_eigenvalue_condition}. Conversely, if an eigenmode's eigenvalues do not satisfy this equality, its corresponding coefficient $(\tilde{L}_u^*)_j$ must be zero. This completes the proof.
		\begin{figure}[H]
			\centering
			\includegraphics[width=0.65\textwidth]{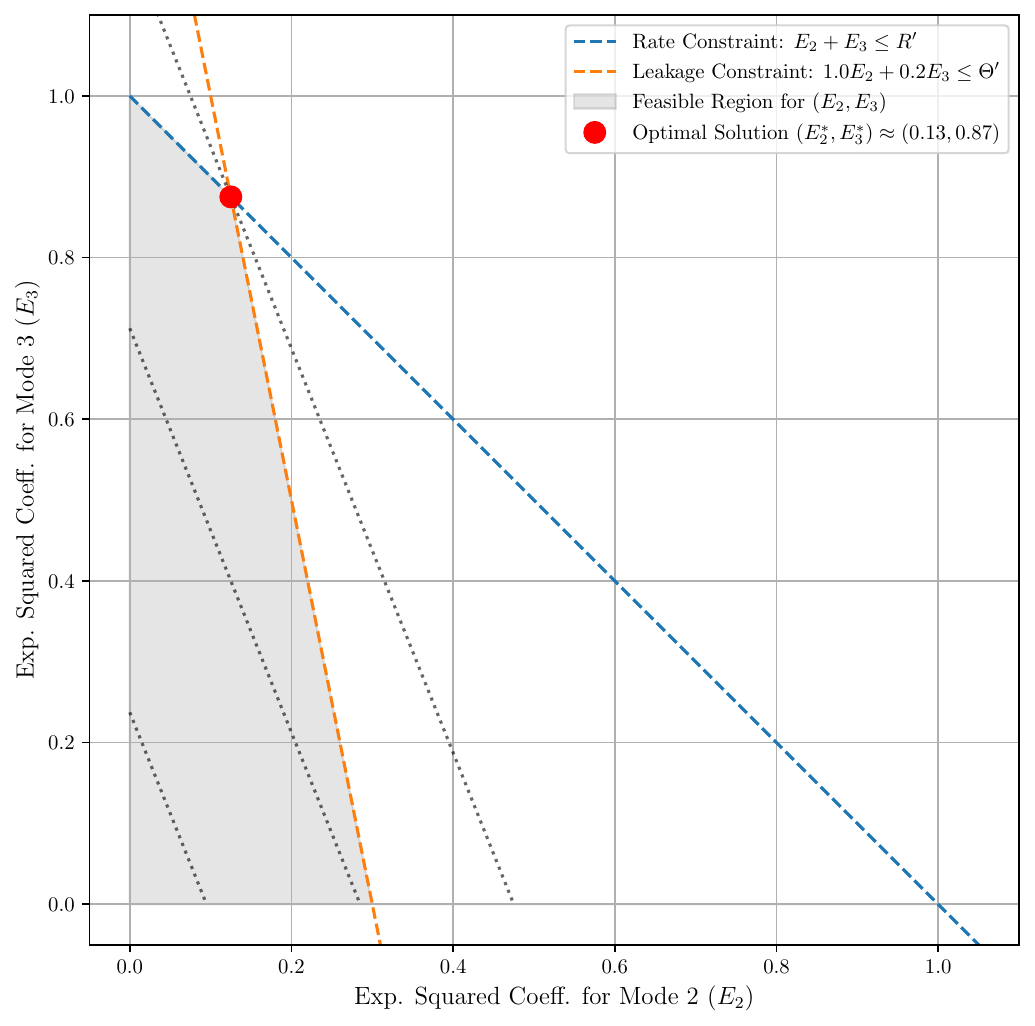}
			\caption{The axes represent the total expected squared magnitudes ($E_2, E_3$) of the perturbation coefficients along each of the two common eigenmodes. The optimal solution (red circle) is found at a vertex of the feasible region defined by $R'$ and $\Theta'$.}
			\label{fig:Energy_Allocation_LP_final_sec6_full}
		\end{figure}
		
		Figure \ref{fig:Energy_Allocation_LP_final_sec6_full} conceptually illustrates the solution to this LP for a system with a 2-dimensional perturbation subspace, $|\mathcal{J}_{\perp}|=2$. For the example shown, with parameters $R'=1.0, \Theta'=0.3$ and eigenvalues $(d_V)_2=5.0, (d_{\Lambda})_2=1.0$ for Mode 2, and $(d_V)_3=2.0, (d_{\Lambda})_3=0.2$ for Mode 3, the optimal allocation of squared coefficients occurs at the vertex $(E_2^* \approx 0.13, E_3^* \approx 0.87)$. This demonstrates the principle of selective activation: the system allocates coefficients to the most profitable eigenmodes to maximize utility while satisfying the overall rate and leakage budgets.
	\end{proof}	
	
		\section*{Appendix I: BSWC Setup and EIT Parameter Derivation}	\label{app:bswc_details}	
		Let Alice's input $ X $, Bob's output $ Y $, and Eve's output $ Z $ be binary, $ \mathcal{X}=\mathcal{Y}=\mathcal{Z}=\{0,1\} $. Bob's channel is a BSC with crossover probability $ p_{\mathrm{bob}} =P(Y \neq X | X)$. Eve's channel is a BSC with crossover probability $ q_{\mathrm{eve}}=P(Z \neq X | X) $. The channel transition probability matrices are:
		\begin{equation*}
			P_{Y|X} =\begin{pmatrix} 1-p_{\mathrm{bob}} & p_{\mathrm{bob}} \\ p_{\mathrm{bob}} & 1-p_{\mathrm{bob}} \end{pmatrix}, \quad
			P_{Z|X} = \begin{pmatrix} 1-q_{\mathrm{eve}} & q_{\mathrm{eve}} \\ q_{\mathrm{eve}} & 1-q_{\mathrm{eve}} \end{pmatrix}.
		\end{equation*}
		We assume a uniform input reference distribution $ P_X(x) = 0.5 $ for $ x \in \{0,1\} $. Thus, $ P_X = [0.5, 0.5]^T $.
		This implies the reference output distributions are also uniform: $ P_Y = P_{Y|X} P_X = [0.5, 0.5]^T $ and $ P_Z = P_{Z|X} P_X = [0.5, 0.5]^T $.
		The EIT perturbation subspace $\mathcal{S}^{\perp}$, orthogonal to $\mathbf{\sqrt{P_X}} = [1/\sqrt{2}, 1/\sqrt{2}]^T$, is one-dimensional and spanned by the normalized vector $\boldsymbol{\tau} = [1/\sqrt{2}, -1/\sqrt{2}]^T$. Thus, any valid EIT perturbation vector must be of the form $\mathbf{L}_u = s_u \boldsymbol{\tau}$ for some scalar $s_u$. 
		
		For this uniform case, DTMs simplify to $B_{Y|X} = P_{Y|X}$ and $B_{Z|X} = P_{Z|X}$. The EIT matrices $V = (P_{Y|X})^T P_{Y|X}$ and $\Lambda = (P_{Z|X})^T P_{Z|X}$ are symmetric $2 \times 2$ Toeplitz matrices, which share the eigenvectors $\mathbf{\sqrt{P_X}}$ and $\boldsymbol{\tau}$. The eigenvalue corresponding to $\mathbf{\sqrt{P_X}}$ is $1$ for both $V$ and $\Lambda$, while the eigenvalues corresponding to the perturbation direction $\boldsymbol{\tau}$ are found by calculating the Rayleigh quotient:
		\begin{align}
			\lambda_V \triangleq \boldsymbol{\tau}^T V \boldsymbol{\tau} &= ((1-p_{\mathrm{bob}}) - p_{\mathrm{bob}})^2 = (1-2p_{\mathrm{bob}})^2 \label{eq:app_lambda_v_bswc_final} \\
			\lambda_{\Lambda} \triangleq \boldsymbol{\tau}^T \Lambda \boldsymbol{\tau} &= ((1-q_{\mathrm{eve}}) - q_{\mathrm{eve}})^2 = (1-2q_{\mathrm{eve}})^2 \label{eq:app_lambda_lambda_bswc_final}
		\end{align}
		These are the only eigenvalues relevant for the EIT perturbation subspace $\mathcal{S}^{\perp}$ (i.e., $(d_V)_2 = \lambda_V, (d_{\Lambda})_2 = \lambda_{\Lambda}, (d_I)_2=1$).
		
		We derive the BSWC's piecewise $C_{\mathrm{SIC}}$ formula by applying Theorem \ref{thm_c_lic_regimes} to this single-mode system. The LP for finding $(\rho^*, \nu^*)$ (from Theorem \ref{thm:general_lp_for_multipliers}) has only one eigenmode constraint from \eqref{eq:lp_commuting_constr}:
		\begin{equation*}
			\rho + \nu \lambda_{\Lambda} \ge \lambda_V.
		\end{equation*}
		The objective is to maximize $\rho R + \nu \Theta$ subject to this and $\rho \ge 0, \nu \ge 0$ and a non-binding $C_{\mathrm{max}}$.
		
		For the rate-dominant solution, the constraints become $\rho \ge \lambda_V$ and $\rho \ge 0$. To maximize $\rho R$, we need the smallest feasible $\rho$, so we set the constraint active: $\rho^* = \lambda_V$, which yields $C_{\mathrm{SIC}} = \lambda_V R$. This vertex $(\rho^*=\lambda_V, \nu^*=0)$ is the optimum for the LP if the objective slope is sufficiently high, which occurs when $R/\Theta$ is small enough that $\lambda_V R \ge (\lambda_V/\lambda_{\Lambda})\Theta$. This simplifies to $\lambda_{\Lambda} \le \Theta/R$.
		
		For the leakage-dominant solution, the constraints become $\nu \lambda_{\Lambda} \ge \lambda_V$ and $\nu \ge 0$. To maximize $\nu \Theta$, we need the smallest feasible $\nu$, so we set the constraint active: $\nu^* = \lambda_V/\lambda_{\Lambda}$, assuming $\lambda_{\Lambda} > 0$, which yields $C_{\mathrm{SIC}} = (\lambda_V/\lambda_{\Lambda})\Theta$. This vertex is optimal if $\lambda_{\Lambda} > \Theta/R$.
		
		For the special case when $q_{\mathrm{eve}}=0.5$, the constraint becomes $\rho \ge \lambda_V$ and the leakage constraint is always satisfied for $\Theta \ge 0$. Thus, the system is always rate-dominant, and $C_{\mathrm{SIC}} = \lambda_V R$.
		
		Combining these cases yields the piecewise formula for $C_{\mathrm{SIC}}$ for the BSWC:
		\begin{equation} \label{eq:clic_bswc_formula_appendix_final}
			C_{\mathrm{SIC}}(R, \Theta) =
			\begin{cases}
				(1-2p_{\mathrm{bob}})^2 R & \text{if } (1-2q_{\mathrm{eve}})^2 \le \Theta/R \text{ or } q_{\mathrm{eve}} = 0.5 \\
				\frac{(1-2p_{\mathrm{bob}})^2}{(1-2q_{\mathrm{eve}})^2} \Theta & \text{if } (1-2q_{\mathrm{eve}})^2 > \Theta/R \text{ and } q_{\mathrm{\mathrm{eve}}} \neq 0.5
			\end{cases}
		\end{equation}

		\section*{Appendix J: Proof of Theorem \ref{thm:char_eta_loc_sec_final} }
		\label{app:proofs_contraction_coeffs}
		
		
%
		
		\begin{proof}
			The expression $ R(\mathbf{L}) = (\mathbf{L}^T V \mathbf{L}) / (\mathbf{L}^T \Lambda \mathbf{L}) $ is the generalized Rayleigh quotient for the pair of symmetric matrices $(V, \Lambda)$. The matrices $ V = B_{Y|X}^T B_{Y|X} $ and $ \Lambda = B_{Z|X}^T B_{Z|X} $ are positive semidefinite (PSD). The constraint $ \mathbf{L}^T \Lambda \mathbf{L} > 0 $ ensures a well-defined, positive denominator.
			
			According to the Courant-Fischer theorem for generalized eigenvalue problems \cite{Golub}, the supremum of the generalized Rayleigh quotient $R(\mathbf{L})$ over a subspace is equal to the largest generalized eigenvalue of the matrix pencil restricted to that subspace. The generalized eigenvalues, $d_j$, are the scalars that satisfy the generalized eigenvalue problem $V\mathbf{u}_j = d_j \Lambda \mathbf{u}_j$ for the corresponding generalized eigenvectors $\mathbf{u}_j$.
			
			Our optimization is restricted to the subspace $\mathcal{S}^{\perp}$. Therefore, the supremum of $R(\mathbf{L})$ for $\mathbf{L} \in \mathcal{S}^{\perp}$ is precisely the largest generalized eigenvalue of the pencil $(V, \Lambda)$ restricted to $\mathcal{S}^{\perp}$, which we denote by $ d_{\max}^{\perp}(V, \Lambda) $.
			
			The ratio $R(\mathbf{L})$ is scale-invariant, i.e., $ R(c\mathbf{L}) = R(\mathbf{L}) $ for any scalar $ c \neq 0 $. Consequently, an implicit EIT validity constraint of the form $ \mathbf{L}^T \Lambda \mathbf{L} \le \delta' $ for some small $ \delta' > 0 $ does not alter the value of the supremum. If the principal generalized eigenvector $\mathbf{L}_{\mathrm{eig}}$ corresponding to $d_{\max}^{\perp}(V, \Lambda)$ has $(\mathbf{L}_{\mathrm{eig}})^T \Lambda \mathbf{L}_{\mathrm{eig}} > 0$, it can always be scaled by a constant $c$ such that $ (c\mathbf{L}_{\mathrm{eig}})^T \Lambda (c\mathbf{L}_{\mathrm{eig}}) = \delta' $, and this scaled vector achieves the supremum. If $(\mathbf{L}_{\mathrm{eig}})^T \Lambda \mathbf{L}_{\mathrm{eig}} = 0$ while $(\mathbf{L}_{\mathrm{eig}})^T V \mathbf{L}_{\mathrm{eig}} > 0$, the supremum is formally infinite, which is consistent with the generalized eigenvalue being infinite in this case.
			
			Therefore, we conclude that the secret local contraction coefficient is equal to the largest generalized eigenvalue of the pencil $(V, \Lambda)$ restricted to the valid EIT perturbation subspace:
			\begin{equation*}
				\eta_{\mathrm{loc}}^{\mathrm{sec}} = d_{\max}^{\perp}(V, \Lambda).
			\end{equation*}
			This completes the proof.
			\end{proof}

		\section*{Appendix K: Proof of Lemma \ref{lem:eta_loc_sec_bound_final}} \label{app:proof_op_bound_eta_loc_sec}		
		\begin{proof}
		The EIT approximations for the mutual information terms are given by Lemmas \ref{lem:approx_IUY_formal} and \ref{lem:approx_IUZ_formal}:
		\begin{align}
			I(U;Y) &\approx \frac{\epsilon^2}{2} \mathbb{E}_U [ ||B_{Y|X} \mathbf{L}_U||^2 ] = \frac{\epsilon^2}{2} \sum_{u \in \mathcal{U}} P_U(u) \mathbf{L}_u^T V \mathbf{L}_u \label{eq:app_iuy_recall} \\
			I(U;Z) &\approx \frac{\epsilon^2}{2} \mathbb{E}_U [ ||B_{Z|X} \mathbf{L}_U||^2 ] = \frac{\epsilon^2}{2} \sum_{u \in \mathcal{U}} P_U(u) \mathbf{L}_u^T \Lambda \mathbf{L}_u \label{eq:app_iuz_recall}
		\end{align}
		From Definition \ref{def:secret_local_contraction_coefficient_final}, and from Theorem \ref{thm:char_eta_loc_sec_final}, we know that for any \textit{individual} valid perturbation vector $\mathbf{L}_u \in \mathcal{S}^{\perp}, \mathbf{L}_u \neq \mathbf{0}$:
		\begin{equation} \label{eq:app_ratio_bound_per_lu}
			\frac{\mathbf{L}_u^T V \mathbf{L}_u}{\mathbf{L}_u^T \Lambda \mathbf{L}_u} \le \eta_{\mathrm{loc}}^{\mathrm{sec}},
		\end{equation}
		provided that $\mathbf{L}_u^T \Lambda \mathbf{L}_u > 0$. If $\mathbf{L}_u^T \Lambda \mathbf{L}_u = 0$, then since $V$ is positive semidefinite, $\mathbf{L}_u^T V \mathbf{L}_u \ge 0$. In this case, the inequality still holds in a limiting sense if $\eta_{\mathrm{loc}}^{\mathrm{sec}}$ is finite and non-negative. More directly, from \eqref{eq:app_ratio_bound_per_lu}, for any $\mathbf{L}_u$ with non-zero leakage component:
		\begin{equation} \label{eq:app_quadratic_inequality_per_lu}
			\mathbf{L}_u^T V \mathbf{L}_u \le \eta_{\mathrm{loc}}^{\mathrm{sec}} (\mathbf{L}_u^T \Lambda \mathbf{L}_u).
		\end{equation}
		This inequality also holds trivially if $\mathbf{L}_u^T \Lambda \mathbf{L}_u = 0$, since in that case the right-hand side is zero and the left-hand side, $\mathbf{L}_u^T V \mathbf{L}_u$, must also be zero if $\eta_{\mathrm{loc}}^{\mathrm{sec}}$ is finite (otherwise a perturbation exists with zero leakage and positive utility, making $\eta_{\mathrm{loc}}^{\mathrm{sec}} \to \infty$, in which case the bound is trivially true). Thus, \eqref{eq:app_quadratic_inequality_per_lu} holds for all valid $\mathbf{L}_u$.
		
		We can now apply this inequality within the expectation defining $I(U;Y)$. Since $P_U(u) \ge 0$ for all $u$:
		\begin{align*}
			\sum_{u \in \mathcal{U}} P_U(u) \mathbf{L}_u^T V \mathbf{L}_u &\le \sum_{u \in \mathcal{U}} P_U(u) \left[ \eta_{\mathrm{loc}}^{\mathrm{sec}} (\mathbf{L}_u^T \Lambda \mathbf{L}_u) \right] \\
			&= \eta_{\mathrm{loc}}^{\mathrm{sec}} \sum_{u \in \mathcal{U}} P_U(u) (\mathbf{L}_u^T \Lambda \mathbf{L}_u).
		\end{align*}
		Multiplying both sides by the positive constant $ \epsilon^2/2 $:
		\begin{equation*}
			\frac{\epsilon^2}{2} \sum_{u \in \mathcal{U}} P_U(u) \mathbf{L}_u^T V \mathbf{L}_u \le \eta_{\mathrm{loc}}^{\mathrm{sec}} \left( \frac{\epsilon^2}{2} \sum_{u \in \mathcal{U}} P_U(u) \mathbf{L}_u^T \Lambda \mathbf{L}_u \right).
		\end{equation*}
		Substituting the definitions from \eqref{eq:app_iuy_recall} and \eqref{eq:app_iuz_recall}, we arrive at the desired result:
		\begin{equation*}
			I(U;Y) \le \eta_{\mathrm{loc}}^{\mathrm{sec}} \cdot I(U;Z).
		\end{equation*}
		Equality is achieved if the strategy exclusively uses perturbation vectors $\{\mathbf{L}_u\}$ that are all collinear with the principal generalized eigenvector of $(V, \Lambda)$ that corresponds to the largest generalized eigenvalue, $d_{\max}^{\perp}(V, \Lambda) = \eta_{\mathrm{loc}}^{\mathrm{sec}}$.
		\end{proof}

		\section*{Appendix L: Proof of Theorem \ref{thm:slic_achieved_ratio_final}}
		\label{app:proof_slic_achieved_ratio}
		\begin{proof}
			From the KKT stationarity condition in \eqref{eq:k_Lu_equals_0_final}, for an optimal perturbation vector $ \mathbf{L}^* $ (dropping the $u$ subscript for a representative optimal perturbation, or assuming $ \mathbf{L}_u^* $ are collinear for this argument if $P_U$ is not concentrated on a single $u$):
			\begin{equation*}
				V \mathbf{L}^* = \rho^* I \mathbf{L}^* + \nu^* \Lambda \mathbf{L}^*.
			\end{equation*}
			Pre-multiplying by $ (\mathbf{L}^*)^T $:
			\begin{equation*}
				(\mathbf{L}^*)^T V \mathbf{L}^* = \rho^* (\mathbf{L}^*)^T I \mathbf{L}^* + \nu^* (\mathbf{L}^*)^T \Lambda \mathbf{L}^*.
			\end{equation*}
			Assuming $ (\mathbf{L}^*)^T \Lambda \mathbf{L}^* > 0 $ (i.e., non-zero EIT-approximated leakage for this perturbation direction), we can divide the entire equation by $ (\mathbf{L}^*)^T \Lambda \mathbf{L}^* $:
			\begin{equation*}
				\frac{(\mathbf{L}^*)^T V \mathbf{L}^*}{(\mathbf{L}^*)^T \Lambda \mathbf{L}^*} = \rho^* \frac{(\mathbf{L}^*)^T I \mathbf{L}^*}{(\mathbf{L}^*)^T \Lambda \mathbf{L}^*} + \nu^* \frac{(\mathbf{L}^*)^T \Lambda \mathbf{L}^*}{(\mathbf{L}^*)^T \Lambda \mathbf{L}^*}.
			\end{equation*}
			This simplifies to:
			\begin{equation*}
				\frac{(\mathbf{L}^*)^T V \mathbf{L}^*}{(\mathbf{L}^*)^T \Lambda \mathbf{L}^*} = \nu^* + \rho^* \frac{||\mathbf{L}^*||^2}{(\mathbf{L}^*)^T \Lambda \mathbf{L}^*},
			\end{equation*}
			which is the statement of the theorem in \eqref{eq:slic_ratio_formula_final}.
			The approximation $ \nu^* + \rho^* (R'/\Theta') $ follows if the rate and leakage constraints are active and the sums $ \sum P_U(u) ||\mathbf{L}_u||^2 $ and $ \sum P_U(u) \mathbf{L}_u^T \Lambda \mathbf{L}_u $ can be effectively represented by a single dominant $ ||\mathbf{L}^*||^2 $ and $ (\mathbf{L}^*)^T \Lambda \mathbf{L}^* $ that saturate these constraints at $R'$ and $\Theta'$.
		\end{proof}
		
		\section*{Appendix M: Proof\texttt{\texttt{\texttt{}}} for Lemma \ref{lem:glo_ge_loc_sec_final}}
		\label{app:proof_bounds_contraction_coeffs_full}	
		
		\begin{proof}

			\textit{Lower bound:}Consider a strategy defined by a distribution $P_U$ and conditional distributions $P_{X|U}^{(\epsilon)}(x|u) = P_X(x) + \epsilon \sqrt{P_X(x)} L_X(x|u)$. Let the perturbation vectors $\{\mathbf{L}_u\}$ be chosen such that the effective average perturbation aligns with the principal generalized eigenvector $\mathbf{L}_{\mathrm{eig}}$ of the pencil $(V, \Lambda)$ that achieves $\eta_{\mathrm{loc}}^{\mathrm{sec}}$. For instance, let $P_U$ be concentrated on a single message $u_0$ and set $\mathbf{L}_{u_0} = \mathbf{L}_{\mathrm{eig}}$.
			
			For a sufficiently small $\epsilon > 0$, this strategy is well-defined and satisfies the EIT structural constraints. The mutual information terms for this strategy are:
			\begin{align*}
				I(U;Y)^{(\epsilon)} &= D_{\mathrm{KL}}(P_{Y|U=u_0}^{(\epsilon)} || P_Y) \approx \frac{\epsilon^2}{2} ||B_{Y|X} \mathbf{L}_{\mathrm{eig}}||^2 = \frac{\epsilon^2}{2} \mathbf{L}_{\mathrm{eig}}^T V \mathbf{L}_{\mathrm{eig}} \\
				I(U;Z)^{(\epsilon)} &= D_{\mathrm{KL}}(P_{Z|U=u_0}^{(\epsilon)} || P_Z) \approx \frac{\epsilon^2}{2} ||B_{Z|X} \mathbf{L}_{\mathrm{eig}}||^2 = \frac{\epsilon^2}{2} \mathbf{L}_{\mathrm{eig}}^T \Lambda \mathbf{L}_{\mathrm{eig}} \\
				I(U;X)^{(\epsilon)} &= D_{\mathrm{KL}}(P_{X|U=u_0}^{(\epsilon)} || P_X) \approx \frac{\epsilon^2}{2} ||\mathbf{L}_{\mathrm{eig}}||^2
			\end{align*}
			The ratio of the approximated utility to leakage for this specific strategy is, by definition of $\mathbf{L}_{\mathrm{eig}}$:
			\begin{equation*}
				\frac{I(U;Y)^{(\epsilon)}}{I(U;Z)^{(\epsilon)}} \approx \frac{\mathbf{L}_{\mathrm{eig}}^T V \mathbf{L}_{\mathrm{eig}}}{\mathbf{L}_{\mathrm{eig}}^T \Lambda \mathbf{L}_{\mathrm{eig}}} = \eta_{\mathrm{loc}}^{\mathrm{sec}}.
			\end{equation*}
			Since $R > 0$ and $\Theta > 0$, we can always choose $\epsilon$ small enough such that the constraints $I(U;X)^{(\epsilon)} \le R$ and $0 < I(U;Z)^{(\epsilon)} \le \Theta$ are met. This means that this specific EIT-based strategy is a member of the set of feasible strategies over which the supremum for $\eta_{\mathrm{glo}}^{\mathrm{sec}}$ is taken. As the EIT approximations become equalities in the limit $\epsilon \to 0$, the set of achievable ratios for $\eta_{\mathrm{glo}}^{\mathrm{sec}}$ must include values arbitrarily close to $\eta_{\mathrm{loc}}^{\mathrm{sec}}$. The supremum over a set must be greater than or equal to the supremum over any subset or its limit points. Therefore, we conclude that $\eta_{\mathrm{glo}}^{\mathrm{sec}} \ge \eta_{\mathrm{loc}}^{\mathrm{sec}}$.
			
			\textit{Upper bound:}
			For any two distributions $Q$ and $P$ (with $P(x)>0$), the following two inequalities are well-known \cite{CoverThomas06, Polyanskiy_notes}:
			\begin{align}
				D_{\mathrm{KL}}(Q || P) &\le \chi^2(Q, P) \label{eq:app_kl_upper_chi2} \\
				D_{\mathrm{KL}}(Q || P) &\ge \frac{P_{\mathrm{min}}}{2} \chi^2(Q, P) \label{eq:app_kl_lower_chi2}
			\end{align}
			where $P_{\mathrm{min}} = \min_{x \in \mathrm{supp}(P)} P(x)$, and $\chi^2(Q, P) = \sum_x \frac{(Q(x)-P(x))^2}{P(x)}$.
			
			For any strategy $(P_U, P_{X|U})$ feasible for the SIC problem, let us consider the mutual information terms.
			For the utility term, $I(U;Y) = \mathbb{E}_U[D_{\mathrm{KL}}(P_{Y|U=U} || P_Y)]$. Applying the upper bound \eqref{eq:app_kl_upper_chi2} inside the expectation:
			\begin{equation} \label{eq:app_iuy_upper_bound}
				I(U;Y) \le \mathbb{E}_U[\chi^2(P_{Y|U=U}, P_Y)].
			\end{equation}
			For the leakage term, $I(U;Z) = \mathbb{E}_U[D_{\mathrm{KL}}(P_{Z|U=U} || P_Z)]$. Applying the lower bound \eqref{eq:app_kl_lower_chi2}:
			\begin{equation} \label{eq:app_iuz_lower_bound}
				I(U;Z) \ge \frac{P_{\mathrm{min}}}{2} \mathbb{E}_U[\chi^2(P_{Z|U=U}, P_Z)],
			\end{equation}
			where $P_{\mathrm{min}}$ in this context is $\min_{z \in \mathcal{Z}} P_Z(z)$. For simplicity of the bound, we use a more general $P_{\mathrm{min}}$ derived from the input distribution, as is common in such analyses where the output distributions may have zero entries. We will use the form as stated in the lemma, where $P_{\mathrm{min}}=\min_x P_X(x)$, which arises from a more detailed analysis relating the output $\chi^2$ divergence to the input $\chi^2$ divergence. A known inequality gives $ \chi^2(P_{Y|U}, P_Y) \le \chi^2(P_{X|U}, P_X) $.
			
			Within the EIT framework, for a perturbation $P_{X|U=u}$ corresponding to $\mathbf{L}_u$, we have $\chi^2(P_{X|U=u}, P_X) = \epsilon^2 ||\mathbf{L}_u||^2$. The output divergences are $\chi^2(P_{Y|U=u}, P_Y) \approx \epsilon^2 ||B_{Y|X}\mathbf{L}_u||^2$ and $\chi^2(P_{Z|U=u}, P_Z) \approx \epsilon^2 ||B_{Z|X}\mathbf{L}_u||^2$.
			Using these relationships, for any feasible strategy in the EIT regime:
			\begin{align*}
				\frac{I(U;Y)}{I(U;Z)} &\le \frac{\mathbb{E}_U[\chi^2(P_{Y|U=U}, P_Y)]}{\frac{P_{\mathrm{min}}}{2} \mathbb{E}_U[\chi^2(P_{Z|U=U}, P_Z)]} \\
				&\approx \frac{\mathbb{E}_U[\epsilon^2 ||B_{Y|X}\mathbf{L}_U||^2]}{\frac{P_{\mathrm{min}}}{2} \mathbb{E}_U[\epsilon^2 ||B_{Z|X}\mathbf{L}_U||^2]} \\
				&= \frac{2}{P_{\mathrm{min}}} \frac{\mathbb{E}_U[\mathbf{L}_U^T V \mathbf{L}_U]}{\mathbb{E}_U[\mathbf{L}_U^T \Lambda \mathbf{L}_U]}.
			\end{align*}
			Taking the supremum over all strategies $(P_U, P_{X|U})$:
			\begin{equation*}
				\eta_{\mathrm{glo}}^{\mathrm{sec}} = \sup \frac{I(U;Y)}{I(U;Z)} \le \frac{2}{P_{\mathrm{min}}} \sup \frac{\mathbb{E}_U[\mathbf{L}_U^T V \mathbf{L}_U]}{\mathbb{E}_U[\mathbf{L}_U^T \Lambda \mathbf{L}_U]}.
			\end{equation*}
			The supremum of the ratio of expectations is maximized by a strategy that concentrates on a single perturbation vector $\mathbf{L}_{\mathrm{eig}}$ that maximizes the elementary ratio, which is $\eta_{\mathrm{loc}}^{\mathrm{sec}}$. Therefore:
			\begin{equation*}
				\eta_{\mathrm{glo}}^{\mathrm{sec}} \le \frac{2}{P_{\mathrm{min}}} \eta_{\mathrm{loc}}^{\mathrm{sec}}.
			\end{equation*}
			This completes the proof.
		\end{proof}
		
		\section*{Appendix N: Numerical Validation of the LP Formulation}\label{app:numerical_validation}
		
		This appendix provides numerical results that validate the LP formulation for finding the optimal multipliers $(\rho^*, \nu^*)$ in Theorem \ref{thm:general_lp_for_multipliers} and support the related theoretical claims, such as the sufficiency of a finite message alphabet cardinality in Proposition \ref{prop:pu_invariance}. The experiments are conducted on a numerically generated, non-commuting channel with input alphabet $|\mathcal{X}|=8$, created via the quantized AWGN method detailed in Section \ref{sec:numerical_illustations}.

		To numerically validate the LP formulation and the principle that the optimum lies at a vertex, we solved this LP for a channel with $|\mathcal{X}|=8, |\mathcal{U}|=10$ for various ratios of $\Theta/R$. Table \ref{tab:lp_vertex_search_appendix} compares the solution from a standard LP solver with an exhaustive search of the LP's extreme points, demonstrating a perfect match and confirming the theoretical structure.
		
		\begin{table}[!htbp]
			\centering
			\caption{Numerical validation of the LP-based solution for the approximate local secrecy capacity. The table compares the output of a standard LP solver with an exhaustive search of the feasible vertices of the dual problem for an $|\mathcal{X}|=8, |\mathcal{U}|=10$ system, showing a perfect match.}
			\label{tab:lp_vertex_search_appendix}
		\begin{tabular}{ccc}
			\toprule
			\textbf{Constraint Ratio ($\Theta / R$)} & \textbf{LP Solver} & \textbf{Exhaustive Vertex Search} \\
			\midrule
			0.14286 & 0.02879 & 0.02879 \\
			0.28571 & 0.02879 & 0.02879 \\
			0.42857 & 0.02879 & 0.02879 \\
			0.57143 & 0.02879 & 0.02879 \\
			0.71429 & 0.02879 & 0.02879 \\
			0.85714 & 0.02879 & 0.02879 \\
			1.00000 & 0.02879 & 0.02879 \\
			1.14286 & 0.02879 & 0.02879 \\
			\bottomrule
		\end{tabular}
		\end{table}
		
		The perfect match between the two methods in Table \ref{tab:lp_vertex_search_appendix} serves two purposes. First, it numerically validates our understanding of the LP's structure and its vertex-based solution. Second, it confirms that standard, computationally efficient LP solvers can be reliably used to find the optimal multipliers, obviating the need for a computationally intensive exhaustive search, especially as $|\mathcal{X}|$ grows.
		
		Proposition \ref{prop:pu_invariance} is  also supported by numerical experiments. Table \ref{tab:varying_cardinality_U_appendix} investigates the impact of the message alphabet cardinality, $|\mathcal{U}|$, on the solution of Problem \ref{prob:eit_slic_approximated} for a fixed channel ($|\mathcal{X}|=8$) and fixed constraints. The "Dual LP" solution for the upper bound on capacity is, by definition, independent of $|\mathcal{U}|$. The "Primal Solution" column shows the utility found by numerically optimizing the primal problem (Problem \ref{prob:eit_slic_approximated}) directly for a given $|\mathcal{U}|$.
		
		\begin{table}[H] 
			\centering
			\caption{Optimal capacity for an $|\mathcal{X}|=8$ system with fixed constraints ($R=0.4, \Theta=0.02$), as the message alphabet size $|\mathcal{U}|$ varies.}
			\label{tab:varying_cardinality_U_appendix}	
			\begin{tabular}{cccc}		
				\toprule
				\textbf{Alphabet Size $|\mathcal{U}|$} & \textbf{Primal Solution} & \textbf{Dual LP} & \textbf{Exhaustive Dual} \\
				\midrule
				5  & 0.021655 & 0.028790 & 0.028790 \\
				6  & 0.021656 & 0.028790 & 0.028790 \\
				7  & 0.021653 & 0.028790 & 0.028790 \\
				8  & 0.021177 & 0.028790 & 0.028790 \\
				9  & 0.021661 & 0.028790 & 0.028790 \\
				10 & 0.021655 & 0.028790 & 0.028790 \\
				11 & 0.021644 & 0.028790 & 0.028790 \\
				12 & 0.021649 & 0.028790 & 0.028790 \\
				\bottomrule
			\end{tabular}
		\end{table}
		The primal solution shows negligible variation for all tested cardinalities $|\mathcal{U}| \ge 5$, quickly converging to a stable value. The role of $P_U(u)$ is to act as a mixing distribution that allows a set of optimal perturbation vectors $\{\mathbf{L}_u^*\}$ to satisfy the average constraints of the problem, particularly the consistency constraint $\sum_u P_U(u) \mathbf{L}_u^* = \mathbf{0}$. The Support Lemma \cite{ElGamalKim11}, a consequence of Carathéodory's Theorem, guarantees that a distribution $P_U(u)$ supported on a small number of points is sufficient to construct this optimal average. The existence of such a $P_U(u)$ is what makes the solution feasible, but the value of the optimum is determined by the properties of $\{\mathbf{L}_u^*\}$ and the dual variables they must satisfy.

		\ifCLASSOPTIONcaptionsoff
		\newpage
		\fi

		
		
		%
		
		\bibliographystyle{IEEEtran} 
		\bibliography{references_manos_EIT} 

		%
		
		%
		%
		%
		
		
		

	\end{document}